%% file: topolo.tex
\documentclass[12pt]{article}
\usepackage[OT1]{fontenc}
\usepackage[utf8]{inputenc}
\usepackage{lscape}
\usepackage{mathptmx}
\usepackage{mathtools}
\usepackage{enumerate}
\usepackage[ruled]{algorithm2e}
\usepackage{helvet}
\usepackage[activate={true,nocompatibility},final,tracking=true,kerning=true,spacing=true,factor=1100,stretch=10,shrink=10,spacing=nonfrench]{microtype}
\usepackage{changepage}
\usepackage{times}
\usepackage{setspace}
\usepackage{lineno}
\usepackage[hyphens]{url}
\usepackage[hypertexnames=false]{hyperref}
\hypersetup{breaklinks=true}
\newcommand*\patchAmsMathEnvironmentForLineno[1]{%
  \expandafter\let\csname old#1\expandafter\endcsname\csname #1\endcsname
  \expandafter\let\csname oldend#1\expandafter\endcsname\csname end#1\endcsname
  \renewenvironment{#1}%
     {\linenomath\csname old#1\endcsname}%
     {\csname oldend#1\endcsname\endlinenomath}}%
\newcommand*\patchBothAmsMathEnvironmentsForLineno[1]{%
  \patchAmsMathEnvironmentForLineno{#1}%
  \patchAmsMathEnvironmentForLineno{#1*}}%
\AtBeginDocument{%
\patchBothAmsMathEnvironmentsForLineno{equation}%
\patchBothAmsMathEnvironmentsForLineno{align}%
\patchBothAmsMathEnvironmentsForLineno{flalign}%
\patchBothAmsMathEnvironmentsForLineno{alignat}%
\patchBothAmsMathEnvironmentsForLineno{gather}%
\patchBothAmsMathEnvironmentsForLineno{multline}%
}
\usepackage{wrapfig}
\usepackage{xcolor}
\usepackage{subdepth}
\usepackage{mdframed}
\everymath{\displaystyle}
\RequirePackage{amsthm,amsmath,amsbsy,amssymb,array,bm,color,graphicx,epstopdf,mathtools,calc,xfrac,mdframed,tikz,pgf,pgffor,enumerate,enumitem,textpos,booktabs,setspace}
\usepackage{subfig}
\DeclareCaptionFont{figtextfont}{\normalfont}
\SetKwInput{KwInput}{Input}
\SetKwInput{KwOutput}{Output}
\usepgfmodule{shapes,plot}
\usetikzlibrary{shapes.geometric,decorations,arrows,positioning,calc,decorations.markings,fadings,snakes,arrows.meta}

\xdefinecolor{darkgreen}{RGB}{161,217,155}
\xdefinecolor{tree}{RGB}{253,192,134}
\xdefinecolor{cycle}{RGB}{190,174,212}

\xdefinecolor{triadclose}{RGB}{141,211,199}

\xdefinecolor{tri_walk_main}{RGB}{51,160,44}
\xdefinecolor{tri_walk_sub}{RGB}{178,223,138}

\xdefinecolor{hexa_walk_main}{RGB}{31,120,180}
\xdefinecolor{hexa_walk_sub}{RGB}{166,206,227}

\xdefinecolor{octo_walk_main}{RGB}{227,26,28}
\xdefinecolor{octo_walk_sub}{RGB}{251,154,153}

\tikzstyle{pop}=[circle, draw, fill=black,
        inner sep=0pt, minimum width=4.5pt]
\input{Figures/Walks2.tex}
\let\originalleft\left
\let\originalright\right
\renewcommand{\left}{\mathopen{}\mathclose\bgroup\originalleft}
\renewcommand{\right}{\aftergroup\egroup\originalright}

\DeclareMathOperator*{\argmax}{argmax}

\allowdisplaybreaks
\makeatletter
\newcommand{\pushright}[1]{\ifmeasuring@#1\else\omit\hfill$\displaystyle#1$\fi\ignorespaces}
\newcommand{\pushleft}[1]{\ifmeasuring@#1\else\omit$\displaystyle#1$\hfill\fi\ignorespaces}
\makeatother
\makeatletter
\def\BState{\State\hskip-\ALG@thistlm}
\makeatother

\newcommand{\Tr}[1]{\ensuremath{{\textstyle\operatorname{Tr}\left(#1\right)}}}

\newcommand{\E}{\ensuremath{\operatorname{\mathbb{E}}}}

\DeclareFontFamily{U}{mathx}{\hyphenchar\font45}
\DeclareFontShape{U}{mathx}{m}{n}{<-> mathx10}{}
\DeclareSymbolFont{mathx}{U}{mathx}{m}{n}
\theoremstyle{plain}
\newtheorem{Theorem}{Theorem}
\newtheorem{Lemma}{Lemma}
\newtheorem{Proposition}{Proposition}
\newtheorem{Corollary}{Corollary}
\newtheorem{Definition}{Definition}
\newtheorem{Assumption}{Assumption}
\theoremstyle{remark}
\theoremstyle{definition}
\newtheorem{Remark}{Remark}


\usepackage[super,comma,sort&compress]{natbib}
\usepackage{bibunits}

\newcommand{\Texttwoedge}{%
\raisebox{0.55mm}{\!
\tikz[scale=.2,clip]{\draw[thick]
(0,0) node[circle,fill=black,inner sep=0pt, minimum width=2.5 pt]{}--(0.5,0) node[circle,fill=black,inner sep=0pt, minimum width=2.5 pt]{}--(1,0) node[circle,fill=black,inner sep=0pt, minimum width=2.5 pt]{};
}}
}
\newcommand{\Textpathtwo}{%
\raisebox{0.0mm}{\!
\tikz[scale=.2,clip]{\draw[thick]
(210:.6) node[circle,fill=black,inner sep=0pt, minimum width=3 pt]{}--
(90:.6) node[circle,fill=black,inner sep=0pt, minimum width=3 pt]{}--
(-30:.6) node[circle,fill=black,inner sep=0pt, minimum width=3 pt]{};
}}
}
\newcommand{\Textpathtwob}{%
\raisebox{0.0mm}{\!
\tikz[scale=.2,clip]{\draw[thick]
(210:.6) node[circle,fill=black,inner sep=0pt, minimum width=3 pt]{}--
(90:.6) node[circle,fill=black,inner sep=0pt, minimum width=3 pt]{}
(-30:.6) node[circle,fill=black,inner sep=0pt, minimum width=3 pt]{};
}}
}

\newcommand{\Texttrianglenew}{%
\raisebox{0.0mm}{\!
\tikz[scale=.2,clip]{\draw[thick]
(210:.6) node[circle,fill=black,inner sep=0pt, minimum width=3 pt]{}--
(90:.6) node[circle,fill=black,inner sep=0pt, minimum width=3 pt]{}--
(-30:.6) node[circle,fill=black,inner sep=0pt, minimum width=3 pt]{}--
(210:.6) node[circle,fill=black,inner sep=0pt, minimum width=3 pt]{};
}}
}

\newcommand{\Textoctagon}{%
\smash{\raisebox{.5mm}{\!
\tikz[clip]{\hspace{-.4cm}\octagon{4.3}{.03}{2}{black}}\hspace{-.25cm}}
}}
\newcommand{\Texttadpolea}{%
\raisebox{-.3mm}{
\!\Textoctagon\hspace{-0.2cm}\raisebox{.6mm}{\Texttwoedge\hspace{-0.15cm}}}
}

\DeclareSymbolFont{symbolsC}{U}{txsyc}{m}{n}
\SetSymbolFont{symbolsC}{bold}{U}{txsyc}{bx}{n}
\DeclareFontSubstitution{U}{txsyc}{m}{n}
\DeclareMathSymbol{\multimapboth}{\mathrel}{symbolsC}{"13}
\DeclareCaptionLabelSeparator{pipe}{ \textbf{\textbar} }
\usepackage[labelfont=bf,labelsep=pipe]{caption}
\title{\vspace{-3\baselineskip}Topology reveals universal features for network comparison}%
\author{Pierre-Andr\'e G.\ Maugis, Sofia C.\ Olhede \& Patrick J.\ Wolfe}%
\date{University College London}%
\begin{document}
\maketitle

\begin{abstract}
The topology of any complex system is key to understanding its structure and function. Fundamentally, algebraic topology guarantees that any system represented by a network can be understood through its closed paths. The length of each path provides a notion of scale, which is vitally important in characterizing dominant modes of system behavior. Here, by combining topology with scale, we prove the existence of universal features which reveal the dominant scales of any network. We use these features to compare several canonical network types in the context of a social media discussion which evolves through the sharing of rumors, leaks and other news. Our analysis enables for the first time a universal understanding of the balance between loops and tree-like structure across network scales, and an assessment of how this balance interacts with the spreading of information online. Crucially, our results allow networks to be quantified and compared in a purely model-free way that is theoretically sound, fully automated, and inherently scalable. 
\end{abstract}
\begin{bibunit}[nature]

\begin{figure}[t]
    \begin{center}
    \vspace{-\baselineskip}%
	   {\includegraphics[width=.67\textwidth]{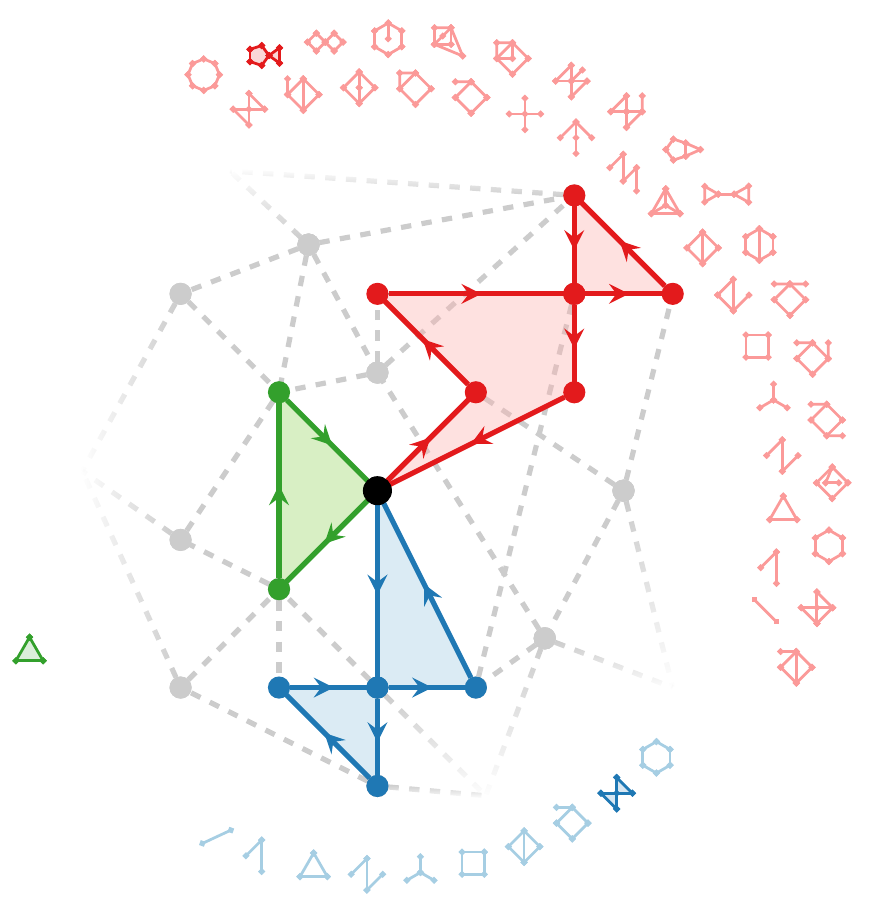}\vspace{-\baselineskip}}%
	\end{center}
	\caption{\label{fig:walks}\textbf{Closed walks at all scales in a network determine its topology.} A closed walk at scale $k$ traces out a closed path, starting at a given node (here, the center node, shown in black) and taking $k$ steps to return to its origin. Highlighted are walks at scales $k = 3$ (green), 6 (blue), and 8 (red), along with all possible closed paths at each scale. Formally, all walks across all scales determine the fundamental group of any connected network, and hence its classification as a topological space.
	}
\end{figure}

Across the sciences, complex physical and biological systems are represented by networks\cite{gao2016universal}. A fundamental challenge is to understand the structure of such networks, and to compare them irrespective of their sizes and origins\cite{Ingalhalikar14, Shashwath14, Hudson14,schieber2017quantification}. Closed paths in a network (Fig.~\ref{fig:walks}) are crucial to this understanding.  They determine the topology of any network through its mathematical symmetries\cite{stallings1983topology}, and hence its behavior as a dynamical system\cite{kotani2000zeta}. The shapes of these paths capture the full range of scales intrinsic to any network: from local features reflecting small-scale properties at the level of individual nodes, to global features revealing large-scale aspects of system behavior such as  diffusion and information flows\cite{brockmann2013hidden}.

Shapes that are over-represented, termed motifs, play a key functional role in networks\cite{milo2002network, Maayan2008cyclic,angulo2015network, sorrells2015making,benson2016higher}. They are equally fundamental to mathematical representations: In the theory of large graph limits which has emerged over the past decade, motif densities correspond directly to moments of probability distributions\cite{janson1995contiguity, diaconis2008graph, bollobas2009metric, BickelLevina2012, lovasz2012large}. However, the question of precisely which shapes are essential to a unified understanding of all networks has long remained open\cite{Milo2004super, alon2007network, gerstein2012architecture, boyle2014comparative}. Here we show that the simplest shapes are uniquely essential: They reveal the dominant scales of any network. This discovery allows us to identify structural differences between networks in an entirely model-free way, and to pinpoint exactly the scales at which these differences occur.

Figure~\ref{fig:walks} describes how walks traveling from node to node in a network determine its classification as a topological space. A closed walk has two essential characteristics: its scale, which is the number of steps it takes before returning to its starting node (Fig.~\ref{fig:walks}, center); and its shape, which describes the closed path it traces out (Fig.~\ref{fig:walks}, periphery). By combining topology with scale, we prove a fundamental result: Walks with the simplest shapes will predominate.  This holds universally across a vast range of network types\cite{holland1983stochastic, hoff2002latent, chung2003spectra, bollobas2007phase, airoldi2008mixed, bickel2009nonparametric, riordan2011explosive, zhao2012consistency, olhede2013network}. As a consequence, we can understand and articulate structural differences between networks obtained under different experimental conditions or at different times.

Walks in a network govern not only its topology, but also its spectrum.  This is crucial when networks represent complex physical systems: Spectral properties determine how phenomena diffuse and spread over networks, with walks describing propagation from node to node\cite{morone2015influence}.  We prove that dominating walks exhibit a sharp phase transition, changing from cycles to trees as networks becomes sparse. By contrast, we show that non-backtracking walks\cite{angel2015non}---which never traverse the same network edge twice in succession, and thus form geodesics analogous to those on a Riemann surface---are impervious to this phase transition. This is a fundamental robustness result, which has broad implications for our understanding of network dynamics\cite{luscombe2004genomic, barzel2013universality} and control\cite{liu2011controllability, ruths2014control, yan2015spectrum}. It reveals why techniques based on the Perron--Frobenius operator, represented by the non-backtracking matrix of a network\cite{fitzner2013backtracking}, have revolutionized algorithms for network community detection\cite{ahn2010link, Krzakala13, Newman14}.

\subsection*{Topology reveals dominant network features}

Surprisingly, there is a natural progression to the shapes traced out by walks in a network. Figure~\ref{fig:walks} shows how these shapes grow in complexity as scale increases. Why then are the simplest shapes guaranteed to dominate all others? There are two crucial reasons: one stemming from topology, the other from scale. 

First, trees and cycles have the simplest topologies of all connected networks. Their Euler characteristics as one-dimensional simplicial complexes---the differences between their numbers of nodes and edges---are as large as possible. Viewing a connected network as a topological surface, one minus its Euler characteristic counts its one-dimensional holes (its first Betti number; Fig.~\ref{fig:extension}). Starting from any tree that spans the entire network, each hole is formed by one network edge outside this spanning tree, which completes a distinct cycle within the network.

Second, because of their fixed Euler characteristics, trees and cycles naturally organize themselves by scale. Cycles maximize the number of nodes that can be visited by a closed walk at any given scale, whereas trees minimize the number of distinct edges that must be traversed. Scale then combines with topology to ensure that trees and cycles predominate.

To see how topology reveals dominant network features, we begin with the simplest setting. Let $G$ be a generalized random graph\cite{riordan2011explosive}, with $\operatorname{\mathbb{P}} \{\textnormal{edge in $G$}\}$ the probability that two randomly chosen nodes are connected, and call $\E \#\{\textnormal{copies of $w$ in $G$}\}$ the average number of walks in $G$ tracing out the same shape as a walk $w$.

Then, as the number of nodes in $G$ becomes large, the following simple and fundamental equation governs $\E \#\{\textnormal{copies of $w$ in $G$}\}$:
\begin{align}
\nonumber \!\!\log \E\#\{\textnormal{copies of $w$ in $G$}\} \sim & \textnormal{ const}_w 
\\ \nonumber & + \#\{\textnormal{nodes visited by $w$}\} \cdot \log \#\{\textnormal{nodes in $G$}\}
\\ \label{eq:nb-of-copies} & - \#\{\textnormal{edges traversed by $w$}\} \cdot \left| \log \operatorname{\mathbb{P}} \{\textnormal{edge in $G$}\} \right|.
\end{align}

\begin{figure}[t]
\begin{center}
{\vspace{-1\baselineskip}%
\includegraphics[width=.5\textwidth,angle=90]{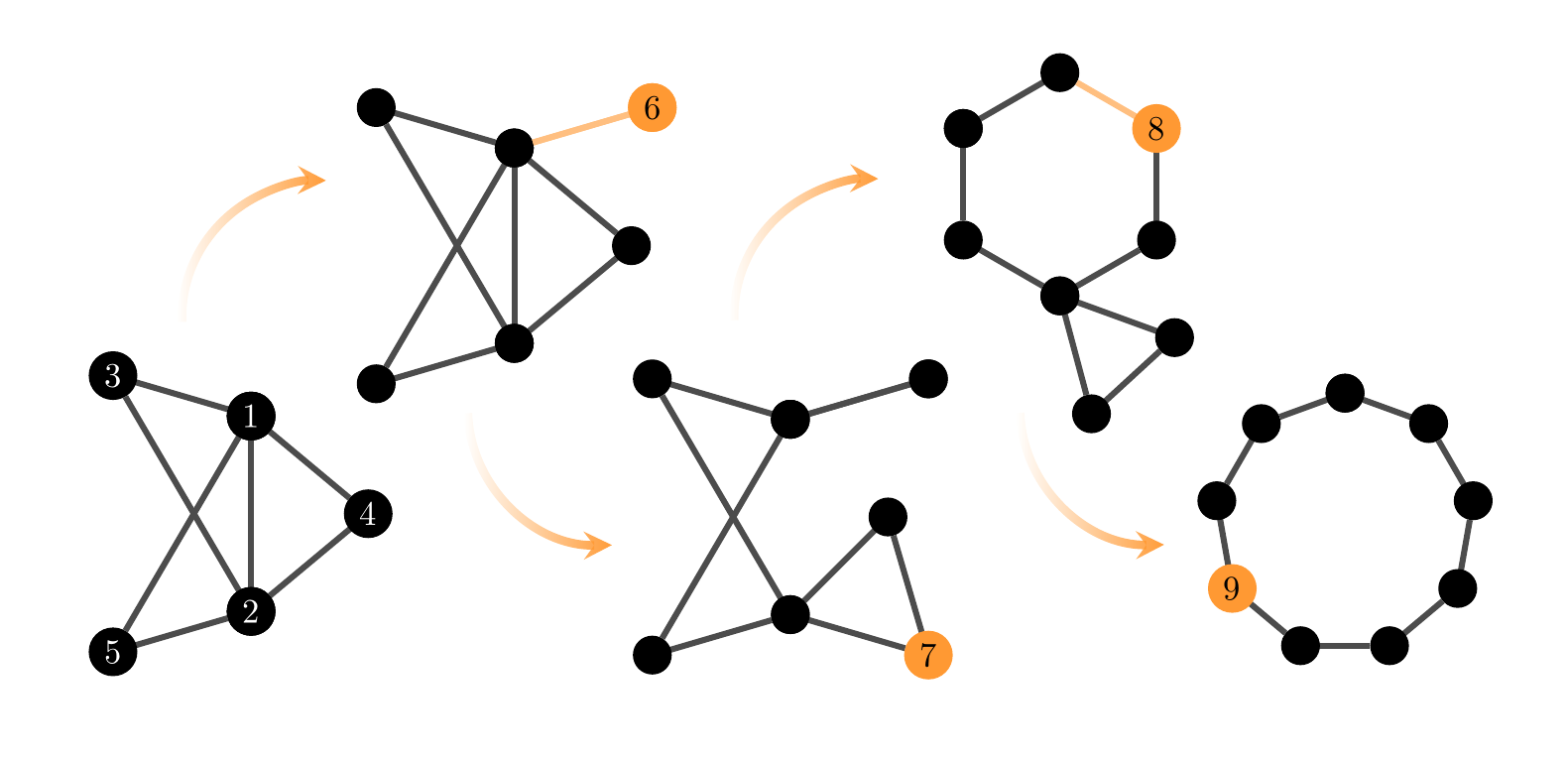}}%
\vspace{-1.75\baselineskip}%
\end{center}
\caption{\label{fig:extension}\textbf{Repeatedly extending a walk simplifies its topology.} Any closed walk extends to either a tree or a cycle. Here the walk 1-2-1-3-2-4-1-2-5-1 extends fully in four steps: Each extension adds one node to the shape traced out (1\underline{6}13241251, 161324\underline{7}251, 16\underline{8}3247251, 1683247\underline{9}51) and can only reduce its first Betti number.}
\end{figure}

Equation~\eqref{eq:nb-of-copies} shows how $\E \#\{\textnormal{copies of $w$ in $G$}\}$ depends on the shape traced out by $w$, revealing a fundamental trade-off between its nodes and edges. This determines which walks dominate at any given scale: As long as $\operatorname{\mathbb{P}} \{\textnormal{edge in $G$}\} > 1 / \#\{\textnormal{nodes in $G$}\} $, a walk that traverses an edge three times can never be dominant. This is because the right-hand side of equation~\eqref{eq:nb-of-copies} could then be increased by visiting a new node instead---even at the cost of traversing an additional edge.

Figure~\ref{fig:extension} shows how extending a walk in this manner simplifies its topology. Specifically, call $\smash{ w' }$ an extension of $w$ if both walks are at the same scale, but the sequences of nodes they visit differ at exactly one entry, with $\smash{ w' }$ visiting one additional node and traversing at most one additional edge. The key topological property of an extension is that its shape cannot have a more negative Euler characteristic than that traced out by $w$. Consequently, repeated extensions lead to cycles or trees (Supplementary Information, section~\ref{sec:extensions}). Walks tracing out these shapes are fully extended: either they traverse all their edges once, or all their edges twice.

\begin{figure}[t]
	\vspace{-2\baselineskip}\hskip2cm{\includegraphics[width=.67\textwidth]{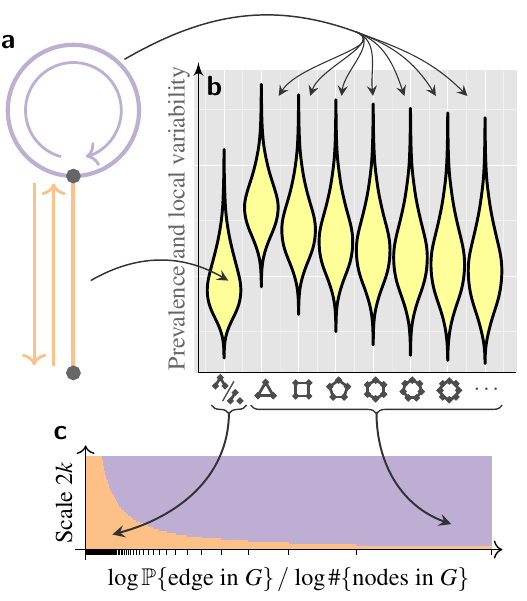}}%
	\caption{\label{fig:tadpole}\textbf{Combining topology with scale enables network comparison.} \protect\textbf{a}, Extending a tadpole graph \,\smash{\protect\Texttadpolea} leads to dominant walks: those traversing their edges either once (cycles, purple) or twice (trees, orange). \protect\textbf{b}, Consequently, measuring the prevalence and local variability of trees and cycles across any network reveals its dominant scales. \protect\textbf{c}, A phase transition (Corollary~\protect\ref{corrol_main}) determines whether cycles or trees dominate the structure of a graph $G$ at a given scale, based on the number of nodes in $G$ and the probability that two randomly chosen nodes are connected.}
\end{figure}

\subsection*{Universal features for network comparison}

Figure~\ref{fig:tadpole} shows how combining topology with scale yields universal features for network comparison. To prove that walks tracing out trees and cycles dominate, we recast equation~\eqref{eq:nb-of-copies} to hold for any network generating mechanism $\mathcal{M}$:
\begin{align}
\nonumber \!\!\log \E\#\{\textnormal{copies of $w$ in $G$}\} \sim & \textnormal{ const}_{w,\mathcal{M}} 
\\ \nonumber & + \#\{\textnormal{nodes visited by $w$}\} \cdot \log \#\{\textnormal{nodes in $G$}\}
\\ \label{eq:nb-of-copies2} & - \left| \log \operatorname{\mathbb{P}}_{\!\!\!\mathcal{M}} \{\textnormal{local copy of $w$ in $G$}\} \right|.
\end{align}

Equation~\eqref{eq:nb-of-copies2} replaces $\#\{\textnormal{edges traversed by $w$}\} \cdot \left| \log \operatorname{\mathbb{P}} \{\textnormal{edge in $G$}\} \right|$ in equation~\eqref{eq:nb-of-copies}, which drove our earlier argument for walk dominance, with a more general quantity $\smash{ \left| \log \operatorname{\mathbb{P}}_{\!\!\!\mathcal{M}} \{\textnormal{local copy of $w$ in $G$}\} \right| }$, which captures local structure in $G$.

To define $ \smash{ \operatorname{\mathbb{P}}_{\!\!\!\mathcal{M}} \{\textnormal{local copy of $w$ in $G$}\} }$, let $G[u]$ be the graph $G$ restricted to any set of its nodes $u$, and $K[u]$ the graph on $u$ with all edges present. Consider a walk $w$ and a graph $G$ generated via mechanism $\mathcal{M}$. If we now let $u$ be chosen at random from all subsets of nodes of $G$ of any fixed size, we may define
\begin{equation*}
\operatorname{\mathbb{P}}_{\!\!\!\mathcal{M}} \{\textnormal{local copy of $w$ in $G$}\} =
\frac{\E\#\{\textnormal{copies of $w$ in $G[u]$}\}}{\#\{\textnormal{copies of $w$ in $K[u]$}\}}.
\end{equation*}
Here $\#\{\textnormal{copies of $w$ in $K[u]$}\}$ is the number of walks on the nodes of $u$ tracing out the same shape as $w$, while $\E\#\{\textnormal{copies of $w$ in $G[u]$}\}$ is the number of such walks present in a random subgraph $G[u]$ on average. 

Thus, $ \smash{ \operatorname{\mathbb{P}}_{\!\!\!\mathcal{M}} \{\textnormal{local copy of $w$ in $G$}\} }$ is the probability that a randomly selected walk local to $u$ is present in $G[u]$. The behavior of $ \smash{ \operatorname{\mathbb{P}}_{\!\!\!\mathcal{M}} \{\textnormal{local copy of $w$ in $G$}\} }$ when $w$ is extended to visit a single additional node drives our main result.

\begin{Theorem}\label{thm:main}
Suppose for every closed walk $w$ at scale $k$ which is not fully extended, there exists an extension $\smash{ w' }$ such that under network generating mechanism $\mathcal{M}$,
\begin{equation}\label{assumption-main}
\#\{\textnormal{nodes in $G$}\} \cdot
	\frac{\operatorname{\mathbb{P}}_{\!\!\!\mathcal{M}} \{\textnormal{local copy of $\smash{ w' }$ in $G$}\}}
		{\operatorname{\mathbb{P}}_{\!\!\!\mathcal{M}} \{\textnormal{local copy of $w$ in $G$}\}}
		\to\infty.
\end{equation}
Then in networks generated by $\mathcal{M}$, closed walks tracing out cycles on $k$ nodes---and also, if $k$ is even, trees on $k/2\!+\!1$ nodes---dominate all closed walks at scale~$k$:
\begin{multline*}
\E\#\{\textnormal{closed $k$-walks in $G$}\}
\sim \E\#\{\textnormal{closed $k$-walks tracing out a $k$-cycle in $G$}\}
\\+1_{\{k\textnormal{ even}\}} \cdot \E\#\{\textnormal{closed $k$-walks tracing out any $(k/2+1)$-tree in $G$}\}.
\end{multline*}
\end{Theorem}

Theorem~\ref{thm:main} (proved in the Supplementary Information) shows in greatest generality when walks tracing out trees and cycles will dominate all walks in a network. Equation~\eqref{assumption-main} is a universal requirement for walks to extend fully. For it to hold in the setting of generalized random graphs, network degrees must grow. This in turn implies the emergence of a giant component\cite{bollobas2007phase}, and so Theorem~\ref{thm:main} links global network structure to local walk properties. 

Theorem~\ref{thm:main} also leads to simple explicit forms for $\E\#\{\textnormal{copies of $w$ in $G$}\}$, depending on whether all network degrees grow uniformly\cite{holland1983stochastic, hoff2002latent, airoldi2008mixed, zhao2012consistency} (Supplementary Information, Proposition~\ref{grg-prop}), or at variable rates, as in a power-law network\cite{chung2003spectra} with hubs and scale-free structure\cite{Barabasi99} (Supplementary Information, Proposition~\ref{powlaw-prop}).

This result points to a fundamental dichotomy: When networks are sufficiently homogeneous that equation~\eqref{eq:nb-of-copies} applies, all trees are equally important, regardless of their degree sequence. But when networks are more heterogeneous and equation~\eqref{eq:nb-of-copies2} applies, as in the case of a power law, trees containing hubs dominate (Supplementary Information, Theorem~\ref{propn:irg_thm_validity}). In both cases a universal phase transition regulates walk dominance.

\begin{Corollary}\label{corrol_main}
Consider the setting of generalized random graphs whose mechanism $\mathcal{M}$ is either a bounded kernel or a power law. If network degrees grow fast enough, then dominant closed walks exhibit the sharp phase transition shown in Fig.~\ref{fig:tadpole}c:
\vspace{-.4\baselineskip}\begin{multline*}
\hskip-0.33cm\E \# \{\textnormal{closed $k$-walks in $G$}\}
\sim \\
\hskip0.515cm\begin{cases}
\E \# \{\textnormal{closed $k$-walks tracing out a $k$-cycle in $G$}\} & \textnormal{\!\!\!if $k$ is odd or $k>k^\ast\!$,} \\
\E \# \{\textnormal{closed $k$-walks tracing out $(k/2+1)$-trees in $G$}\} & \textnormal{\!\!\!otherwise.}
\end{cases}
\vspace{-.4\baselineskip}\end{multline*}
The boundary between dominating regimes of trees and cycles occurs at
\vspace{-.1\baselineskip}\begin{equation*}
k^\ast={\log \#\{\textnormal{nodes in $G$}\}}/{\log\smash{\sqrt{\#\{\textnormal{nodes in $G$}\} \cdot \operatorname{\mathbb{P}} \{\textnormal{edge in $G$}\}}}}.
\vspace{-.1\baselineskip}\end{equation*}
By contrast, dominant non-backtracking walks always trace out cycles:
\vspace{-.4\baselineskip}\begin{multline*}
\E\#\{\textnormal{non-backtracking closed $k$-walks in $G$}\} \sim\\
\E\#\{\textnormal{closed $k$-walks tracing out a $k$-cycle in $G$}\}.
\vspace{-.4\baselineskip}\end{multline*}
\end{Corollary}
 
Corollary~\ref{corrol_main} reveals a sudden shift in dominating regimes between trees and cycles, driven entirely by $\operatorname{\mathbb{P}} \{\textnormal{edge in $G$}\}$.  This shift is scale-dependent and simple to describe: the sparser the network, the more scales are dominated by trees. Fluctuations near the boundary $k^\ast$, and the emergence of a regime dominated by star trees when degrees grow slowly, can be quantified precisely (Supplementary Information, sections~\ref{sec:S3} and~\ref{sec:S5}).  Non-backtracking walks, by contrast, are robust to this phase transition: They cannot backtrack and traverse the same network edge twice in succession, and so cannot trace out shapes containing trees. Instead, Corollary~\ref{corrol_main} shows that non-backtracking walks are dominated by walks tracing out cycles.

Corollary~\ref{corrol_main} also quantifies the behavior of a network's adjacency matrix $A$ relative to its non-backtracking matrix $B$. While $A$ tabulates all paths of length one, $B$ tabulates all non-backtracking paths of length two. It follows in turn that $ \smash{ \E\lambda_A^k = \E\#\{\textnormal{closed $k$-walks in $G$}\} } $, where $\E\lambda_A^k$ is the average value of an eigenvalue of $A$ chosen at random, while $ \smash{ \E\lambda_B^k = \E\#\{\textnormal{non-backtracking closed $k$-walks in $G$}\} } $. Corollary~\ref{corrol_main} pinpoints how $\lambda_A$ shifts regimes as $G$ becomes sparser, whereas $\lambda_B$ remains stable, effectively decoupling from $\operatorname{\mathbb{P}} \{\textnormal{edge in $G$}\}$. This stability explains why $B$ rather than $A$ governs the fundamental limits of network community detection\cite{ahn2010link, Krzakala13, Newman14}, and suggests that its optimality properties may hold much more widely.  

\subsection*{Dominant scales in static and dynamic networks}

Trees and cycles reveal a network's dominant scales, thereby establishing a theoretical basis for network comparison. To implement this comparison, we sample sub-networks at random and count the normalized proportions of trees and cycles present (Supplementary Information, Algorithm~\ref{alg:violins}). We characterize the distributions of these proportions using violin plots\cite{hintze1998violin} (Figs.~\ref{fig:tadpole}b and~\ref{fig:compnets}), to isolate and quantify the heterogeneity of network structure locally at each scale. We determine the necessary size of sub-networks to sample by iteratively adapting to the network under study (Supplementary Information, Algorithm~\ref{alg:sizes}).

\begin{figure}[t]
\vspace{-5\baselineskip}\hskip-1.75cm{\includegraphics[width=1.25\textwidth]{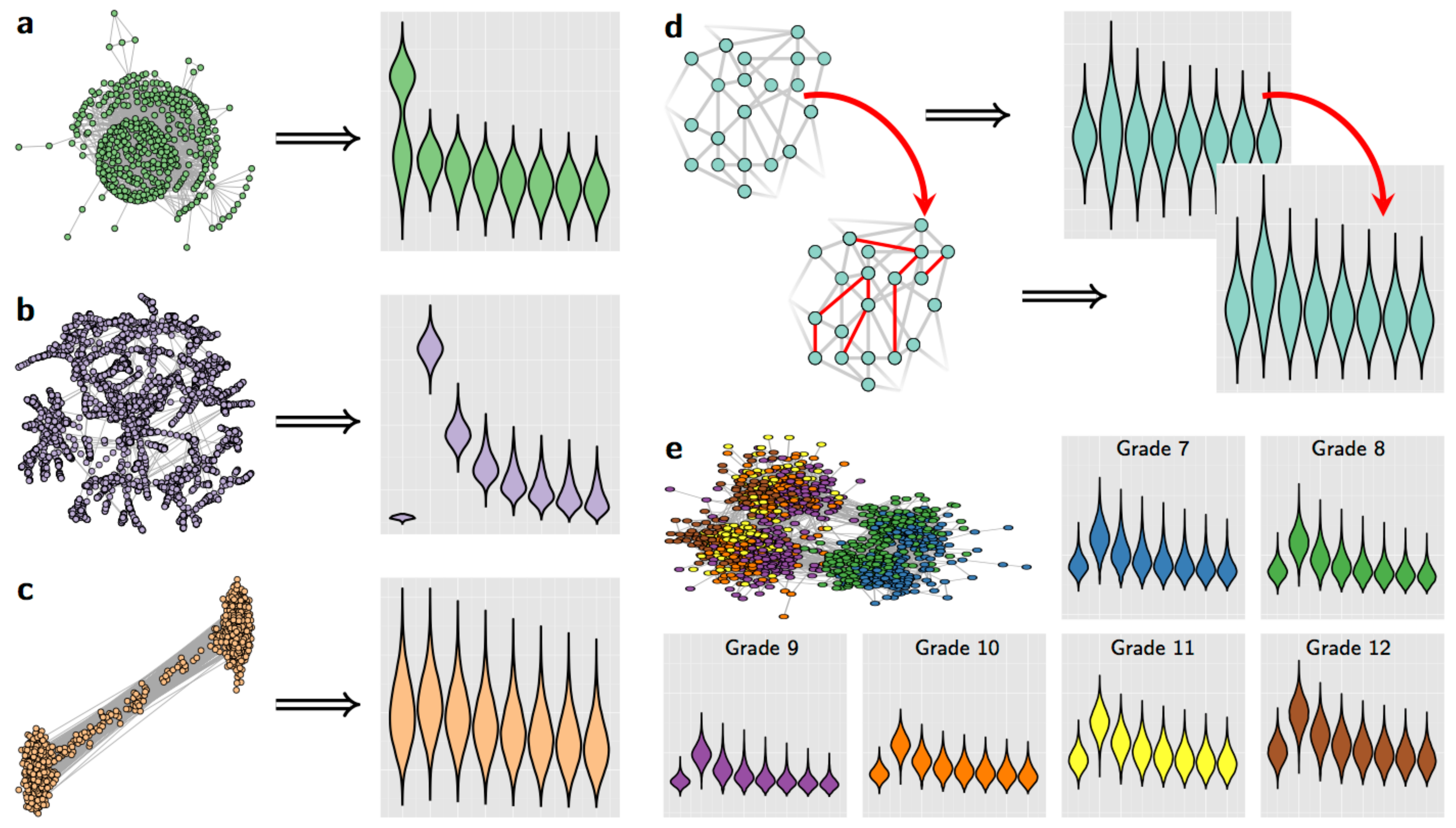}}%
	\caption{\label{fig:compnets}\textbf{Comparing networks and their dominant scales.} \protect\textbf{a}, The \emph{C.~elegans} metabolic network is dominated by tree-like structure characteristic of preferential attachment, which manifests at the smallest scale (cf.~Fig.~\protect\ref{fig:typexamp}a). \protect\textbf{b}, A power grid network, by contrast, exhibits small-world connectivity (cf.~Fig.~\protect\ref{fig:typexamp}b), and is dominated by medium-range scales. \protect\textbf{c}, A network of political weblogs has strong community structure (cf.~Fig.~\protect\ref{fig:typexamp}c), but surprisingly shows relatively uniform contributions from all scales. \protect\textbf{d}, A random graph with triadic closure (\smash{\protect\Textpathtwo} $\!\!\rightarrow\!\!$ \,\smash{\protect\Texttrianglenew}\!) shows a single elevated \smash{\protect\Texttrianglenew}\! scale. \protect\textbf{e}, Adolescent friendship networks evolve across school grades, maintaining triadic closure despite changes in connectivity (see main text).}
\end{figure}

Figure~\ref{fig:compnets} shows a comparative analysis of networks with archetypical features: preferential attachment characteristics\cite{Barabasi99}, small-world connectivity\cite{watts1998small}, community structure\cite{adamic2005political}, triadic closure (\smash{\Textpathtwo} $\!\!\rightarrow\!\!$ \,\smash{\Texttrianglenew}\!)\cite{granovetter1973strength}, and assortative mixing\cite{resnick1997protecting}. Our analysis highlights the fundamental differences in network structure reflected by these features, distinct from differences in network size and sparsity. Not only can we detect and visualize these structural differences in the context of network comparison, but also we can identify changes in networks at the level of their individual scales.

To this end, Figure~\ref{fig:compnets}e compares adolescent friendship networks across different years (grades) in school.  Among high school students (grades 9--12), overall connectivity levels---known in this context as sociality\cite{goodreau2009birds}---increase with grade.  Yet at the same time, dominant scales persist across grades, even in the transition from middle school (grades 7--8) to high school. These dominant scales reveal a strong triadic closure effect (cf.~Fig.~\ref{fig:compnets}d). Previous studies have shown that the aggregate friendship network across grades shows evidence of selective mixing\cite{goodreau2009birds}, with a complex community structure based in part on gender, race, and grade\cite{olhede2013network}. Despite this complexity, we see that scale-based structure in this setting persists and appears robust to the many social changes taking place in adolescence.

\begin{figure}[t]
	\vspace{-5\baselineskip}\hskip-1.75cm{\includegraphics[width=1.25\textwidth]{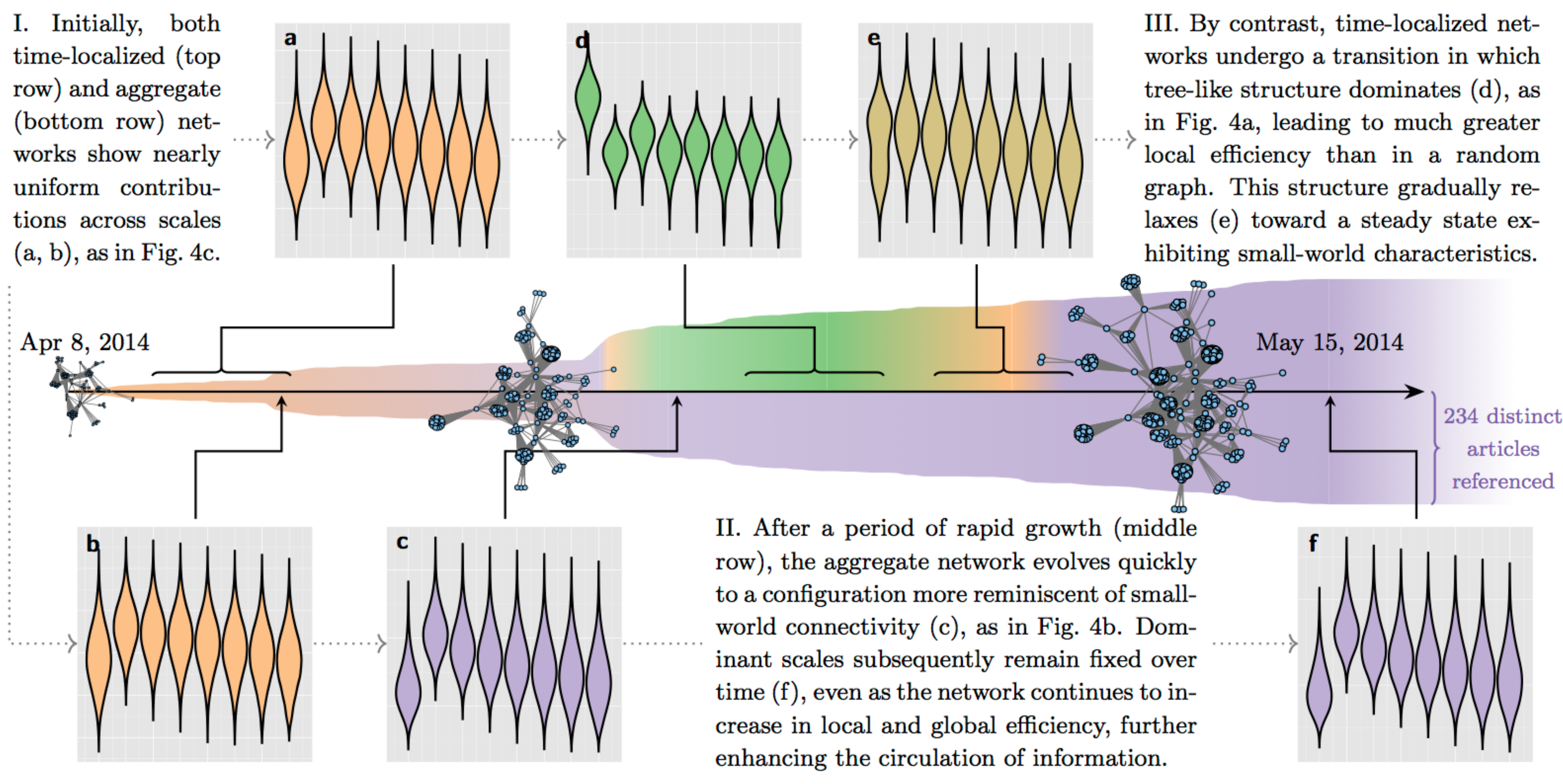}}%
	\caption{\label{fig:twitter}\textbf{Evolution of dominant scales in an online discussion network.} A social media (Twitter) discussion about the release of a new consumer product (iPhone 6 mobile device) begins early in April 2014 as rumors and leaks of the new product are shared. The dynamics of this discussion manifest in time-localized networks (top row), where edges form whenever two users broadcast similar messages within a 4-day window, as well as an aggregate network (middle and bottom rows), where edges form whenever two users refer to the same online news article.}
\end{figure}

Figure~\ref{fig:twitter} shows the rich temporal dynamics of an evolving social media discussion, revealed through the archetypical features of Fig.~\ref{fig:compnets}. Initially, the corresponding discussion network is simply structured and shows relatively uniform contributions across scales (Figs.~\ref{fig:twitter}a and~\ref{fig:twitter}b; cf.~Fig.~\ref{fig:compnets}c). As rumors of a new product and its manufacturer dominate the discussion, a large clique forms and connects to several smaller, nearly disjoint clusters of users. This clique breaks apart after April 15, when leaked product images appear online and are shared across the network (Supplementary Information, section~\ref{sec:S7}). In response, the network converges rapidly to small-world connectivity (Figs.~\ref{fig:twitter}c and~\ref{fig:twitter}f; cf.~Fig.~\ref{fig:compnets}b)---known to circulate information efficiently\cite{latora2001physrevE}---as its first, tree-like scale deprecates.

As the discussion shown in Fig.~\ref{fig:twitter} progresses, we continue to recognize canonical signatures of fundamental network generating mechanisms. A rapid burst of network activity follows the manufacturer's earnings report, released on April 23.  This manifests as a sharp elevation of the first, tree-like network scale indicative of preferential attachment (Fig.~\ref{fig:twitter}d; cf.~Fig.~\ref{fig:compnets}a), along with increases in local and global efficiency\cite{latora2001physrevE} relative to random and lattice graphs\cite{watts1998small}. Star-like nodes appear suddenly and dominate the network structure, then gradually recede as the network consolidates (Fig.~\ref{fig:twitter}e), adding triangles and higher-order cycles to reflect small-world connectivity properties. Overall, as the discussion network evolves through a combination of growth and structural changes, the dynamics of its dominant scales enable us to identify and describe different developmental phases in its life cycle.

\subsection*{Discussion}

Here we have shown how algebraic topology leads naturally to a set of universal features for network comparison. These features provide the first theoretically justified, automated, and scalable means to compare networks in a model-free way. Critical to our discovery is Theorem~\ref{thm:main}, which shows that under very weak assumptions, certain shapes---trees and cycles---dominate at every network scale. These shapes and scales arise from properties of non-backtracking closed walks, answering longstanding open questions in the field of network motif analysis\cite{Milo2004super, alon2007network, BickelLevina2012, gerstein2012architecture, boyle2014comparative}. 

The need to characterize and compare networks is fundamental to many fields, from the physical and life sciences to the social, behavioral, and economic sciences. This need is currently particularly significant in understanding information cascades and spreading online.  From sharing trees\cite{del2016spreading} to the role of cycles in the strength of weak ties\cite{granovetter1973strength}, social networks modulate the diffusion of information and influence. At the same time, advances in machine learning and artificial intelligence mean that the spread of contemporary news is strongly influenced by the uniquely personalized online social context of each individual in a network\cite{bakshy2015exposure}. Consequently, understanding how structural network properties interact with the dynamics of spreading processes is a more timely and important problem now than ever before.

To this end, we have shown, through a comparative analysis of several rich and varied examples, how our discovery enables new insights into the structure and function of of static and dynamic networks. Our approach is sufficiently flexible to permit the comparison of networks with different numbers of nodes and edges, enabling scientists to extract common patterns and signatures from a variety of canonical network types. Surprisingly, we observe that these same signatures appear in the analysis of a social media discussion which evolves through the sharing of rumors, leaks and other types of news.  This allows us to observe and quantify, in a way never before possible, the direct and strong symbiosis between the dynamics of an evolving network and the ways in which information spreads within it. The discoveries we report are a crucial first step in disentangling the features that make networks similar or different from those that facilitate the diffusion of information.

\putbib[motif]
\end{bibunit}

\medskip

\textbf{Supplementary Information} begins on the following page.

\medskip

\textbf{Acknowledgements} This work was supported in part by the US Army Research Office under Multidisciplinary University Research Initiative Award 58153-MA-MUR; by the US Office of Naval Research under Award N00014-14-1-0819; by the UK Engineering and Physical Sciences Research Council under Mathematical Sciences Leadership Fellowship EP/I005250/1, Established Career Fellowship EP/K005413/1, Developing Leaders Award EP/L001519/1, and Award EP/N007336/1; by the UK Royal Society under a Wolfson Research Merit Award; and by Marie Curie FP7 Integration Grant PCIG12-GA-2012-334622 and the European Research Council under Grant CoG 2015-682172NETS, both within the Seventh European Union Framework Program. The authors thank FSwire for making available the data used to produce Figs.~\ref{fig:tadpole}b and~\ref{fig:twitter}. This work was partially supported by a grant from the Simons
Foundation, and the authors simultaneously acknowledge the Isaac Newton Institute for Mathematical Sciences, Cambridge, UK, for support and hospitality during the program Theoretical Foundations for Statistical Network Analysis (Supported by EPSRC award EP/K032208/1) where a portion of the work on this article was undertaken.



\medskip

\textbf{Author Contributions} All authors contributed to all aspects of the paper.

\medskip

\textbf{Author Information} The authors declare no competing financial interests. Correspondence should be addressed to P.J.W.\ (p.wolfe@ucl.ac.uk).


\newpage
\renewcommand\refname{}
\bibpunct{[}{]}{,}{n}{}{;}
\begin{bibunit}[nature]
\setcounter{figure}{0}%
\renewcommand{\thefigure}{S\arabic{figure}}%
\setcounter{table}{0}%
\renewcommand{\thetable}{S\arabic{table}}%
\setcounter{equation}{0}%
\renewcommand{\theequation}{S.\arabic{equation}}
\setcounter{section}{0}%
\renewcommand{\thesection}{S.\arabic{section}}
\setcounter{Theorem}{0}%
\renewcommand{\theTheorem}{S.\arabic{Theorem}}
\setcounter{Assumption}{0}%
\renewcommand{\theAssumption}{S.\arabic{Assumption}}
\setcounter{Remark}{0}%
\renewcommand{\theRemark}{S.\arabic{Remark}}
\setcounter{Lemma}{0}%
\renewcommand{\theLemma}{S.\arabic{Lemma}}
\setcounter{Proposition}{0}%
\renewcommand{\theProposition}{S.\arabic{Proposition}}
\setcounter{Definition}{0}%
\renewcommand{\theDefinition}{S.\arabic{Definition}}
\setcounter{Corollary}{0}%
\renewcommand{\theCorollary}{S.\arabic{Corollary}}
\newpage

\begin{center}
\begin{minipage}[t][13\baselineskip][t]{\textwidth}%
\begin{center}
{{\LARGE{\vspace{-3\baselineskip}Supplementary Information: \\ Topology reveals universal features for network comparison\vspace{-0.67\baselineskip}}\\%
\ \\%
{\large{Pierre-Andr\'e G.\ Maugis, Sofia C.\ Olhede \& Patrick J.\ Wolfe}}%
}}\\\ \\ {\large{\vspace{-0.67\baselineskip}University College London}}%
\end{center}%
\end{minipage}%
\end{center}

\setcounter{section}{0}
\renewcommand{\thesection}{\arabic{section}}

\vspace{-12\baselineskip}%
\tableofcontents

\newpage

\section{Introduction}\label{sec:S-1}
%
In this Supplementary Information we provide proofs of all results from the main text, along with details of the corresponding methods, algorithms, and datasets. It is written as a self-contained document, with the remainder of this introductory section relating the results and notation featured here to the main text.

A key driver of our proofs comes via the introduction of a walk extension, which in turn permits us to compare the prevalence of any two walks of the same length (Fig.~\ref{fig:walks}, main text). We define as ``simplest'' precisely those walks that cannot be extended any further (Fig.~\ref{fig:extension}, main text). We then show that the simplest walks dominate other walks in terms of their expected prevalence. By the argument of Fig.~\ref{fig:tadpole} in the main text, the simplest closed walks turn out to be those mapping out either trees or cycles at maximal scales. Together these results lead to Theorem~\ref{thm:main} in the main text, and are obtained in two steps. First, after providing basic definitions in Section~\ref{sec:S0} of this Supplemental Information, we introduce the general framework and necessary preliminary results in Section~\ref{sec:S1}. Then, in Section~\ref{sec:S2}, we prove that under suitable conditions, the total number of closed walks in a network is dominated in expectation by walks inducing trees or cycles. 

We then show additional results which apply if we assume more about the network of interest, leading to Corollary~\ref{corrol_main} in the main text. In Sections~\ref{sec:S3} and~\ref{sec:S5} of this Supplemental Information, we consider the setting of generalized random graphs~\citep{bollobas2007phase}, generated either from a bounded kernel or from an unbounded kernel giving rise to a power-law degree distribution. Applying Theorem~\ref{thm:main} in these settings leads to the following two propositions, which hold respectively for the two large families of graphs described by Definitions~\ref{model} and~\ref{powerlaw} in Sections~\ref{sec:S3} and~\ref{sec:S5}.

\begin{Proposition}\label{grg-prop}
Let $\{G\}$ be a sequence of generalized random graphs generated from a bounded, symmetric kernel $\smash{ \kappa: (0,1)^2 \to [0,\infty) }$, with $\| \kappa \|_1 = 1$ and $\#\{\textnormal{nodes in $G$}\}\to\infty$. Assume edges in $G$ form independently, conditionally upon a random sample $\xi$ of $\operatorname{Uniform}(0,1)$ variates, and that there exists a sequence $\{ \operatorname{\mathbb{P}} \{\textnormal{edge in $G$}\} \}$ taking values in $ \smash{ (0, \| \kappa \|_{\infty}^{-1} ) } $ such that
\begin{equation*}
\mathbb{P}\{\textnormal{edge connecting nodes $i$ and $j$ in $G \,\vert\, \xi$}\} = \operatorname{\mathbb{P}} \{\textnormal{edge in $G$}\} \cdot \kappa(\xi_i,\xi_j).
\end{equation*}
Then for any closed $k$-walk $w$, we recover equation~\eqref{eq:nb-of-copies}, with
\begin{multline*}
\log \operatorname{\mathbb{P}} \{\textnormal{local copy of $w$ in $G$}\} = \textnormal{const}\\ - \#\{\textnormal{edges traversed by $w$}\} \cdot |\log\operatorname{\mathbb{P}} \{\textnormal{edge in $G$}\}|,
\end{multline*}
and if $\#\{\textnormal{nodes in $G$}\} \cdot \operatorname{\mathbb{P}} \{\textnormal{edge in $G$}\}\to\infty$, then Theorem~\ref{thm:main} applies.
\end{Proposition}

\begin{Proposition}\label{powlaw-prop}
Let $\{G\}$ be a sequence of inhomogeneous random graphs generated from a rank one kernel yielding power-law degrees, with $\#\{\textnormal{nodes in $G$}\}\to\infty$. Specifically, assume that edges in $G$ form independently, and that there exists a constant $\gamma \in (0,1)$ and a monotone sequence $\{\theta\}$ taking values in $(0,1]$ such that
\begin{equation*}
\smash{ \mathbb{P}\{\textnormal{edge connecting nodes $i$ and $j$ in $G$}\} = \theta^2 \cdot (ij)^{-\gamma} } .
\end{equation*}
Then $\operatorname{\mathbb{P}} \{\textnormal{edge in $G$}\} \sim (1-\gamma)^{-2} \theta^2 \cdot \#\{\textnormal{nodes in $G$}\}^{-2\gamma}$, and if the average degree tends to infinity faster than $\#\{\textnormal{nodes in $G$}\}^\gamma$, then Theorem~\ref{thm:main} applies, whence for degrees $d_1,\dots , d_v$ of the shape traced out by any closed $k$-walk $w$,
\begin{multline}
\log \operatorname{\mathbb{P}} \{\textnormal{local copy of $w$ in $G$}\}
\sim - \textstyle{\sum}_{t=1}^v \min(1,d_t\gamma) \cdot \log \#\{\textnormal{nodes in $G$}\} 
\\ - \smash{ \#\{\textnormal{edges traversed by $w$}\} \cdot | \log\theta^2 | } . \label{eq:powerlaw_prob}
\end{multline}
\end{Proposition}

To prove these propositions, as well as Corollary~\ref{corrol_main} in the main text, we proceed as follows. In Section~\ref{sec:S3}, we prove for kernel-based random graphs (also known as inhomogeneous or generalized random graphs) that closed walks are dominated in expectation by walks inducing either trees or cycles, with a phase transition between the two regimes. We additionally prove that non-backtracking closed walks are dominated in expectation by walks inducing cycles. We derive error rates, making these results more precise than the general results of Section~\ref{sec:S2}, which are based only on an assumption about how walk extensions scale. Next, we show in Section~\ref{sec:S5} that for inhomogeneous random graphs with a power-law degree distribution, closed walks are again dominated in expectation by walks inducing either trees or cycles when Theorem~\ref{thm:main} applies, and that non-backtracking closed walks are again dominated by those inducing cycles. This provides a set of developments parallel to Section~\ref{sec:S3}, and establishes the above propositions as well as Corollary~\ref{corrol_main} in the main text. We also establish dominant walks outside the setting of Theorem~\ref{thm:main}, showing when stars and star-like graphs containing cycles will dominate. Finally, in Section~\ref{sec:S6} we detail the methods and algorithms we develop to make use of these results, and in Section~\ref{sec:S7} we provide details of the network datasets we analyze in the main text.

To facilitate mapping results from the main text to the notation of this Supplementary Information, we detail the following relationship.

\begin{description}
\item[Definitions of walk densities and extensions] The term introduced in the main text as $\smash{\operatorname{\mathbb{P}}_{\!\!\!\mathcal{M}} \{\textnormal{local copy of $w$ in $G$}\}}$ corresponds to the graph walk density of Definition~\ref{varphi}. The main text definition of a walk extension corresponds to Definition~\ref{def:extens}. We give the basic properties of the graph walk density and of walk extensions in Lemma~\ref{lem:sampling} and in Remark~\ref{remark:extension}.
\item[Theorem~\ref{thm:main}]
We prove Theorem~\ref{thm:main} in Section~\ref{sec:S2} as Theorem~\ref{thm:dominant-walk}. Proving the theorem requires Proposition~\ref{w-k-ast-prop} from Section~\ref{sec:S2} and Lemma~\ref{extens-prop} from Section~\ref{sec:S1}. The proof is made under Assumptions~\ref{asump1.1} and~\ref{asump1.3}. The latter assumption is equivalent to the condition of~\eqref{assumption-main} in the main text.
\item[Proposition~\ref{grg-prop}] This proposition is a consequence of Remark~\ref{remark:kernelver}. We show in the remark that the conditions of Theorem~\ref{thm:main} are satisfied for generalized random graphs with bounded kernels. This then implies that the proposition holds for generalized random graphs that are sufficiently dense and arise from bounded kernels.
\item[Proposition~\ref{powlaw-prop}] We prove this proposition in Section~\ref{sec:S5} as Proposition~\ref{ihrg-ptop}. The proof relies on Theorem~\ref{thm:main}.
\item[Corollary~\ref{corrol_main}, part 1] We prove this result by splitting it into two cases. First, for the case of generalized random graphs with bounded kernels, we show the result via Theorem~\ref{thm:phase} in Section~\ref{sec:S3}. Its proof relies on Lemmas~\ref{exp-ind} and~\ref{dominantscaling}. Second, for an inhomogeneous random graph with an unbounded but separable kernel, we show the result via Proposition~\ref{ihrg-ptop}.
\item[Corollary~\ref{corrol_main}, part 2] We also prove this result in two parts. First, for generalized random graphs with bounded kernels, we establish the result via Theorem~\ref{thm:back} in Section~\ref{sec:S3}. Its proof requires Lemma~\ref{non-back:w}. Second, for an inhomogeneous random graph with an unbounded but separable kernel, we show the result as above via Proposition~\ref{ihrg-ptop}.
\item[Methods and algorithms] The methods and algorithms we use for all data analysis in the main text are described in Section~\ref{sec:S6}.
\item[Network datasets] The datasets we study are described in Section~\ref{sec:S7}.
\end{description}

Finally, Table~\ref{tabby} overleaf shows how the notation in this Supplementary Information maps to the notation we use in the main text.

\begin{table}
\begin{center}
\begin{tabular}{ >{\centering\arraybackslash}p{5cm} >{\centering\arraybackslash}m{5cm} >{\centering\arraybackslash}c}
{\bf Main text notation} & {\bf Supplementary notation} & {\bf Definition}\tabularnewline
\hline
Shape traced out by walk $w$ & $F_w$ & \ref{closedalkdef} \tabularnewline
$\#\{\textnormal{nodes visited by $w$}\} $ & $v(F_w)$ & - \tabularnewline
$\#\{\textnormal{edges traversed by $w$}\} $ & $e(F_w)$ & - \tabularnewline
$\#\{\textnormal{nodes}\textnormal{ in $G$}\} $ & $n$ & - \tabularnewline
$\mathbb{P}\{\textnormal{edge in }G\} $ & $\rho$ & - \tabularnewline
Average network degree & $\mu(1-n^{-1})$ & - \tabularnewline
Set of closed $k$-walks in $G$ & $ \mathcal{W}_k(G) $ & \ref{closedwalks} \tabularnewline
$\#\{\textnormal{closed $k$-walks in $G$}\}$ & $\left| \mathcal{W}_k(G) \right|$ & - \tabularnewline
\begin{minipage}[c][1cm][c]{5cm}\centering
	Set of unlabeled subgraphs\\[-.2\baselineskip]
	induced by closed $k$-walks\end{minipage}& $W_k$ &\ref{closedalkdef}\tabularnewline
Subgraph extension & $\vartriangleleft_k$ &\ref{def:extens} \tabularnewline
$\#\{\textnormal{copies}\textnormal{ of $w$ in $G$}\} $ & $\mathrm{ind}_k(F_w,G)$ & \ref{defnocop} \tabularnewline
$\#\{\textnormal{copies}\textnormal{ of $F$ in $G$}\} $ & $X_F(G)$ & - \tabularnewline
$\operatorname{\mathbb{P}} \{\textnormal{local copy of $w$ in $G$}\}$ & $\E\varphi(w,G)$ & \ref{varphi} \tabularnewline
Kernel & $\kappa$ & \ref{kernel} \tabularnewline
\begin{minipage}[c][1cm][c]{5cm}\centering
	Generalized\\[-.25\baselineskip]
	random graph\end{minipage} & $G(n,\rho\kappa)$ &  \ref{model}\tabularnewline
\begin{minipage}[c][1cm][c]{5cm}\centering
	\#\{\textnormal{closed $k$-walks visiting}\\[-.2\baselineskip]
	\textnormal{a $k$-cycle in $G$}\}\end{minipage} & $\left|\{w\in\mathcal{W}_k(G) : F_w\equiv C_k\}\right|$ & - \tabularnewline
\begin{minipage}[c][1cm][c]{5cm}\centering
	\#\{\textnormal{closed $k$-walks visiting}\\[-.2\baselineskip]
	\textnormal{any $(k/2+1)$-tree in $G$}\}\end{minipage} & $\left|\{w\in\mathcal{W}_k(G) : F_w \in \mathcal{T}_{k/2+1}\}\right|$ & -\tabularnewline
\begin{minipage}[c][1cm][c]{5cm}\centering
	Set of non-backtracking\\[-.2\baselineskip]
	closed $k$-walks in $G$\end{minipage}&$\mathcal{W}_k^{_b}(G)$ & \ref{def-non-back} \tabularnewline
\begin{minipage}[c][1cm][c]{5cm}\centering
	\#\{\textnormal{non-backtracking closed}\\[-.2\baselineskip]
	\textnormal{$k$-walks in $G$}\} \end{minipage}&$\left| \mathcal{W}_k^{_b}(G) \right|$ & - \tabularnewline
\begin{minipage}[c][1.5cm][c]{5cm}\centering
	Set of unlabeled subgraphs\\[-.2\baselineskip]
	induced by non-backtracking,\\[-.2\baselineskip]
	tailless closed walks\end{minipage}& $\smash{W_k^{_b}}$ &\ref{defWkb}\tabularnewline
\begin{minipage}[c][1.5cm][c]{5cm}\centering
	Number of non-backtracking,\\[-.2\baselineskip]
	tailless closed $k$-walks\\[-.2\baselineskip]
	inducing copy of $F$\end{minipage} &$\mathrm{ind}_k^b(F,G)$ & \ref{defnocopb}\tabularnewline
\hline
\end{tabular}
\end{center}
\caption{\label{tabby}\textbf{Notation used in the main text and Supplementary Information}}
\end{table}

\newpage

\section{Preliminaries}\label{sec:S0}
\subsection{Notation}

We begin by providing basic graph-theoretic definitions needed for our analysis. All graphs throughout are assumed to be finite and simple (unweighted, undirected, and without self-loops).
\begin{enumerate}
\item We write $G$ for a simple graph with vertex set $v(G)$ and edge set $e(G)$. Often we write $n = |v(G)|$.
\item We write $F\subset G$ if $F$ is a subgraph of $G$; i.e., $v(F)\subset v(G)$ and $e(F)\subset e(G)$.
\item We say that two simple graphs $G$ and $G'$ are isomorphic and write $G\equiv G'$ if there exists a bijection $\Phi \colon v(G)\to v(G')$ such that $pq\in e(G) \Leftrightarrow \Phi(p)\Phi(q)\in e(G')$.
\item By a labeled graph, we mean any finite, simple graph $G$. By an unlabeled graph, we mean an element of the set of isomorphism classes of finite simple graphs (or a representative thereof).
\item We write $\mathrm{aut}(G)$ for $\left| \left\{ \Phi \in \mathrm{Sym}\left( v(G) \right) : pq\in e(G) \Leftrightarrow \Phi(p)\Phi(q)\in e(G) \right\} \right|$, the order of the automorphism group of $G$; i.e., the number of adjacency-pres\-erving permutations of $v(G)$.
\item We denote by $\mathrm{emb}(F,G)$ for the number of embeddings (injective homomorphisms) of $F$ into $G$; i.e., the number of labeled copies of $F$ in $G$.
\item We denote by $X_F(G) = \left| \left\{ F'\subset G : F'\equiv F \right\} \right| = \mathrm{emb}(F,G) / \mathrm{aut}(F) $ the number of subgraphs of $G$ isomorphic to $F$; i.e., the number of unlabeled or isomorphic copies of $F$ in $G$.
\item We denote by $K_n$ the complete graph on $n$ vertices (i.e., with $\smash{\tbinom{n}{2}}$ edges).
\item A $k$-cycle $C_k$ is the cycle graph on $k$ vertices.
\item A $k$-tree is any tree on $k$ vertices; i.e, any connected graph on $k$ vertices without cycles, or equivalently with $k-1$ edges. We write $\mathcal{T}_k$ for the set of all unlabeled $k$-trees.
\item A $k$-path $P_k$ is the $k$-tree containing $k-2$ vertices of degree two and $2$ vertices of degree one (its endpoint vertices or leaves), for $k\geq 2$.  For $k=1$, $P_1 = K_1$ (the singleton graph).
\item A $(k,l-1)$-tadpole $C_k P_l$ is the graph obtained by joining $C_k$ to $P_l$ by identifying a single vertex. It is also known in the literature as a balloon graph (and sometimes even as a dragon, kite, canoe paddle or lollipop graph, though often the last of these refers instead to a clique joined to a path).
\item A $(k,l)$-lemniscate $C_k C_l$ is the graph obtained by joining $C_k$ to $C_l$ by identifying a single vertex. It is also known in the literature as a bouquet or flower graph.
\end{enumerate}

\subsection{Counting closed walks}
We begin by grouping closed walks of a given length $k$ according to the subgraphs they induce. We now introduce the sets $\mathcal{W}_k(G)$ and $W_k$ necessary to implement such a grouping.

\begin{Definition}[The set $\mathcal{W}_k(G)$ of closed $k$-walks in a simple graph $G$]\label{closedwalks}
Fix a simple graph $G$ and $k\in \mathbb{N}$. A walk $w$ of length $k$ in $G$ is a sequence of adjacent vertices in $G$:
\begin{equation*}
w = v_0 v_1 \cdots v_k,
\end{equation*}
where $ \cup_{i=0}^k \{v_i\} \subset v(G) $ and $ \cup_{i=0}^{k-1} \{ v_i v_{i+1} \} \subset e(G) $. If $v_0 = v_k$ then the walk is closed. We denote by $\mathcal{W}_k(G)$ the set of all closed $k$-walks in a given graph $G$ and write $|w|$ for the length of $w$.
\end{Definition}

If $A(G)$ is the adjacency matrix of a simple graph $G$, then $ \left| \mathcal{W}_k(G) \right| = \mathrm{Tr}(A(G)^k) $ for any $k\in \mathbb{N}$.

\begin{Definition}[Walk-induced subgraphs $F_w \subset G$ and the set $W_k$ of unlabeled subgraphs induced by closed $k$-walks]\label{closedalkdef}
Fix a simple graph $G$ and a walk $w$ in $G$. We call $F_w \subset G$ the labeled subgraph of $G$ induced by the edges traversed by $w$; i.e., the labeled subgraph with vertex set $v(F_w) = \cup_{i=0}^k \{v_i\}$ and edge set $e(F_w) = \cup_{i=0}^{k-1} \{ v_i v_{i+1} \}$. We denote by $W_k$ the set of all unlabeled graphs induced by closed walks of length $k$:
\begin{equation*}
W_k = \{F\subset K_k : \exists w\in\mathcal{W}_k(K_k) \textnormal{ s.t. } F = F_w\} \mathbin{/} \equiv.
\end{equation*}
\end{Definition}

Thus, for any fixed $k$, $W_k$ is a subset of the set of isomorphism classes of finite simple graphs. When enumerating its elements, we will implicitly choose a representative of the corresponding equivalence class for each element, so that we may treat each $ F \in W_k $ as an arbitrarily labeled graph.

\begin{Lemma}[Properties of $W_k$]\label{wk-props}
For every integer $k \geq 2$, the set $W_k$ is non-empty and satisfies the following properties:
\begin{enumerate}
\item For each $F \in W_k$, there exists a closed $k$-walk $w = v_0v_1\cdots v_{k-1}v_0$ in $F$ for which $F_w = F$.
\item Every $F \in W_k$ is connected.
\item Any $ F \in W_k $ with $|e(F)| < |v(F)| $ is a tree, and hence in this case $|e(F)| = |v(F)| - 1$.
\item If $w$ is a closed $k$-walk and $ F_w \in W_k $ is a tree, then every edge in $F_w$ is traversed at least twice by $w$.
\item Any tree $ T \in W_k $ has $ |e(T)| = |v(T)| - 1 \leq k/2 $.
\item It holds that $W_k \subset W_{k+2}$.
\item If $k$ is odd, then $W_k$ contains no trees.  If $k$ is even, then $W_k$ contains all unlabeled trees on $2$ to $k/2+1$ vertices.
\item The $k$-cycle $C_k$ is an element of $W_k$ for all $k \geq 3$, and is the only element of $W_k$ on $k$ vertices.
\item The $(k-2,1)$-tadpole $C_{k-2}P_2$ is an element of $W_k$ for all $k \geq 5$, and is the only element of $W_k$ that is both on $k-1$ vertices and with $k-1$ edges.
\end{enumerate}
\end{Lemma}

\begin{proof}  
We shall prove all items in order, using results from each part in turn.  Items~7 and~8, once proved, imply that $W_k$ is non-empty for every integer $k \geq 2$.

\vspace{.5\baselineskip}
\emph{Proof of 1:} 
Fix $F\in W_k$, recalling that we treat $F$ as an arbitrarily labeled graph. From Definition~\ref{closedalkdef}, there exists $F'\subset K_k$ such that: i) $F'\equiv F$; and ii) there exists $w'\in\mathcal{W}_k(K_k)$ such that $F_{w'}=F'$. We call $\phi$ the adjacency preserving bijection from the vertex set of $F'$ to the vertex set of $F$. Then, we write $w' = v_0v_1\cdots v_k$, where $v_0=v_k$, and define the walk $w$ as $w = \phi(v_0)\phi(v_1)\cdots \phi(v_k)$. Finally, we observe that $w\in\mathcal{W}_k(F)$ and by construction,
\begin{align*}
F_w
&=(\{\phi(v_t)\}_{0\leq t\leq k}, \{\phi(v_t)\phi(v_{t+1})\}_{0\leq t< k})\\
&=(\{\phi(x)\}_{x\in v(F')}, \{\phi(x)\phi(y)\}_{xy\in e(F')})\\
&= F.
\end{align*}

\vspace{.5\baselineskip}
\emph{Proof of 2:} 
Fix $F\in W_k$. Fix $w$ such that $F_w=F$. Then, $w$ contains a path between any pair of nodes in $F$. Hence, $F$ is connected.

\vspace{.5\baselineskip}
\emph{Proof of 3:} 
Suppose that there exists $F\in W_k$ such that $|e(F)|<|v(F)|$. Recall that by definition, trees are the only connected graphs with fewer edges than nodes. Hence, since we have already established that $F$ is connected, $F$ must be a tree.

\vspace{.5\baselineskip}
\emph{Proof of 4:} 
Assume there exists a closed $k$-walk $w = v_0v_1\cdots v_{k-1}v_0$ such that $F_w$ is a tree. We will prove the claimed result by contradiction. Assume there exists at least one edge traversed by $w$ exactly once. Let $o$ and $t$ be the nodes corresponding to such an edge. Since the composition of any cyclic permutation and reversal of $w$ will also induce $F_w$, as the walk is closed, we assume without loss of generality that $o$ and $t$ are the two first steps of $w$; i.e., $w = otv_2\dots v_{k-1}o$.

Let $w_1 = ot$ and $w_2 = tv_2\dots v_{k-1}o$. Then, $w_1$ and $w_2$ contain paths between vertices $o$ and $t$, and by assumption neither $ot$ nor $to$ appears in $w_2$. Thus, the paths contained within $w_1$ and $w_2$ are edge disjoint.  However, the existence of two edge-disjoint paths between a pair of vertices in $F_w$ contradicts our assumption that $F_w$ is a tree.  Thus we conclude that every edge in $F_w$ is traversed at least twice by $w$.

\vspace{.5\baselineskip}
\emph{Proof of 5:} 
Fix a tree $T\in W_k$. Fix a closed $k$-walk $w$ such that $F_w=T$. Then $w$ visits all edges of $T$ at least twice, hence $|e(T)|\leq k/2$. Since $|e(T)| = |v(T)|-1$ for any tree $T$, we conclude $|v(T)|-1\leq k/2$.

\vspace{.5\baselineskip}
\emph{Proof of 6:} 
Fix $F\in W_k$. We will show that $F \in W_{k+2}$. Let $w = v_0v_1\cdots v_{k-1}v_0$ be a closed $k$-walk such that $F_w = F$. Then $w' = v_0v_1v_0v_1v_2\cdots v_{k-1}v_0$ is a closed $(k+2)$-walk such that $F_{w'} = F$. Hence $W_k\subset W_{k+2}$.

\vspace{.5\baselineskip}
\emph{Proof of 7:} 
To prove the first point of the claim, note that any closed walk of odd length must contain a cycle (of odd length). Trees contain no cycles, and so if $k$ is odd, then $W_k$ contains no trees.

To prove the second part of the claim, assume that $k$ is even and set $k=2r$. Recall that $\mathcal{T}_r$ denotes the set of unlabeled trees on $r$ vertices. We will show by induction that the statement
\begin{equation*}
P(r) = \{\forall T\in\mathcal{T}_{r+1}, \exists w\in \mathcal{W}_{2r}(T)\textnormal{ s.t. } F_w = T\},
\end{equation*}
is true for all $r\geq1$. This directly yields that $\mathcal{T}_{r+1}\subset W_{2r}$ for any fixed $r\geq1$. Then, noting that $W_{2t}\subset W_{2r}$ for all $t\leq r$, we conclude that $\mathcal{T}_{t+1}\subset W_{2r}$ for all $t\leq r$. Thus, we will have established our claim that all unlabeled trees on up to $k/2+1$ vertices are elements of $W_k$. We now prove that $P(r)$ is true for all $r\geq1$.
\begin{enumerate}
\item Set $r=1$. Then, the only $2$-tree is $K_2$, which is visited by any walk in $\mathcal{W}_{2}(K_2)$. Thus, $P(1)$ is true.
\item Fix $r>1$. Assume $P(r-1)$ is true and fix $T\in\mathcal{T}_{r+1}$. We now build a closed $2r$-walk that induces $T$.

First, choose a leaf $l$ of $T$ and call $o$ the unique node adjacent to $l$. Because any finite tree on at least two nodes possesses at least two leaves, $l$ always exists. Then, let $L$ be the graph given by $L = (\{l,o\},\{lo\})$ and $T'$ be the graph given by $T' = (v(T)\setminus \{l\},e(T)\setminus\{lo\})$. By construction, $T = L\cup T'$, and  $T'$ is an $r$-tree.

Second, consider any closed $2(r-1)$-walk that induces $T'$. (Such a walk exists by Item 1 of the current lemma, because we assume $P(r-1)$ to be true.) Since the composition of any cyclic permutation and reversal of this walk also induces $T'$, without loss of generality we may choose the walk $w'= o v_1\cdots v_{2(r-2)} o$ such that $F_{w'}=T'$.

Third, extend $w'$ to $T$ by defining $w = o l o v_1\cdots v_{2(r-2)} o$, so that $F_w = L\cup T' = T$ and $|w| = |w'|+2 = 2r$.

Thus, for any $(r+1)$-tree $T$, we have exhibited a closed walk $w$ in $\mathcal{W}_{2r}(T)$ such that $F_w = T$. Hence, $P(r)$ is true.
\end{enumerate}
Finally, since $P(r-1)\Rightarrow P(r)$ and $P(1)$ holds, it follows by induction that $P(r)$ is true for all $r\geq1$.

\vspace{.5\baselineskip}
\emph{Proof of 8:} 
The closed walk $w=12\cdots k1$ is in $\mathcal{W}_k(K_k)$ and induces $F_w\equiv C_k$. Hence, $C_k\in W_k$. For a closed $k$-walk to visit $k$ nodes, it must visit a new node at each step (apart from the last). Let $w' = v_1v_2\cdots v_{k}v_1$ be any such closed $k$-walk in $\mathcal{W}_k(K_k)$, where all elements $v_1, v_2, \ldots v_k \in v(K_k)$ are distinct.  Then through the vertex bijection $v_i \mapsto i$, we see that $F_{w'} \equiv F_w$.  Hence, we conclude that $C_k$ is the only element of $W_k$ on $k$ vertices.

\vspace{.5\baselineskip}
\emph{Proof of 9:} 
We show that $C_{k-2}P_2$ is the only element of $W_k$ that is both on $k-1$ vertices and with $k-1$ edges. To begin, consider a closed $k$-walk that induces a $(k-2)$-cycle and then immediately traverses a pendant edge.  In turn, this walk induces a graph isomorphic to $C_{k-2}P_2$, and so $C_{k-2}P_2 \in W_k$. Hence we may now fix some $F\in W_k$ on $k-1$ vertices and with $k-1$ edges, since we have shown at least one such $F$ to exist. We will show that $F \equiv C_{k-2}P_2$ for any choice of $F$.

First, by Item~1 of the current lemma, we may fix a closed $k$-walk $w$ such that $F_w = F$. Then, since $F$ possesses $k-1$ edges while $w$ traverses $k$ edges, we conclude that there is exactly one edge in $F$ traversed twice by $w$. Label this edge $ot$, and assume without loss of generality (since the composition of any cyclic permutation and reversal of $w$ will also induce $F$) that $w = otv_2v_3\cdots v_{k-1}o$.

Next, note that for $ot$ to be traversed twice, one of its vertices must be visited twice. However, since $F$ possesses $k-1$ vertices while $w$ visits $k$ vertices, we see that exactly one vertex in $F$ is visited twice. Without loss of generality, assume this vertex to be $o$, in which case all vertices in $v(F)\setminus\{o\}$ must be visited exactly once.

Finally, observe that for $t$ to be visited exactly once and $ot$ traversed exactly twice, the sequence $oto$ must occur exactly once in $w$. Thus we conclude that $w = otov_3\cdots v_{k-1}o$. Furthermore, by construction neither $o$ nor $t$ can be otherwise visited by $w$.  It follows that $\{o,t\}\cap \{v_3,\ldots,v_{k-1}\} = \emptyset$, and consequently that $|\{v_3,\ldots,v_{k-1}\}| = k-3$.

To complete the proof, consider the graph $F_1$ induced by $oto$ and the graph $F_2$ induced by $ov_3\cdots v_{k-1}o$.  We have $F_w = F = F_1 \cup F_2$, with $\{o,t\}\cap \{v_3,\ldots,v_{k-1}\} = \emptyset$ implying that $v(F_1) \cap v(F_2) = \{o\}$.  We see directly that $F_1 \equiv P_2$.  Furthermore, we see that $F_2 \in W_{k-2}$, with $|\{v_3,\ldots,v_{k-1}\}| = k-3$ implying that $|v(F_2)| = k-2$.  Hence $F_2 \equiv C_{k-2}$. It therefore follows that $F\equiv C_{k-2}P_2$, since we have shown that, up to isomorphism, $F$ can be obtained by joining $C_{k-2}$ to $P_2$ by identifying a single vertex.
\end{proof}

\section{Closed walks and extensions}\label{sec:S1}

\subsection{Extensions and a partial order on walk-induced subgraphs}\label{sec:extensions}

We have introduced the set $W_k$, which organizes walk-induced graphs by scale. We next construct a relation ``$\vartriangleleft_k$'' on each given set $W_k$.

\begin{Definition}[Graph extension ``$\vartriangleleft_k$'' and corresponding walk extension]
\label{def:extens}
Fix an integer $k\geq 2$ and $F,F'\in W_k$. We call $F'$ an {\em extension} of $F$ and write $F\vartriangleleft_k F'$ if:
\begin{enumerate}
	\item $|v(F')|-|v(F)| = 1$,
	\item $|e(F')|-|e(F)| \leq 1$,
	\item $
\min_{\substack{
w = v_0\cdots v_k
	\in\mathcal{W}(K_k)\,:\,F_w\ \equiv F\\
w' = v_0'\cdots v_k'
	\in\mathcal{W}(K_k)\,:\,F_{w'} \equiv F'
}}\ \sum_{i=0}^k 1_{\{v_i \neq v_i'\}} = 1$.
\end{enumerate}
Given two closed $k$-walks $w=v_0v_1\cdots v_{k-1}v_0$ and $w'=v_0'v_1'\cdots v_{k-1}'v_0'$ in a simple graph $G$, we call $w'$ an {\em extension} of $w$ in $G$ if $F_w \vartriangleleft_k F_{w'}$ and $w'$ has a Hamming distance of exactly 1 from $w$.
\end{Definition}

From Item~3 of Definition~\ref{def:extens}, we see directly that whenever $F \vartriangleleft_k F'$ for some pair $F,F'\in W_k$, there exists a pair $w,w'$ of closed $k$-walks in $K_k$ such that $F_w\equiv F$, $F_{w'}\equiv F'$, and $w'$ is an extension of $w$ in $K_k$.

We now determine when any $F\in W_k$ admits an extension.

\begin{Lemma}\label{extens-prop}
Fix an integer $k\geq 2$ and let $F\in W_k$. Then $F$ admits an extension unless either $F\equiv C_k$ when $k\geq 3$, or $F$ is isomorphic to an element of $\mathcal{T}_{k/2+1}$ when $k$ is even and $k\geq 2$. In these latter two cases $F$ admits no extension.

Furthermore, whenever $F$ admits an extension, then at least one such extension $F'$ has the following property: There exist orderings $(d_1,\dots,d_v)$ and $(d_1',\dots,d_{v+1}')$ of the degrees of $F$ and $F'$, respectively, such that exactly one of the following four cases holds:
\begin{enumerate}
\item $d_t' = d_t$ for all $t \leq v-1$, $d_v' = d_v-2$, and $d_{v+1}' = 2$;
\item $d_t' = d_t$ for all $t \leq v-2$, $d_{v-1}' = d_{v-1}-1$, $d_v' = d_v+1$, and $d_{v+1}'=2$;
\item $d_t' = d_t$ for all $t \leq v-2$, $d_{v-1}'\leq d_{v-1}$, $d_v' \leq d_v+1$, and $d_{v+1}'=1$; or 
\item $d_t' = d_t$ for all $t \leq v-1$, $d_v' = d_v-1$, and $d_{v+1}'=1$.
\end{enumerate}
\end{Lemma}

\begin{proof}
We first show that $F\in W_k$ admits no extension if either i) $F\equiv C_k$ when $k\geq 3$; or ii) $F$ is isomorphic to an element of $\mathcal{T}_{k/2+1}$ when $k\geq 2$ is even.

First, consider i) so that $F\equiv C_k$. Then it follows that $|v(F)| = k$. Assume there exists $F' \in W_k$ such that $F\vartriangleleft_k F'$. Then $ |v(F')| = |v(F)| + 1 $ by Definition~\ref{def:extens}.  But no closed $k$-walk can visit more than $k$ nodes, and since by construction every element of $W_k$ is induced by some closed $k$-walk, no element of $W_k$ has more than $k$ nodes.  Thus we obtain a contradiction, and so conclude that $F\in W_k$ admits no extension if $F\equiv C_k$ when $k\geq 3$.

Second, consider ii) so that $F$ is isomorphic to some element of $\mathcal{T}_{k/2+1}$, implying that $|v(F)| = k/2+1$ and $|e(F)| = k/2$. Assume there exists $F' \in W_k$ such that $F\vartriangleleft_k F'$. By Item~2 of Lemma~\ref{wk-props}, any $F' \in W_k$ is necessarily a connected graph. Since $ |v(F')| = |v(F)| + 1 = k/2 + 2 $ by Definition~\ref{def:extens}, it follows that in order to be connected, $F'$ must necessarily have $k/2+1$ edges.  This implies that $F'$ is a tree with $|e(F')|=|e(F)|+1=k/2+1$.  By Item~5 of Lemma~\ref{wk-props}, however, any tree in $W_k$ contains no more than $ k / 2 $ edges.  Thus we obtain a contradiction, and so conclude that $F\in W_k$ admits no extension if $F$ is isomorphic to an element of $\mathcal{T}_{k/2+1}$ when $k\geq 2$ is even.

We next construct an extension of any $F\in W_k$ that is i) not isomorphic to $C_k$ when $k\geq 3$; and ii) not isomorphic to an element of $\mathcal{T}_{k/2+1}$ when $k\geq 2$ is even. We consider two mutually exclusive and exhaustive cases, depending on whether or not a closed walk inducing $F$ traverses every edge in $e(F)$ at least twice:
\begin{description}
\item[Claim~1] If for $k\geq 3$ there exists a closed $k$-walk $w$ inducing $F\in W_k$ which traverses at least one edge in $e(F)$ exactly once, then $F$ admits an extension whenever $F\not\equiv C_k$.
\item[Claim~2] If for $k\geq 2$ there exists a closed $k$-walk $w$ inducing $F\in W_k$ which traverses every edge in $e(F)$ at least twice, then $F$ admits an extension whenever $k$ is odd, or whenever $k$ is even and $F$ is not isomorphic to an element of $\mathcal{T}_{k/2+1}$.
\end{description}

For each of Claim~1 and Claim~2 in turn, we will exhibit an $F'\in W_k$ such that $F\vartriangleleft_k F'$. To construct $F' \subset K_k$ as a labeled graph, we must first ensure the existence of an additional node in $K_k$ to be visited by a closed $k$-walk $w'$ that is a candidate extension of a closed $k$-walk inducing $F\in W_k$. To do so we must exclude any $F$ with $|v(F)| = k$.  By Item~8 of Lemma~\ref{wk-props}, the only element of $W_k$ on $k$ nodes is $C_k$. We have shown that if $F\equiv C_k$, then $F$ admits no extension.  In Claim~1, we assume $F\not\equiv C_k$. In Claim~2, we assume all edges in $F$ are traversed at least twice, so $|e(F)|\leq \lfloor k/2 \rfloor$ and hence $F\not\equiv C_k$. Thus for both claims, clearly it follows that $F\not\equiv C_k$, which in turn implies $|v(F)|<k$.  Therefore, we can fix a new node $l\in v(K_k)\setminus v(F)$ to be visited by a closed $k$-walk $w'$ that is a candidate extension of a closed $k$-walk inducing $F\in W_k$. 

Furthermore, note that since the composition of any cyclic permutation and reversal of any walk inducing $F\in W_k$ will also induce $F$, we will choose from among all such walks at our convenience throughout this proof, always without loss of generality.

\begin{figure}[t]
\hspace{-.05\textwidth}{\includegraphics[width=1.15\textwidth]{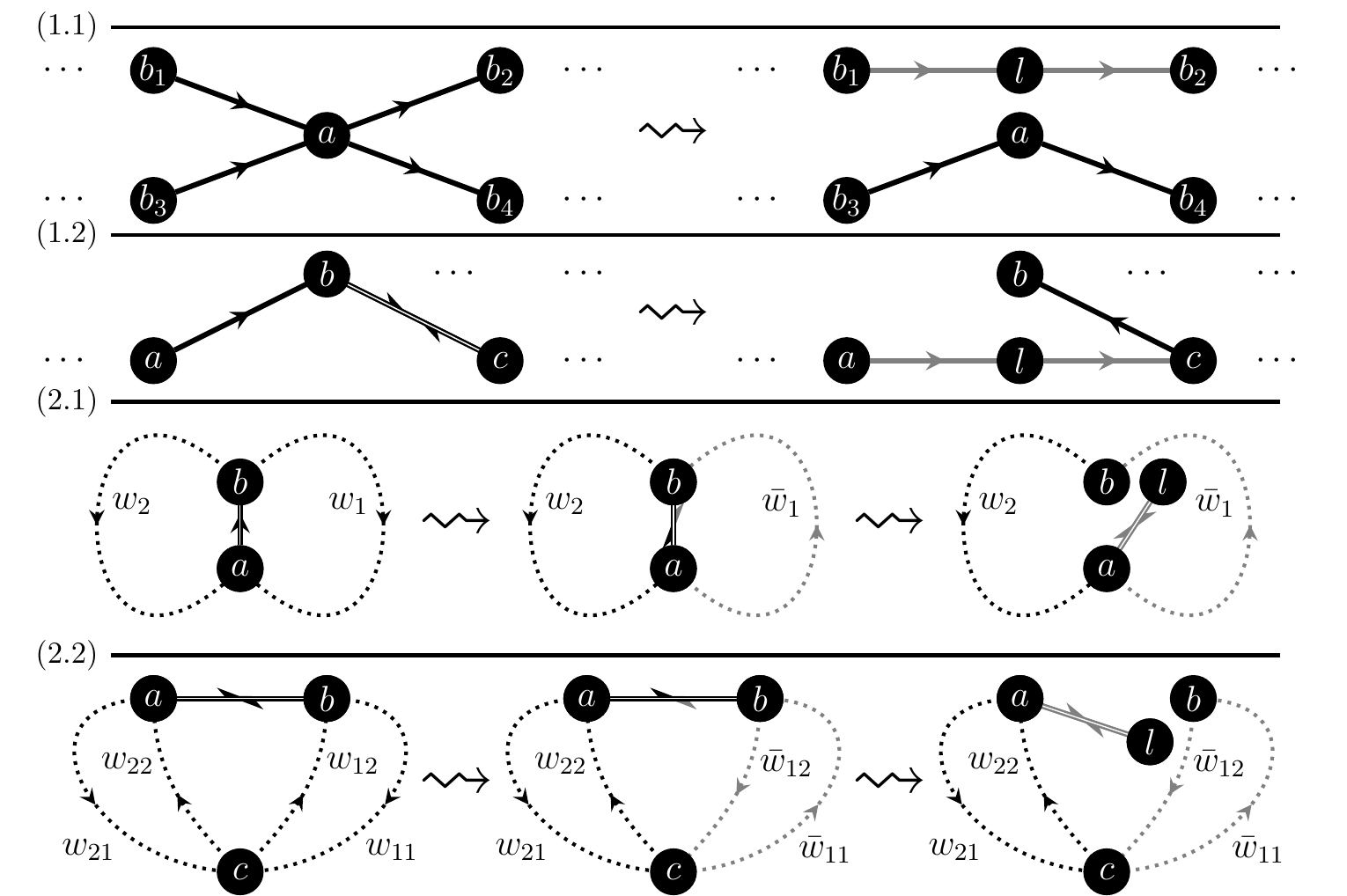}}%
\caption{\label{extens-prop-fig}Diagram showing the four different walk modifications used to prove Lemma~\ref{extens-prop}. Note that new walk steps are marked gray, and arrows indicate the direction of walk travel. Split arrows indicate that both directions are traveled.}
\end{figure}

\vspace{.5\baselineskip}
\emph{Proof of Claim 1:} 
We first consider $k\geq 5$, treating the cases $k=3$ and $k=4$ separately below. Fix $F\in W_k\setminus\{C_k\}$ such that there exists a closed $k$-walk inducing $F\in W_k$ which traverses at least one edge in $e(F)$ exactly once; for example, by Item~9 of Lemma~\ref{wk-props}, we may choose the $(k-2,1)$-tadpole $F \equiv C_{k-2}P_2$ for any $k \geq 5$. To exhibit an extension $F'\in W_k$ of $F$, we consider two mutually exclusive and exhaustive cases for $k \geq 5$:
\begin{enumerate}
\item Assume that $k \geq 5$, and that all edges in $e(F)$ are traversed exactly once by some closed $k$-walk inducing $F$. Therefore this walk is an Eulerian circuit, and so all vertices in $F$ must have even degrees. If all vertices of $F$ were to have degree two, then we would conclude $F\equiv C_k$, since $C_k$ is the only connected $2$-regular graph on $k$ vertices. However, since $F\not\equiv C_k$ by hypothesis, we immediately conclude that $F$ must possess at least one vertex of even degree at least four.

Pick any such vertex of even degree $d \geq 4$ and label it $a$. Since we assume the existence of a closed $k$-walk inducing $F$ with the property that every edge is traversed exactly once, we may assign labels $b_1, b_2, \ldots, b_d$ to all vertices connected to $a$ such that the segments $b_1ab_2$ and $b_3ab_4$ appear exactly once in this walk, with $b_1ab_2$ preceding $b_3ab_4$. The first diagram of Fig.~\ref{extens-prop-fig}.(1.1) depicts an example for $d=4$. Therefore (after an appropriate cyclic permutation if necessary to begin the walk at $b_1$) we can write the assumed closed $k$-walk inducing $F$ as $w = b_1ab_2\cdots b_3ab_4\cdots b_1$.

Now, define another closed $k$-walk $w' = b_1lb_2\cdots b_3ab_4\cdots b_1$, and refer to the second diagram of Fig.~\ref{extens-prop-fig}.(1.1). By construction, Items~1 and~2 of Definition~\ref{extens-prop} are verified for $F_w$ and $F_{w'}$, since $|v(F_{w'})| = |v(F_w)| + 1$ and $|e(F_{w'})| = |e(F_w)|$. To verify Item~3 of Definition~\ref{extens-prop}, note that $w$ and $w'$ disagree in exactly one entry. Furthermore, no other walks inducing $F$ and $F_{w'}$ can disagree in fewer entries, since by Item~1 we have that $v(F) \neq v(F_{w'})$. Therefore, $F\vartriangleleft_k F_{w'}$.

\item Assume that $k \geq 5$, and that at least one edge in $e(F)$ is traversed more than once by some closed $k$-walk inducing $F$. Then at some point in this walk, an edge traversed exactly once and an edge traversed more than once must be traversed in immediate succession. Assume, reversing the walk if necessary, that the former edge precedes the latter edge. Then, calling the edge traversed exactly once $ab$, and the edge traversed more than once $bc$, we may (after an appropriate cyclic permutation if necessary to begin the walk at $a$) write the assumed closed $k$-walk inducing $F$ as $w = abc\cdots a$. An example is depicted in the first diagram of Fig.~\ref{extens-prop-fig}.(1.2).

    Now, define the closed $k$-walk $w' = alc\cdots a$, noting that $|v(F_{w'})| = |v(F)| + 1$ since the edge $bc$ is traversed more than once by hypothesis and thus $b \in v(F_{w'})$, and refer to the second diagram of Fig.~\ref{extens-prop-fig}.(1.2). Observe that since $ab$ is traversed only one time by $w$, it follows that $ab\not\in F_{w'}$; in contrast, since $bc$ is traversed multiple times by $w$, it follows that $bc\in e(F_{w'})$. Thus, $|e(F_{w'})|  = |(e(F)\setminus\{ab\})\cup \{al,lc\}| = |e(F)| + 1$. Therefore, Items~1 and~2 of Definition~\ref{extens-prop} are verified for $F$ and $F_{w'}$. To verify Item~3, observe as before that $w$ and $w'$ disagree in exactly one entry.
\end{enumerate}

It remains to consider $k=3$ and $k=4$.  For $k=3$, we have $W_3\setminus\{C_3\} = \emptyset$. For $k=4$, we have $W_4\setminus\{C_4\} = \{ P_2, P_3\}$. No graphs of the form $ F \equiv P_2$ or $ F \equiv P_3$ may be induced by a closed $4$-walk which traverses at least one edge in $e(F)$ exactly once. Thus we have exhibited an extension of $F$ under the hypothesis of Claim~1.

\vspace{.5\baselineskip}
\emph{Proof of Claim 2:} 
We first consider the cases $k=4$, $k=6$, and $k\geq 8$, treating the cases $k=2$, $k=3$, $k=5$ and $k=7$ separately below. Fix $F\in W_k\setminus\{C_k\}$ if $k$ is odd, or $F\in W_k\setminus(\{C_k\}\cup\mathcal{T}_{k/2+1})$ if $k$ is even, such that there exists a closed $k$-walk traversing every edge in $e(F)$ at least twice. To show that such an element exists, note that if $k$ is even, we may choose $F\equiv P_2$, while if $k$ is odd we may choose $F\equiv C_3$, with a walk traversing $C_3$ three times and then repeatedly visiting an adjacent node a total of $(k-9)/2$ times.

To exhibit an extension $F'\in W_k$ of $F$ for $k=4$, $k=6$, or $k\geq 8$, we consider two mutually exclusive and exhaustive cases:
\begin{enumerate}
\item We let $k=4$, $k=6$, or $k\geq 8$, and assume that at least one edge in $e(F)$ is traversed twice in the same direction by some closed $k$-walk inducing $F$.  Then there exists an edge $ab$ such that we can (after an appropriate cyclic permutation of the walk if necessary)  begin the walk at $a$.  Then, letting $w_1$ and $w_2$ be walk segments (sequences of adjacent nodes, possibly of length zero), we may write $w = ab w_1 ab w_2 a$. The first diagram of Fig.~\ref{extens-prop-fig}.(2.1) provides an illustration of this scenario.
    
    Define $\bar w_1$ to be $w_1$ in reversed order, such that if $w_1 = w_{11}w_{12}\cdots w_{1|w_1|}$, then $\bar w_1 = w_{1|w_1|}w_{1(|w_1|-1)}\cdots w_{11}$; if $w_1 = w_{11}$ then $\bar w_1 = w_{11}$; and if $w_1$ is empty then $\bar w_1$ is in turn empty.
    
    Now consider the closed $k$-walk $w' = aba \bar w_1 b w_2 a$. We note that $F_{w'}=F_w$, and refer to the second diagram of Fig.~\ref{extens-prop-fig}.(2.1). Define $w'' = ala \bar w_1 b w_2 a$, and note that Items~1 and~2 of Definition~\ref{extens-prop} are verified for $F=F_{w'}$ and $F_{w''}$ by construction (refer to the third diagram of Fig.~\ref{extens-prop-fig}.(2.1)). To verify Item~3 of Definition~\ref{extens-prop}, we observe (as before) that $w'$ and $w''$ disagree in exactly one entry. Therefore, we have exhibited $w''$ such that $F\vartriangleleft_k F_{w''}$.

\item We let $k=4$, $k=6$, or $k\geq 8$, and assume that no edge in $e(F)$ is traversed twice in the same direction by some closed $k$-walk inducing $F$. No element of $W_4$ verifies this assumption, since $W_4\setminus(\{C_4\}\cup\mathcal{T}_3)=\{P_2\}$, and $P_2$ is induced only by walks traversing all edges in both directions. Therefore we may assume that $k=6$ or $k\geq 8$.
    
    In this setting, all edges in $F$ are traversed at least twice but none twice in the same direction. Thus all the edges in $e(F)$ must be traversed by $w$ exactly once in each direction. This means all edges are traversed exactly twice, and importantly, it therefore follows that $k$ is even and $|e(F)|=k/2$.
    
    We begin by showing (by contradiction) that $F$ is not a tree. If $F$ were a tree, then $F$ would be a tree over $k/2$ edges, and therefore $F$ would be isomorphic to an element of $\mathcal{T}_{k/2+1}$. Since we have already assumed that $F$ is not isomorphic to an element of $\mathcal{T}_{k/2+1}$, this cannot hold.
    
    We now use that $F$ is not a tree to build a walk $w'$ inducing $F$. First, since $F$ is connected and not a tree, $F$ contains at least one cycle. We fix $C$ to be any cycle in $F$ and fix $ab$ as any edge in $e(C)$. Then, we can (after an appropriate cyclic permutation if necessary), begin the walk at $a$ and can write $w = ab w_1 ba w_2 a$, where $w_1$ and $w_2$ are walk segments. Either $w_1$ or $w_2$ could be empty, but since $k\geq 6$, at least one of the two segments is not empty. We distinguish two mutually exclusive and exhaustive cases:
    \begin{enumerate}
    \item If neither $w_1$ nor $w_2$ is the empty walk, for $C$ to be a subgraph of $F_w$, $w_1$ and $w_2$ must cross over at least one vertex; i.e., must both visit at least one common vertex (see Fig.~\ref{extens-prop-fig}.(2.2)). To prove this, we proceed by contradiction. If $w_1$ and $w_2$ do not cross, then either $w_1$ or $w_2$ must induce $C$, and since $ab$ is an edge of $C$, $w_1$ or $w_2$ must traverse $ab$ and therefore $ab$ is traversed thrice by $w$. However, this contradicts the assumption that $ab$ is traversed exactly twice. Therefore, call without loss of generality $c$ any of the possible vertices where $w_1$ and $w_2$ cross. Then, writing $w_1 = w_{11} c w_{12}$ and $w_2 = w_{21} c w_{22}$, we have $w = ab w_{11} c w_{12} ba w_{21} c w_{22} a$ (refer to the first diagram in Fig.~\ref{extens-prop-fig}.(2.2) and note that $w_{11}$, $w_{12}$, $w_{21}$ and $w_{22}$ could be empty). Define $\bar w_{11}$ and $\bar w_{12}$ to be the walks $w_{11}$ and $w_{12}$ in reversed order and set $w' = aba w_{21} c \bar w_{11} b  \bar w_{12} c w_{22} a$.
    \item If $w_1$ (resp. $w_2$) is the empty walk, then $w = abaw_2a$ (resp. $w = abw_1ba$) and $w_2$ must visit $b$ (resp. $a$) for $ab$ to be an edge of $C$ (and satisfy the constraint of every edge being visited at least twice). We then write $w_2=w_{21}b w_{22}$, and thus we may set $w'=aba w_{21} b w_{22} a$ (resp. $w' = aba \bar w_{11} b  \bar w_{12} a$). Note a slight asymmetry in the proof, namely that $w\equiv w'$ if $w_1$ is the empty walk, while $w'$ had been redefined if $w_2$ is the empty walk.
    \end{enumerate}
    
To simplify notation, in these two cases, we write $w'  = aba\tilde w$. Now we observe that: i) $w'$ is a closed walk of the same length as $w$ and; ii) $w'$ is such that $F_{w'}=F_w$ (refer to the second diagram in Fig.~\ref{extens-prop-fig}.(2.2)). Finally, call $w'' = ala \tilde w$ and note that Items~1 and~2 of Definition~\ref{extens-prop} are verified for $F=F_{w'}$ and $F_{w''}$ by construction (as shown in the third diagram in Fig.~\ref{extens-prop-fig}.(2.2)). To verify Item~3 of Definition~\ref{extens-prop}, we observe once again that $w'$ and $w''$ disagree in exactly one entry. Therefore $F\vartriangleleft_k F_{w''}$.
\end{enumerate}

We now treat the remaining cases of $k=2$, $k=3$, $k=5$, and $k=7$. If $k=2$, then $P_2$ is the only possible walk-induced subgraph. Then $W_2\setminus\mathcal{T}_{2}=\emptyset$, and so the claim trivially holds. If $k=3$, $k=5$, or $k=7$, then no closed $k$-walk $w$ can traverse every edge in the edge set of $F$ at least twice. This is because every closed odd walk must contain an odd cycle, and $k$ is not large enough to traverse every edge in the the smallest odd cycle $C_3$ at least twice.

Thus, whenever $F\in W_k\setminus(\{C_k\}\cup\mathcal{T}_{k/2+1})$ if $k$ is even, or $F\in W_k\setminus\{C_k\}$ if $k$ is odd, we have exhibited an extension of $F$ under the hypothesis of Claim~2.

\vspace{.5\baselineskip}
Having shown the first part of the lemma, we now show that its second statement is a consequence of the constructive proof above. We fix $F\in W_k$ and assume that $F$ admits at least one extension, implying that $k\geq 4$. Thus we may appeal to the four methods illustrated in Fig.~\ref{extens-prop-fig} to construct an extension $F'$ of $F$. As we now show, these four methods enable us to directly verify that there exists orderings $(d_1,\dots,d_v)$ and $(d_1',\dots,d_{v+1}')$ of the degrees of $F$ and $F'$, respectively, such that exactly one of the four cases stated in the lemma holds. We write the degree of the node labeled $x$ in $F$ as $d_x$, and that of the node labeled $y$ in $F'$ as $d_y'$.

\begin{enumerate}
\item We begin by considering $F'$ as constructed in Item~1 of Claim~1 in the proof above, referring to Fig.~\ref{extens-prop-fig}.(1.1) for an illustration. With the labeling of Fig.~\ref{extens-prop-fig}.(1.1), we observe directly that $d_l'=2$ and $d_a'=d_a-2$, while $d_x = d_x'$ for $x\in v(F)\setminus \{a\}$. It follows that if we relabel the vertex labeled $l$ as $v+1$, the vertex labeled $a$ as $v$, and the other vertices arbitrarily, then $F'$ obeys the first case stated in the lemma: $d_t' = d_t$ for all $t \leq v-1$, $d_v' = d_v-2$, and $d_{v+1}' = 2$.
\item We then consider $F'$ as constructed in Item~2 of Claim~1 in the proof above. Referring to Fig.~\ref{extens-prop-fig}.(1.2), we observe directly that $d_l=2$, $d_c' = d_c+1$, and $d_b'= d_b-1$ (because the edge $ab$ is traversed exactly once, and so there is no edge between $a$ and $b$ in $F'$, reducing the degree of $b$ by 1), while $d_x = d_x'$ for $x\in v(F)\setminus \{b,c\}$. It follows that if we relabel the vertex labeled $l$ as $v+1$, the vertex labeled $c$ as $v$, the vertex labeled $b$ as $v-1$, and the other vertices arbitrarily, then $F'$ obeys the second case stated in the lemma: $d_t' = d_t$ for all $t \leq v-2$, $d_{v-1}' = d_{v-1}-1$, $d_v' = d_v+1$, and $d_{v+1}'=2$.
\item We next consider $F'$ as constructed in Item~1 of Claim~2, referring here to Fig.~\ref{extens-prop-fig}.(2.1) to illustrate the change in the degree sequence of $F'$ relative to $F$. We observe directly that $d_l=1$, $d_a' \leq d_a+1$, and $d_b'\leq d_b$, while $d_x = d_x'$ for $x\in v(F)\setminus \{a,b\}$. It follows that if we relabel the vertex labeled $l$ as $v+1$, the vertex labeled $a$ as $v$, the vertex labeled $b$ as $v-1$, and the other vertices arbitrarily, then $F'$ obeys the third case stated in the lemma: $d_t' = d_t$ for all $t \leq v-2$, $d_{v-1}'\leq d_{v-1}$, $d_v' \leq d_v+1$, and $d_{v+1}'=1$.
\item Finally, we consider $F'$ as constructed in Item~2 of Claim~2, referring to Fig.~\ref{extens-prop-fig}.(2.2)) to illustrate how the degrees of $F$ and $F'$ are respectively modified. We observe directly that $d_l=1$, $d_a' = d_a$ and $d_b'=d_b-1$ (because the edge $ab$ is traversed exactly twice), while $d_x = d_x'$ for $x\in v(F)\setminus \{a,b\}$. It follows that if we relabel the vertex labeled $l$ as $v+1$, the vertex labeled $b$ as $v$, and the other vertices arbitrarily, then $F'$ obeys the fourth and final case stated in the lemma: $d_t' = d_t$ for all $t \leq v-1$, $d_v' = d_v-1$, and $d_{v+1}'=1$.\qedhere
\end{enumerate}
\end{proof}

\begin{Remark}
As defined above, ``$\vartriangleleft_k$'' is a binary relation over $W_k$. By Items~1 and~2 of Definition~\ref{def:extens}, the edge density $ |e(\cdot)| / |v(\cdot)| $ is non-increasing in ``$\vartriangleleft_k$,'' while the Euler characteristic $ |v(\cdot)| - |e(\cdot)| $ is non-decreasing in ``$\vartriangleleft_k$.'' Furthermore, ``$\vartriangleleft_k$'' naturally induces a finer partial ordering than the edge density or the Euler characteristic. To see this, consider the directed graph $\mathfrak{W}_k=(W_k,\{(F,F') : F\vartriangleleft_k F'\})$. This graph can be extended into a partial ordering of $W_k$ by writing $F\trianglelefteq_k F'$ if there exists a directed path between $F$ and $F'$ in $\mathfrak{W}_k$; i.e., if $F=F'$, $F\vartriangleleft_k F'$, or there exists $p\geq 1$ and $F_1,\dots, F_p\in W_k$ such that
$F\vartriangleleft_k F_1 \vartriangleleft_k \cdots \vartriangleleft_k F_p \vartriangleleft_k F'$.

With this definition, $(W_k,\trianglelefteq_k)$ is a poset; i.e., a partially ordered set. We see directly that if $k\geq 2$ is even, then a minimal element of $W_k$ according to this partial ordering is $K_2$, while if $k\geq 3$ is odd, a minimal element is $C_3$.  This follows because both $K_2$ and $C_3$ have in-degrees of zero when they are present in $\mathfrak{W}_k$.
A direct consequence of Lemma~\ref{extens-prop} is that if $k\geq 3$ is odd, then $C_k$ is the greatest element of $(W_k,\trianglelefteq_k)$, while if $k\geq 4$ is even, then $\{C_k\}\cup\mathcal{T}_{k/2+1}$ is the set of maximal elements of $(W_k,\trianglelefteq_k)$.
\end{Remark}

\subsection{\!Counting walks that induce isomorphic copies of subgraphs}%

The set $W_k$ and the isomorphism relation ``$\equiv$'' allow us to partition the set of all closed $k$-walks $\mathcal{W}_k(G)$ in any simple graph $G$. Letting $\sqcup$ denote the disjoint union, we have directly from Definition~\ref{closedwalks} that
\begin{equation}
\label{eq:walk-partition}
\mathcal{W}_k(G) = \bigsqcup_{F \in W_k} \left\{ w\in\mathcal{W}_k(G) : F_w \equiv F \right\} .
\end{equation}
With this partition it is natural to count the number of closed $k$-walks in $G$ that induce a given unlabeled graph $F$.

\begin{Definition}[Number $\mathrm{ind}_k(F,G)$ of closed $k$-walks in $G$ inducing an isomorphic copy of $F$]\label{defnocop}
Fix a walk length $k$ and two simple graphs $F,G$. We write
\begin{equation*}
\mathrm{ind}_k(F,G) = \left| \left\{ w\in\mathcal{W}_k(G) : F_w \equiv F \right\} \right|
\end{equation*}
for the number of closed $k$-walks in $G$ inducing a subgraph of $G$ isomorphic to $F$.
\end{Definition}

It follows from~\eqref{eq:walk-partition} that for any $k\in \mathbb{N}$,
\begin{equation*}
\Tr{A(G)^k}
	= \left| \mathcal{W}_k(G) \right|
	= \sum_{F \in W_k} \mathrm{ind}_k(F,G) .
\end{equation*}
Next we relate $\mathrm{ind}_k(F,G)$ to $X_F(G) = \left| \left\{ F' \! \subset \! G : F'\equiv F \right\} \right|$, the number of isomorphic copies of $F$ in $G$.

\begin{Lemma}\label{nb-copies}
Fix a walk length $k$ and two simple graphs $F,G$. Then
\begin{equation*}
\mathrm{ind}_k(F,G)=\mathrm{ind}_k(F,F) \, X_F(G).
\end{equation*}
\end{Lemma}

\begin{proof}
Recall from Definition~\ref{closedwalks} that $\mathcal{W}_k(G)$ is the set of all closed $k$-walks in $G$. By Definition~\ref{defnocop},
\begin{align*}
\mathrm{ind}_k(F,G)
&=\sum_{w\in \mathcal{W}_k(G)}1_{\{F_w\equiv F\}}\\
&=\sum_{F'\subset G \,:\, F'\equiv F} \,\,\, \sum_{w\in \mathcal{W}_k(F')}1_{\{F_w\equiv F\}}\\
&=\sum_{F'\subset G \,:\, F'\equiv F} \,\,\, \mathrm{ind}_k(F,F')\\
&=\mathrm{ind}_k(F,F) \cdot | \{ F'\subset G : F'\equiv F \} |,
\end{align*}
since each $F'$ is isomorphic to $F$. The relation $\left| \left\{ F'\subset G : F'\equiv F \right\} \right|=X_F(G)$ completes the proof.
\end{proof}

\begin{Remark}
The growth rate of $ \mathrm{ind}_k(F,F) $ with $k$ can vary substantially, depending on the structure of $F$. For example, we have $ \mathrm{ind}_k(C_k,C_k) = 2k = \mathrm{aut}(C_k) $, whereas $ \mathrm{ind}_k(F,F) = 2(k/2)! = 2 \mathrm{aut}(F) $ for the $(k/2+1)$-star $F = K_{1,k/2}$. The structure of $F$, or indeed the parity of $k$, can also imply that $ \mathrm{ind}_k(F,F) = 0 $. For example, since every closed walk of odd length contains an odd cycle, $ \mathrm{ind}_k(F,F) = 0 $ whenever $F$ is a tree and $k$ is odd.  More generally, independently of the parity of $k$, $ \mathrm{ind}_k(F,F) = 0 $ if $F$ is not connected, or if $ |v(F)|$ or $ |e(F)| $ exceeds $k$.
\end{Remark}

Lemma~\ref{nb-copies} shows how $\mathrm{ind}_k(F_w,G)$, the number of copies of $w$ in $G$, relates to $X_{F_w}(G)$.  Similarly, we may count the number of copies of $w$ in any vertex-induced subgraph of $G$.  Averaging over all such subgraphs of fixed order leads to the notion of a graph walk density.

\begin{Definition}[Graph walk density $\varphi(w,G)$]\label{varphi}
Fix an integer $k\geq 2$, a closed $k$-walk $w\in\mathcal{W}_k(K_k)$, a graph $G$, and a neighborhood size $|u|$ such that $|v(F_w)| \leq |u| \leq |v(G)|$. We let $u\subset v(G)$ be chosen uniformly at random from amongst all $|u|$-subsets of $v(G)$ and take the expectation with respect to the randomized choice of subset $u$.
Then we define the graph walk density of $w$ in $G$ to be
\begin{equation*}
\varphi(w,G)
= \frac{\E \mathrm{ind}_{|w|}(F_w,G[u])}{\mathrm{ind}_{|w|}(F_w,K_{|u|})}
= \frac{\E X_{F_w}(G[u])}{X_{F_w}(K_{|u|})},
\end{equation*}
where $K_{|u|}$ is the complete graph on $|u|$ vertices, and $G[u] = (u, e(G)\cap u^{_{(2)}})$ is the subgraph of $G$ induced by $u \subset v(G)$, with $u^{_{(2)}}$ the set of all unordered pairs of elements of $u$.
\end{Definition}

We see that $\varphi(w,G)$ takes values between $0$ and $1$, and so we are justified in referring to $\varphi(w,G)$ as a density. Furthermore, as we shall show, $\varphi(w,G)$ is independent of the choice of neighborhood size $|u|$. Setting $|u| = |v(G)|$ immediately allows us to recognize the walk density $\varphi(w,G)$ as a graph homomorphism density~\citep{bollobas2007phase,lovasz2012large}.

\begin{Lemma}\label{lem:sampling}
Let $F,G$ be finite graphs with $|v(F)| \leq |v(G)|$, and fix $1 \leq |u| \leq |v(G)|$. Then for $u \subset v(G)$ chosen uniformly at random from among all $|u|$-subsets of $v(G)$,
\begin{equation*}
\E X_F(G[u])
=  \frac{ X_F(K_{|u|}) }{ X_F(K_{|v(G)|}) } X_F(G).
\end{equation*}
\end{Lemma}

\begin{proof}
For notational convenience, let $n = |v(G)|$ and $s = |u|$. We begin by expressing the expectation sum associated to $\E X_F(G[u])$ directly:
\begin{align}
\sum_{u\subset v(G) \,:\, |u|=s} X_F(G[u])
  & = \sum_{u\subset v(G) \,:\, |u|=s} \,\, \sum_{F'\subset G[u]} 1_{\{F'\equiv F \}} \nonumber
\\ & = \sum_{u\subset v(G) \,:\, |u|=s} \,\, \sum_{F'\subset G[u]} \,\, \sum_{i=1}^{X_F(G)} 1_{\{F'=F_i \}}, \label{eq:three-sums}
\end{align}
where $F_1,\dots,F_{X_F(G)}$ enumerate all copies of $F$ in $G$.  Rearranging the order of these three sums, each of which is finite for finite $G$, and subsequently fixing $i$ and $u$ to focus on the summation in $F'$, we have
\begin{align}
\sum_{F'\subset G[u]} 1_{\{F'=F_i \}}
   & = \sum_{u' \subset v(G[u])} \,\, \sum_{t'\subset e(G[u])\cap u'^{(2)} } 1_{\{(u',t')=F_i \}} \nonumber
\\ & = \sum_{u' \subset u} \,\, \sum_{t'\subset e(G[u'])} 1_{\{(u',t')=F_i \}} \nonumber
\\ & = \sum_{u' \subset u} 1_{\{u'=v(F_i) \}} \sum_{t'\subset e(G[u'])} 1_{\{t'=e(F_i) \}} \nonumber
\\ & = \sum_{v(F_i) \subset u} \,\, \sum_{t'\subset e(G[v(F_i)])} 1_{\{t'=e(F_i) \}} \nonumber
\\ & = \sum_{v(F_i) \subset u} 1, \label{eq:inner-sum}
\end{align}
since $e(F_i) \subset e(G[v(F_i)])$.  We next count how many $s$-subsets of $v(G)$ contain each $v(F_i)$ in turn, whence from~\eqref{eq:three-sums} and~\eqref{eq:inner-sum} we obtain
\begin{equation*}
\sum_{i=1}^{X_F(G)} \,\, \sum_{u\subset v(G) \,:\, |u|=s} \,\, \sum_{v(F_i) \subset u} 1
= \sum_{i=1}^{X_F(G)} \binom{n-|v(F_i)|}{s-|v(F_i)|}.
\end{equation*}
Since $|v(F_i)| = |v(F)|$ for all $i$, we obtain the intermediate result that
\begin{equation*}
\sum_{u\subset v(G) \,:\, |u|=s} X_F(G[u])
= \binom{n-|v(F)|}{s-|v(F)|} X_F(G).
\end{equation*}
Now, if we select $u$ with probability $1 / \tbinom{n}{s} $, the claimed result follows immediately upon noting that $ X_F(K_m) = (m)_{|v(F)|} / \mathrm{aut}(F) $  for any positive integer $m$, and thus $ \tbinom{n-|v(F)|}{s-|v(F)|}  / \tbinom{n}{s}  = X_F(K_s) / X_F(K_n) $.
\end{proof}

A direct consequence of Lemma~\ref{lem:sampling} is that for a closed $k$-walk $w$ in $K_k$, a fixed $G$, and subsets $t,u\subset v(G)$ of sizes $|t|,|u|$, each chosen uniformly at random and with $|v(F_w)| \leq |t|,|u| \leq |v(G)| $, we have
\begin{equation*}
\frac{\E X_{F_w}(G[u])}{X_{F_w}(K_{|u|})}
=  \frac{\E X_{F_w}(G[t])}{X_{F_w}(K_{|t|})}.
\end{equation*}
Thus, using Lemma~\ref{nb-copies}, we find that
\begin{equation*}
\varphi(w,G) = \frac{\E\mathrm{ind}_k(F_w,G[u])}{\mathrm{ind}_k(F_w,K_{|u|})}
=  \frac{\E \mathrm{ind}_k(F_w,G[t])}{\mathrm{ind}_k(F_w,K_{|t|})},
\end{equation*}
verifying that $\varphi(w,G)$ does not depend on the choice of neighborhood size $|u|$.

\begin{Remark}\label{remark:extension}
Fix a graph $G$ and two closed $k$-walks $w$ and $w'$ where $F_w\vartriangleleft_k F_{w'}$. In this setting, $\varphi(w',G)/\varphi(w,G)$ governs $\mathrm{ind}_k(F_w,G[u])/\mathrm{ind}_k(F_{w'},G[u])$, and therefore determines which of $w$ and $w'$ has a larger number of copies. To see this, observe that the number of copies of $w'$ over the number of copies of $w$ takes the form:
\begin{align*}
\frac{\mathrm{ind}_k(F_{w'},G)}{\mathrm{ind}_k(F_w,G)}
&= \frac{\mathrm{ind}_k(F_{w'},K_n)}{\mathrm{ind}_k(F_w,K_n)}
\frac{\mathrm{ind}_k(F_{w'},G) / \mathrm{ind}_k(F_{w'},K_n)}{\mathrm{ind}_k(F_w,G) / \mathrm{ind}_k(F_w,K_n)}\\
&= \frac{\mathrm{ind}_k(F_{w'},K_n)}{\mathrm{ind}_k(F_w,K_n)}\frac{\varphi(w',G)}{\varphi(w,G)}.
\end{align*}
Since for $F\subset G\subset K_n$, $\mathrm{ind}_k(F,K_n) = (n)_{|v(F)|}\mathrm{ind}_k(F,F)/\mathrm{aut}(F)$, and as $|v(F_{w'})| = |v(F_w)|+1$ whenever $F_w\vartriangleleft_k F_{w'}$, we recover
\begin{align*}
\frac{\mathrm{ind}_k(F_{w'},G)}{\mathrm{ind}_k(F_w,G)}
&= \frac{(n)_{|v(F_{w'})|}{\mathrm{ind}_k(F_{w'},F_{w'})}/{\mathrm{aut}(F_{w'})}}{(n)_{|v(F_w)|}{\mathrm{ind}_k(F_w,F_w)}/{\mathrm{aut}(F_w)}}
\frac{\varphi(w',G)}{\varphi(w,G)}\\
&= (n-|v(F_{w'})|)\frac{\mathrm{ind}_k(F_{w'},F_{w'})}{\mathrm{ind}_k(F_w,F_w)}\frac{\mathrm{aut}(F_w)}{\mathrm{aut}(F_{w'})}\frac{\varphi(w',G)}{\varphi(w,G)},
\end{align*}
and the asymptotic behavior of the number of copies of an extension $w'$ over the number of copies of the original walk $w$ is the same as that of $n{\varphi(w',G)}/{\varphi(w,G)}$.
\end{Remark}
%
%
\section{Determining dominating closed walks}\label{sec:S2}
\subsection{Defining asymptotically dominating walks}

Above we have defined two key quantities for any simple graph $G$: $\mathrm{ind}_k(F,G)$, the number of closed $k$-walks in $G$ inducing an isomorphic copy of $F$, and $\varphi(w,G)$, the corresponding graph walk density. Recalling~\eqref{eq:walk-partition}, it is then natural to ask how different walks contribute to the set $\mathcal{W}_k(G)$ of closed $k$-walks in $G$; i.e, whether one or more terms of the form $\mathrm{ind}_k(F,G)$ dominate the sum
\begin{equation}
\label{tracedecomp}
\Tr{A(G)^k} = \left| \mathcal{W}_k(G) \right| = \sum_{F \in W_k} \mathrm{ind}_k(F,G).
\end{equation}
In the theory of dense graph limits (see, e.g.,~\cite{lovasz2012large}), it is well known that
$k$-walks inducing cycles will dominate all other walks in number.

To understand which terms are significant in~\eqref{tracedecomp}, we will study its dominating terms. To do so, we define an asymptotic regime corresponding to a sequence of random graphs $\{G_n\}$ whose number of nodes tends to infinity. We assume that for all $n$ sufficiently large, $\E\left| \mathcal{W}_k(G_n) \right|>0$. This fact implies that eventually in $n$, the ratio $\E\mathrm{ind}_k(F,G_n)/\E\left| \mathcal{W}_k(G_n) \right|$ is well defined. From this, appealing to the compactness of $[0,1]$ and the monotone convergence theorem, we define
\begin{equation}
\label{eq:gamma_propn}
\gamma_{k,F}(\{G_n\}) = \liminf_{n\to\infty} \frac{\E\mathrm{ind}_k(F,G_n)}{\E\left| \mathcal{W}_k(G_n) \right|}.
\end{equation}

The quantity $\left| \mathcal{W}_k(G_n) \right|$ in~\eqref{eq:gamma_propn} is a sum of non-negative counts, as shown by~\eqref{tracedecomp}, and furthermore this sum is over a finite set.  This in turns implies that we may exchange the order of expectation and summation, and so $\gamma_{k,F}(\{G_n\})\in[0,1]$.

\begin{Definition}[Set $W_k^\ast(\{G_n\})$ of asymptotically dominating walk-induced subgraphs]\label{def:wkast2} Fix a walk length $k\in \mathbb{N}$. Let $\{G_n\}$ be a sequence of random simple graphs whose number of nodes tends to infinity and such that for all $n$ sufficiently large, $\E\left| \mathcal{W}_k(G_n) \right|>0$. We then define the set $W_k^\ast(\{G_n\})$ of asymptotically dominating walk-induced subgraphs as follows:
\begin{equation*}
W_k^\ast(\{G_n\}) = \left\{F \in W_k\ : \gamma_{k,F}(\{G_n\})>0\right\}.
\end{equation*}
\end{Definition}

We next exhibit four key properties of $W_k^\ast(\{G_n\})$.

\begin{Proposition}\label{w-k-ast-prop}
Consider a sequence $\{G_n\}$ of random simple graphs whose number of nodes tends to infinity. Under the condition that for all $n$ sufficiently large, $\E\left| \mathcal{W}_k(G_n) \right|>0$, the set of asymptotically dominating walk-induced subgraphs $W_k^\ast(\{G_n\})$, verifies the following four properties:
\begin{enumerate}
\item $W_k^\ast(\{G_n\})$ is non-empty for all $k\geq 2$.
\item \label{item:walksE} $W_k^\ast(\{G_n\})$ fully determines $\E\left| \mathcal{W}_k(G_n) \right|$ asymptotically, in the sense that
\begin{equation}
\E\left| \mathcal{W}_k(G_n) \right| \sim \sum_{F \in W_k^\ast(\{G_n\})} \E\mathrm{ind}_k(F,G_n).
\label{walksE}
\end{equation}
\item Any element of $\,W_k^\ast(\{G_n\})$ dominates all elements of $\,W_k\setminus W_k^\ast(\{G_n\})$, so that for any $F^\ast\in W_k^\ast(\{G_n\})$,
\begin{equation*}
\lim_{n\to\infty} \frac{\sum_{F' \in W_k\setminus W_k^\ast(\{G_n\})}\E\mathrm{ind}_k(F',G_n)}{\E\mathrm{ind}_k(F^\ast,G_n)} = 0.
\end{equation*}
\item Every element of $\,W_k^\ast(\{G_n\})$ is of the same order of magnitude. Specifically, for $F^\ast_1,F^\ast_2\in W_k^\ast(\{G_n\})$, there exist positive constants $C_1$ and $C_2$ such that for all sufficiently large $n$
\begin{equation*}
C_1\E\mathrm{ind}_k(F^\ast_2,G_n)\leq  \E\mathrm{ind}_k(F^\ast_1,G_n)
\leq C_2\E\mathrm{ind}_k(F^\ast_2,G_n).
\end{equation*}
Then, as for any fixed $k$ the number of closed $k$-walks are finite, for all $F^\ast\in W_k^\ast(\{G_n\})$ we have $ \E\mathrm{ind}_k(F^\ast,G_n) = \Theta\left( \E\left| \mathcal{W}_k(G_n) \right| \right) $.
\end{enumerate}
\end{Proposition}

\begin{proof} We will prove the four results in order.

\vspace{.5\baselineskip}
\emph{Proof of 1:} 
Since for $n$ large enough $\E\left| \mathcal{W}_k(G_n) \right|>0$, directly from~\eqref{tracedecomp} it follows
\begin{equation*}
\sum_{F \in W_k} \frac{\E\mathrm{ind}_k(F,G_n)}{\E\left| \mathcal{W}_k(G_n) \right|}=1.
\end{equation*}
Thus, as $W_k$ is a finite set, by taking the inferior limit on both sides of the above equation we obtain $\textstyle{\sum}_{F \in W_k}\gamma_{k,F}(\{G_n\}) = 1$.  This implies that at least one $\gamma_{k,F}(\{G_n\})$ is strictly positive, and so $W_k^\ast(\{G_n\})$ is non-empty.

\vspace{.5\baselineskip}
\emph{Proof of 2:} 
We prove that the limit superior and the limit inferior of the ratio of the two quantities in~\eqref{walksE} are both tending to one. Therefore, the limit exists and is equal to one, yielding the desired result.
First, by Item~1 of the current proposition,
\begin{equation*}
{\liminf}_{n\to\infty} \frac{{\sum}_{F \in W_k^\ast(\{G_n\})} \E\mathrm{ind}_k(F,G_n)}{\E\left| \mathcal{W}_k(G_n) \right|} = 1.
\end{equation*}
Furthermore, by~\eqref{tracedecomp}, the ratio of the two quantities in~\eqref{walksE} is smaller than one, and therefore
\begin{equation*}
{\limsup}_{n\to\infty} \frac{{\sum}_{F \in W_k^\ast(\{G_n\})} \E\mathrm{ind}_k(F,G_n)}{\E\left| \mathcal{W}_k(G_n) \right|}\leq 1.
\end{equation*}
Consequently, we conclude that
\begin{equation}
\label{eq:lim_term}
{\lim}_{n\to\infty}\frac{{\sum}_{F \in W_k^\ast(\{G_n\})} \E\mathrm{ind}_k(F,G_n)}{\E\left| \mathcal{W}_k(G_n) \right|}= 1.
\end{equation}

\vspace{.5\baselineskip}
\emph{Proof of 3:} 
Fix $F^\ast\in W_k^\ast(\{G_n\})$. We first rewrite the ratio of interest as follows:
\begin{multline}\label{eq-wk-props-item3}
\frac{\sum_{F' \in W_k\setminus W_k^\ast(\{G_n\})}\E\mathrm{ind}_k(F',G_n)}{\E\mathrm{ind}_k(F^\ast,G_n)} \\=
\left(1-\frac{\sum_{F \in W_k^\ast(\{G_n\})}\E\mathrm{ind}_k(F,G_n)}{\E\left| \mathcal{W}_k(G_n) \right|}\right)\left(\frac{\E\mathrm{ind}_k(F^\ast,G_n)}{\E\left| \mathcal{W}_k(G_n) \right|}\right)^{-1}.
\end{multline}
The first term in the product comprising the right-hand side of~\eqref{eq-wk-props-item3} admits a limit of zero by~\eqref{eq:lim_term} (Item~2 of the current proposition). To conclude that the product itself is tending to zero, we will show that its second term has a finite superior limit. To this end, we use both the monotonicity and the continuity of the function $g(x)=x^{_{-1}}$ for all $x>0$, yielding
\begin{align*}
\limsup_{n\rightarrow \infty}\left(
	\frac{\E\mathrm{ind}_k(F^\ast,G_n)}
	{\E\left| \mathcal{W}_k(G_n) \right|}\right)^{-1}
& = \lim_{n\rightarrow \infty}
	\sup_{n'>n}\left\{\left(
	     \frac{\E\mathrm{ind}_k(F^\ast,G_{n'})}
	     {\E\left| \mathcal{W}_k(G_{n'}) \right|}\right)^{-1}
	\right\}\\
& = \lim_{n\rightarrow \infty}
	\left(\inf_{n'>n}\left\{
	     \frac{\E\mathrm{ind}_k(F^\ast,G_{n'})}
	     {\E\left| \mathcal{W}_k(G_{n'}) \right|}
	  \right\}\right)^{-1}\\
& = \left(\lim_{n\rightarrow \infty}
	\inf_{n'>n}\left\{
	     \frac{\E\mathrm{ind}_k(F^\ast,G_{n'})}
	     {\E\left| \mathcal{W}_k(G_{n'}) \right|}
	\right\}\right)^{-1}\\
& = \left( \liminf_{n\rightarrow \infty}
	\frac{\E\mathrm{ind}_k(F^\ast,G_n)}
	{\E\left| \mathcal{W}_k(G_n) \right|}\right)^{-1}\\
& = \gamma_{k,F^\ast}(\{G_n\})^{-1},
\end{align*}
whence from~\eqref{eq:gamma_propn} and Definition~\ref{def:wkast2} we observe that $\gamma_{k,F^\ast}(\{G_n\})^{-1}$ is finite and greater than $1$. Thus we have shown that the limit of the first term in the product comprising the right-hand side of~\eqref{eq-wk-props-item3} is zero, and that the superior limit of the second term is finite and strictly positive.  Now, the left-hand side of~\eqref{eq-wk-props-item3}, being a ratio of non-negative terms, is itself non-negative. Therefore we conclude that its limit is zero as claimed.

\vspace{.5\baselineskip}
\emph{Proof of 4:} 
Fix $F^\ast_1,F^\ast_2\in W_k^\ast(\{G_n\})$, and consider the ratio
\begin{equation}
\label{eq:lim_ratio}
\frac{\E\mathrm{ind}_k(F^\ast_1,G_n)}
{\E\mathrm{ind}_k(F^\ast_2,G_n)}
=
\frac{{\E\mathrm{ind}_k(F^\ast_1,G_n)}/
{\left| \mathcal{W}_k(G_n) \right|}}
{{\E\mathrm{ind}_k(F^\ast_2,G_n)}/
{\left| \mathcal{W}_k(G_n) \right|}}.
\end{equation}
The inferior limits of both the numerator and denominator in the right-hand side of~\eqref{eq:lim_ratio} (i.e., $\gamma_{k,F_1^\ast}(\{G_n\})$ and $\gamma_{k,F_2^\ast}(\{G_n\})$, respectively) are finite and strictly positive, while the superior limits of both are upper bounded by unity. Using the same arguments as in the proof of Item~3 of the current proposition, it follows that: i) the superior limit of the left-hand side of~\eqref{eq:lim_ratio} is upper bounded by $\gamma^{-1}_{k,F_2^\ast}(\{G_n\}) \in [1, \infty)$ and ii) its inferior limit is lower bounded by $\gamma_{k,F_1^\ast}(\{G_n\}) \in (0,1]$. These two bounds justify the claimed result and the use of the $\Theta$ notation.
\end{proof}

\subsection{Conditions when walks inducing trees and cycles dominate}

We now characterize the set $W_k^\ast(\{G_n\})$ of asymptotically dominating walk-indu\-ced subgraphs under mild conditions on sequences $\{G_n\}$ of random simple graphs whose number of nodes tends to infinity.  To do so we will use the graph walk density $\varphi(\cdot,\cdot)$.  This density can be related to classical concepts: For example, the assumption underpinning the theory of dense graph limits is that any $\varphi(w,G_n)$ has a strictly positive limit~\citep{lovasz2012large}. For generalized random graphs with bounded kernels, the corresponding assumption is that there exists a sequence $\{ \rho \}$ taking values in $(0,\|\kappa\|_{\infty}^{-1})$ such that $\smash{\rho^{{}_{-|e(F_w)|}}}\varphi(w,G_n)$ always has a strictly positive limit~\citep{bollobas2007phase,olhede2013network}.

Recalling the notion of walk extensions from Definition~\ref{def:extens}, we introduce the following assumptions.

\begin{Assumption}
 \label{asump1.1} $\{G_n\}$ is a sequence of random simple graphs whose number of nodes tend to infinity and such that for $n$ sufficiently large,
\begin{equation*}
\E\left| \mathcal{W}_k(G_n) \right|>0.
\end{equation*}
\end{Assumption}

\begin{Assumption}
\label{asump1.3} For all $w\in\mathcal{W}_k(K_k)$ such that for $n$ sufficiently large $\E\varphi(w,G_n)>0$, if any extensions of $w$ exist in $\mathcal{W}_k(K_k)$, then at least one of them---say $w'$---satisfies
\begin{equation*}
\frac{|v(G_n)| \E \varphi(w',G_n)}{ \E\varphi(w,G_n) }\rightarrow \infty.
\end{equation*}
\end{Assumption}

This allows us to characterize which walks dominate in expectation.

\subsection{Dominance of walks mapping out trees and cycles}

\begin{Theorem}\label{thm:dominant-walk}
Let $\{G_n\}$ be a sequence of random simple graphs satisfying Assumptions~\ref{asump1.1} and~\ref{asump1.3}. Then, for $k=2$ we have $W_2^\ast(\{G_n\}) = W_2 = \{ K_2 \}$. Furthermore, if $k\geq 3$ is odd,
\begin{equation*}
W_k^\ast(\{G_n\}) =  \{C_k\},
\end{equation*}
while if $k\geq 4$ is even,
\begin{equation*}
W_k^\ast(\{G_n\}) \subset  \{C_k\}\cup\mathcal{T}_{k/2+1}.
\end{equation*}
\end{Theorem}

\begin{proof}
First, we appeal to Item~1 of Proposition~\ref{w-k-ast-prop} (which holds under Assumption~\ref{asump1.1}), and conclude that $W_k^\ast(\{G_n\})$ is not empty for all $k\geq 2$. Therefore, as $W_2$ is the singleton $\{ K_2 \}$, we have $W_2^\ast(\{G_n\}) = W_2 = \{ K_2 \}$. Similarly, $W_3^\ast(\{G_n\}) = W_3 = \{ C_3 \}$.

Having established that the set $W_k^\ast(\{G_n\})$ is non-empty and having determined $W_2^\ast(\{G_n\})$ and $W_3^\ast(\{G_n\})$, we next show that for all $k\geq 4$, if $F\in W_k$ admits at least one extension, then $F\not\in W_k^\ast(\{G_n\})$.  The result will then follow immediately from Lemma~\ref{extens-prop}, which asserts that every $F\in W_k$ not isomorphic either to $C_k$ when $k\geq 3$, or to an element of $\mathcal{T}_{k/2+1}$ when $k\geq 2$ is even, admits an extension.

Thus, we fix $k\geq 4$ and $F\in W_k$ such that $F$ admits at least one extension. To show that $F\not\in W_k^\ast(\{G_n\})$, we will appeal to the definition of $W_k^\ast(\{G_n\})$ in Definition~\ref{def:wkast2}, whereby $F\in W_k^\ast(\{G_n\})$ if and only if
\begin{equation}\label{proof-thms1-case0}
\liminf_{n\to\infty}\frac{\E\mathrm{ind}_k(F,G_n)}{\E\left| \mathcal{W}_k(G_n) \right|}>0.
\end{equation}
Note that this ratio is always well defined under Assumption~\ref{asump1.1} since for all $n$ sufficiently large, $\E\left| \mathcal{W}_k(G_n) \right|>0$. We distinguish two (exhaustive) cases, namely the case where $\E \mathrm{ind}_k(F,G_n)$ is positive but becomes negligble
in comparison to $\E\left| \mathcal{W}_k(G_n) \right|$, or when we allow $\E \mathrm{ind}_k(F,G_n)$ to be zero for large $n$. Then:
\begin{enumerate}
\item Suppose $\E\mathrm{ind}_k(F,G_n) > 0$ for all $n$ sufficiently large. First, since $F$ admits at least an extension, we can fix $w\in\mathcal{W}_k(G_n)$ such that $F_w \equiv F$ and $w$ admits at least one extension (see Definition~\ref{def:extens}). Then, by construction, $\E\mathrm{ind}_k(F_w,G_n) = \E\mathrm{ind}_k(F,G_n) > 0$, so that $\E\varphi(w,G_n) > 0$. Therefore, we can appeal to Assumption~\ref{asump1.3}, and fix $w'\in\mathcal{W}_k(G_n)$ to be an extension of $w$ such that
\begin{equation}\label{proof-thms1-assump1}
\lim_{n\to\infty}\frac{|v(G_n)|\E\varphi(w',G_n)}{\E\varphi(w,G_n)}=\infty.
\end{equation}

We now use~\eqref{proof-thms1-assump1} to show that ${\E\mathrm{ind}_k(F,G_n)}/{\E\left| \mathcal{W}_k(G_n) \right|}$ is tending to zero. To this end we first lower-bound the ratio and note that by construction
\begin{equation}\label{proof-thms1-case1.1}
0\leq \frac{\E\mathrm{ind}_k(F,G_n)}{\E\left| \mathcal{W}_k(G_n) \right|}.
\end{equation}
Next we upper bound this ratio. To this end we note that since $\E\left| \mathcal{W}_k(G_n) \right|$ is larger than $\E\mathrm{ind}_k(F_{w'},G_n)$, we have
\begin{equation}\label{proof-thms1-case1.2}
\frac{\E\mathrm{ind}_k(F,G_n)}{\E\left| \mathcal{W}_k(G_n) \right|}
\leq
\frac{\E\mathrm{ind}_k(F,G_n)}{\E\mathrm{ind}_k(F_{w'},G_n)}.
\end{equation}
Then, as for any $w''$, $\E\mathrm{ind}_k(F_{w''},G_n) = \Theta(|v(G_n)|^{|v(F_{w''})|}\E\varphi(w'',G_n))$ (see Definition~\ref{varphi}), we obtain from~\eqref{proof-thms1-case1.2} that
\begin{align}
\nonumber
\frac{\E\mathrm{ind}_k(F,G_n)}{\E\left| \mathcal{W}_k(G_n) \right|}
&=\mathcal{O}\left(\frac{|v(G_n)|^{|v(F_w)|}\E\varphi(w,G_n)}{|v(G_n)|^{|v(F_{w'})|}\E\varphi(w',G_n)}\right)\\
\label{proof-thms1-case1.3}
&=\mathcal{O}\left(\frac{\E\varphi(w,G_n)}{|v(G_n)|\E\varphi(w',G_n)}\right),
\end{align}
since as $w'$ is an extension of $w$ it visits exactly one more vertex than $w$, and therefore $|v(F_{w'})| = |v(F_{w})|+1$. We recognize on the right hand side of~\eqref{proof-thms1-case1.3} the inverse of the term on the left hand side of~\eqref{proof-thms1-assump1}. Therefore, jointly with~\eqref{proof-thms1-case1.1}, we recover
\begin{equation*}
0\leq\frac{\E\mathrm{ind}_k(F,G_n)}{\E\left| \mathcal{W}_k(G_n) \right|}
=o(1),
\end{equation*}
and therefore that
\begin{equation*}
\liminf_{n\to\infty}\frac{\E\mathrm{ind}_k(F,G_n)}{\E\left| \mathcal{W}_k(G_n) \right|}=0.
\end{equation*}
This shows that $F\not\in W_k^\ast(\{G_n\})$ as it does not verify the necessary condition presented in~\eqref{proof-thms1-case0}.

\item Alternatively, assume there exists no $n$ such that $ \E\mathrm{ind}_k(F,G_{n'}) > 0 $ for all $n'>n$. Then, for all $n$ we have that $ \textstyle{\inf}_{n'>n}\,\{\E\mathrm{ind}_k(F,G_{n'})\}=0$. It follows that for all $n$, $ \textstyle{\inf}_{n'>n}\,\{{\E\mathrm{ind}_k(F,G_n)}/{\E\left| \mathcal{W}_k(G_n) \right|}\} = 0$.
Therefore
\begin{align*}
\liminf_{n\to\infty} \, \frac{\E\mathrm{ind}_k(F,G_n)}{\E\left| \mathcal{W}_k(G_n) \right|}
&= \lim_{n\to\infty}\inf_{n'>n} \, \left\{\frac{\E\mathrm{ind}_k(F,G_{n'})}{\E\left| \mathcal{W}_k(G_{n'}) \right|}\right\}\\
&= \lim_{n\to\infty}0 = 0.
\end{align*}
As $F$ does not verify the necessary condition presented in~\eqref{proof-thms1-case0}, this shows that $F\not\in W_k^\ast(\{G_n\})$.
\end{enumerate}
To conclude, in both cases we have shown that $F\not\in W_k^\ast(\{G_n\})$. As explained above, this yields the result.
\end{proof}

\section{\!Dominating walks in kernel-based random graphs}\label{sec:S3}
\subsection{Kernel-based random graphs}

We now undertake a detailed analysis of walk dominance within kernel-based random graphs~\citep{bollobas2007phase, bickel2009nonparametric}, assuming that the kernel is bounded.

\begin{Definition}[Kernel $\kappa$]\label{kernel}
A kernel $\kappa$ is a bounded symmetric map from $(0,1)^2$ to $[0,\infty)$, normalized to integrate to unity so that $\|\kappa\|_1 = 1$.
\end{Definition}

\begin{Definition}[Kernel-based random graph $G(n,\rho\kappa)$]\label{model}
Fix a kernel $\kappa$ and a scalar $\rho\in(0,\|\kappa\|_{\infty}^{-1})$. We call $G(n,\rho\kappa)$ the kernel-based random graph whose symmetric adjacency matrix $A \in \{0,1\}^{n\times n}$ is obtained from a sequence $ x = (x_i)_{i=1}^n$ of $n$ independent $\operatorname{Uniform}(0,1)$ variates by independently setting
\begin{equation*}
	\mathbb{P}(A_{ij}=1\,\vert\, x_i,x_j) =
		\rho\cdot\kappa(x_i,x_j)
\end{equation*}
for all $1 \leq i < j \leq n$.  The matrix $A$ is completed by taking $A_{ji} =A_{ij}$ for $1\leq i<j\leq n$, and $A_{ii} =0$ for $1\leq i\leq n$, yielding a simple random graph $G = G(n,\rho\kappa)$ with vertex labels $1,2,\dots,n$ and $\mathbb{P}(A_{ij}=1) = \rho$.
\end{Definition}

By construction, under the model of Definition~\ref{model}, whenever graphs $F$ and $F'$ are defined on the same vertex set as $G$, then $\mathbb{P} \left\{F' \subset G\right\} = \mathbb{P} \left\{F \subset G\right\}$ if $F'\equiv F$.  In particular, for $v(F) \subset \{1,2,\dots,n\}$,
\begin{multline}\label{PrFinG}
\!\mathbb{P} \left\{F \subset G\right\}=\E \prod_{ij\in e(F)} A_{ij}
=\E_x \left\{\prod_{ij\in e(F)}\rho \kappa(x_i,x_j)\right\}\\
=\rho^{|e(F)|}\int_{[0,1]^{|v(F)|}} \prod_{ij\in e(F)}\kappa(x_i,x_j)\prod_{i\in v(F)} \, dx_i.\!\!
\end{multline}

Equipped with~\eqref{PrFinG}, we may state the following classical result. It allows us to deduce $ \E \mathrm{ind}_k(F,G)$, which is necessary to characterize walks within kernel-based random graphs.

\begin{Lemma}\label{exp-XFG}
Fix a graph $F$, assume without loss of generality that $v(F)$ is a subset of $\{1,2,\dots,n\}$, and let $G = G(n,\rho\kappa)$ in accordance with Definition~\ref{model}. Then the expected number of copies of $F$ in $G$ is
\begin{equation*}
\E X_F(G) = X_F(K_n) \, \mathbb{P} \left\{F \subset G\right\} .
\end{equation*}
\end{Lemma}

\begin{proof}
From the definition $X_F(G) = \left| \left\{ F'\subset G : F'\equiv F \right\} \right|$, we obtain
\begin{align*}
\nonumber
\E X_F(G) &= \E \sum_{F' \subset G}1_{\{F'\equiv F\}}
\\ &=\E \sum_{F' \subset K_n}1_{\{F'\equiv F\}}1_{\{F' \subset G\}}
\\ &= \sum_{F' \subset K_n}1_{\{F'\equiv F\}}\mathbb{P} \left\{F' \subset G\right\}
\\ &= \sum_{F' \subset K_n}1_{\{F'\equiv F\}}\mathbb{P} \left\{F \subset G\right\} ,
\end{align*}
where the final equality follows since $ \mathbb{P} \left\{F' \subset G\right\} $ is constant for all $F'\equiv F$.
\end{proof}
Thus we deduce the form of $ \E \mathrm{ind}_k(F,G)$ under the model of Definition~\ref{model} as we shall now show.

\begin{Lemma}\label{exp-ind}
Fix a graph $F$, and let $G = G(n,\rho\kappa)$ in accordance with Definition~\ref{model}. Then
\begin{equation*}
\E \mathrm{ind}_k(F,G) = \mathrm{ind}_k(F,K_n) \, \mathbb{P} \left\{F \subset G\right\} .
\end{equation*}
\end{Lemma}

\begin{proof}
Lemma~\ref{nb-copies} shows that for fixed graphs $F,G$ and $k\in\mathbb{N}$, it follows that $ \mathrm{ind}_k(F,G)\! =\! \mathrm{ind}_k(F,F) X_F(G) $.  Taking expectations under the model of Definition~\ref{model} and then applying Lemma~\ref{exp-XFG} for $ \E X_F(G) $, we have
\begin{align*}
\E \mathrm{ind}_k(F,G) & = \mathrm{ind}_k(F,F) \E X_F(G)
\\ & = \mathrm{ind}_k(F,F) \, X_F(K_n) \, \mathbb{P} \left\{F \subset G\right\}.
\end{align*}
Applying Lemma~\ref{nb-copies} again via $ \mathrm{ind}_k(F,K_n) = \mathrm{ind}_k(F,F) X_F(K_n) $, we obtain the claimed result.
\end{proof}

The ratios $ \mathrm{ind}_k(F,G) / \mathrm{ind}_k(F,K_n) = X_F(G) / X_F(K_n) = \mathrm{emb}(F,G) / \mathrm{emb}(F,K_n) $ are natural empirical counterparts to $ \mathbb{P} \left\{F \subset G\right\} $, as we have $ \E X_F(G) / X_F(K_n)$ $=$ $\E \mathrm{ind}_k(F,G) / \mathrm{ind}_k(F,K_n)$ $=$ $\mathbb{P} \left\{F \subset G\right\} $ by Lemmas~\ref{exp-XFG} and~\ref{exp-ind}, and in turn $X_F(G) = \mathrm{emb}(F,G) / \mathrm{aut}(F) $ for any graph $G$.  This leads to the following standard definitions of embedding densities.

\begin{Definition}[Embedding densities $s_\rho(F,G)$ and $s(F,\kappa)$~\citep{bollobas2009metric}]\label{embedding}
Fix a graph $F$, assume without loss of generality that $v(F) \subset \{1,2,\dots,n\}$, and let $G = G(n,\rho\kappa)$ in accordance with Definition~\ref{model}. Then, we denote the $\rho$-normalized embedding density of $F$ in $G$ by $s_\rho(F,G)$ and its expectation under the model of Definition~\ref{model}, the kernel embedding density, by $s(F,\kappa)$:
\begin{align*}
s_\rho(F,G) & = \frac{ X_F(G) }{ \rho^{|e(F)|} X_F(K_n) } ,
\\ s(F,\kappa) = \E s_\rho(F,G) & = \int_{[0,1]^{|v(F)|}}\prod_{ij\in e(F)}\kappa(x_i,x_j)\prod_{i\in v(F)} \, dx_i.
\end{align*}
\end{Definition}

From Definition~\ref{embedding}, we may use Lemma~\ref{exp-ind} and that $ \mathrm{emb}(F,K_n) = (n)_{|v(F)|} $, with $ (n)_{|v(F)|} = n!/(n-|v(F)|)!$ the falling factorial, and conclude that
\begin{align}
\E \mathrm{ind}_k(F,G) & = \mathrm{ind}_k(F,K_n) \, \mathbb{P} \left\{F \subset G\right\}
\nonumber \\ & = \mathrm{ind}_k(F,K_n) \rho^{|e(F)|} s(F,\kappa)
\nonumber \\ & = \mathrm{ind}_k(F,F) X_F(K_n) \rho^{|e(F)|} s(F,\kappa)
\nonumber \\ & = (n)_{|v(F)|} \rho^{|e(F)|} \frac{ \mathrm{ind}_k(F,F)  s(F,\kappa) }{ \mathrm{aut}(F) } .
\label{ind-const}
\end{align}

Inspecting~\eqref{ind-const}, we quantify the order of magnitude of the expected number of walk-induced copies of $F$ in the random graph model of Definition~\ref{model} as follows.

\begin{Definition}[Order $\Psi_F$ of the expected number of walks inducing $F$~\citep{janson2004tails}]\label{def:varphi}
Fix a graph $F$, a network size $n$ and a kernel $\kappa$. Let $\Psi_F$ be such that
\begin{equation*}
\Psi_F = n^{|v(F)|}\rho^{|e(F)|}.
\end{equation*}
\end{Definition}

\subsection{Dominating closed walks in kernel-based random graphs}

When considering sequences $\{G_n\}$ of kernel-based random graphs, it is convenient to work only with the sequence $\{\rho\}$, rather than with the sequences $\{\E \mathrm{ind}_k(F,G)\}$ for $F\in W_k$, as we now show.

\begin{Remark}\label{remark:kernelver}
Fix a walk length $k$ and a sequence of random graphs $\{G_n\}$ where each $G_n$ is generated in accordance with $G(n,\rho\kappa)$ from Definition~\ref{model} for some sequence $\{\rho\}$ taking values in $(0,\|\kappa\|_{\infty}^{-1})$. In this setting, if $n\rho=\omega(1)$, then Assumptions~\ref{asump1.1} and~\ref{asump1.3} are satisfied, and so Theorem~\ref{thm:dominant-walk} applies to the sequence of graphs $\{G_n\}$. To see this, first observe that from~\eqref{tracedecomp} and~\eqref{ind-const}, we have
\begin{align*}
\E\left| \mathcal{W}_k(G_n) \right| & =\sum_{F\in W_k} \E \mathrm{ind}_k(F,G_n) \\ & =\sum_{F\in W_k}
(n)_{|v(F)|} \rho^{|e(F)|} \frac{ \mathrm{ind}_k(F,F)  s(F,\kappa) }{ \mathrm{aut}(F) } \\ & >0,
\end{align*}
as $s\neq 0$ by construction, $W_k$ is non-empty and $\rho>0$ by assumption. Therefore, we see directly that Assumption~\ref{asump1.1} is satisfied. Second, for $w$ and $w'$ two closed $k$-walks such that $w'$ is an extension of $w$ (see Definition~\ref{def:extens}), from Lemma~\ref{exp-XFG} and recalling from Definition~\ref{def:extens} that $|e(F_{w'})|-|e(F_w)| \leq 1$, we obtain
\begin{align*}
\frac{|v(G_n)|\E\varphi(w',G_n)}{\E\varphi(w,G_n)} &= \frac{n \E X_{F_{w'}}(G_n)/ X_{F_{w'}}(K_n)}{\E X_{F_w}(G_n)/ X_{F_w}(K_n)}\\
&=\frac{n X_{F_{w'}}(K_n) \, \rho^{|e(F_{w'})|}\E s_\rho(F_{w'},G_n)/X_{F_{w'}}(K_n) }{X_{F_w}(K_n) \, \rho^{|e(F_w)|}\E s_\rho(F_w,G_n)/X_{F_w}(K_n) }\\
&=n \rho^{|e(F_{w'})|-|e(F_w)|} s(F_{w'},\kappa) / s(F_w,\kappa)\\
&=\Omega \left( n \rho \right),
\end{align*}
since $ \rho < \|\kappa\|_{\infty}^{-1} \leq 1 $. Therefore $\{G_n\}$ verifies Assumption~\ref{asump1.3} if $n\rho=\omega(1)$.
\end{Remark}

A uniform scaling of network edge probabilities is implicit to the model of Definition~\ref{model}. This uniform scaling will permit us to derive explicit error rates for the remainder terms in Theorem~\ref{thm:dominant-walk}. We now introduce $\widetilde{W}_k^\ast$, which for every fixed $n$ plays a role analogous to $W_k^\ast(\{G_n\})$ but is constructed directly from $\Psi_F$. This simplification enables us to introduce $W_k'$, which contains second-order or sub-dominating walk-induced subgraphs, and will be crucial to obtaining error rates in Theorem~\ref{thm:dominant-walk}.

\begin{Definition}[Sets $\widetilde{W}_k^\ast, W_k'$ of dominating walk-induced subgraphs]\label{def:wkast} Fix a walk length $k\in \mathbb{N}$. Define the set of dominating walk-induced subgraphs to be $\widetilde{W}_k^\ast$, and the set of sub-dominating walk-induced subgraphs to be  $W_k'$, where the two sets satisfy the following equations
\begin{align*}
\widetilde{W}_k^\ast &= \{F \in W_k\ : \Psi_F = \max_{F'\in W_k}\ \Psi_{F'}\}, \\
W_k' &= \{F \in \overline{W_k^\ast} : \Psi_F = \max_{F'\in \overline{W_k^\ast}} \Psi_{F'}\}; \quad \overline{W_k^\ast} = W_k \setminus \widetilde{W}_k^\ast.
\end{align*}
\end{Definition}

\begin{Lemma}\label{dominantscaling}
Consider the setting of Definition~\ref{model}. Fix an integer $k >3$ and a scalar $\rho\in(0,\|\kappa\|_{\infty}^{-1})$ such that $n\rho>1$ for all integers $n>k$. Set
\begin{equation*}
k^\ast = \frac{\log n}{\log \sqrt{n\rho} },
\end{equation*}
and recall that $C_k$ is the $k$-cycle, $C_kP_l$ is the $(k,l)$-tadpole and $\mathcal{T}_k$ is the set of unlabeled trees on $k$ vertices. Then for every fixed $n>k$,
\begin{alignat}{2}
\nonumber
\widetilde{W}_k^\ast & =
\begin{cases}
\{ C_k \} & \textnormal{if $k$ is odd, or if $k$ is even and $k > k^\ast$,}
\\ \{ C_k \} \cup \mathcal{T}_{k/2+1} & \textnormal{if $k$ is even and $k = k^\ast$ for $k^\ast$ an integer,}
\\ \mathcal{T}_{k/2+1} & \textnormal{if $k$ is even and $k < k^\ast$.}
\end{cases}
\label{W_k_prime}
\\ W_k' & =
\begin{cases}
\{C_{k-2}P_2\} & \textnormal{if $k$ is odd, or if $k$ is even and $k>k^\ast+2$,}\\
\{C_{k-2}P_2\} \cup \mathcal{T}_{k/2+1}& \textnormal{if $k$ is even and $k = k^\ast+2$ for $k^\ast$ an integer,}\\
\mathcal{T}_{k/2+1} & \textnormal{if $k$ is even and $k\in (k^\ast,k^\ast+2)$,}\\
\{C_{k-2}P_2\} \cup \mathcal{T}_{k/2} & \textnormal{if $k$ is even and $k = k^\ast$ for $k^\ast$ an integer,}\\
\{ C_k \} & \textnormal{if $k$ is even and $k\in (k^\ast-2,k^\ast)$,}\\
\{ C_k \} \cup \mathcal{T}_{k/2}& \textnormal{if $k$ is even and $k = k^\ast-2$ for $k^\ast$ an integer,}\\
\mathcal{T}_{k/2} & \textnormal{if $k$ is even and $k < k^\ast-2$.}
\end{cases}
\end{alignat}
\end{Lemma}%

\begin{proof}
First, observe that $\textstyle \max_{F\in W_k} \Psi_F $ exists, since $W_k$ is finite and non-empty for all $k \geq 2$.  Now write
\begin{equation*}
\Psi_F=(n\rho)^{|v(F)|}\rho^{|e(F)|-|v(F)|}.
\end{equation*}
We draw two conclusions from the above equation: First, the hypothesis $ n\rho > 1 $ implies that $\Psi_F$ is monotone increasing in $|v(F)|$ for fixed $|e(F)|-|v(F)|$.  Second, the hypothesis $ 0 < \rho < \|\kappa\|_\infty^{-1} \leq 1 $ implies that $\Psi_F$ is monotone decreasing in $|e(F)|-|v(F)|$ for fixed $|v(F)|$. Using these two facts we proceed as follows:

\begin{enumerate}
\item Suppose $|e(F)|=|v(F)|$. Then $ \Psi_F = (n\rho)^{|v(F)|} $.  By monotonicity, we have $\smash{ \textstyle \max_{F \in W_k : |e(F)| = |v(F)|} \Psi_F = (n \rho)^k} $, achieved uniquely by the $k$-cycle $C_k$, since no other closed $k$-walk visits $k$ vertices.

\item Suppose $|e(F)|>|v(F)|$. Then $ \Psi_F = (n \rho)^{|v(F)|} \rho^{|e(F)|-|v(F)|} < (n \rho)^{|v(F)|} \leq (n \rho)^k $ for $C_k$, since $\rho<1$.

\item Suppose $|e(F)|<|v(F)|$. Since an odd closed walk contains an odd cycle, $ F \in W_{k = 2r+1} \Rightarrow |e(F)| \geq |v(F)| $.  By Lemma~\ref{wk-props}, $F \in W_{k = 2r} \Leftrightarrow F \in \cup_{i=1}^{r} \mathcal{T}_{i+1} $, implying that $ |e(F)| = |v(F)| - 1 \leq k/2 $.  Thus
\begin{equation*}
\Psi_F = (n \rho)^{|v(F)|} \rho^{|e(F)|-|v(F)|} = (n \rho)^{|v(F)|} \rho^{-1} \leq (n \rho)^{k/2+1} \rho^{-1}
\end{equation*}
by monotonicity, with
\begin{equation*}
\max_{F \in W_{k=2r} \,:\, |e(F)| < |v(F)|} \Psi_F = (n \rho)^{k/2+1} \rho^{-1}
\end{equation*}
achieved by all $(k/2+1)$-trees $ F \in \mathcal{T}_{k/2+1} $ for $k$ even; this shape is not achievable for $k$ odd.
\end{enumerate}

Thus, if $k$ is odd we only compare Cases~1 and~2 to determine the maximizer of $\Psi_F$ over $W_k$.  We see that Case~1 always dominates Case~2.  If $k$ is even, we may still discard Case~2 in favor of Case~1, leaving us to compare which of Cases~1 and~3 dominates.  Setting $ (n \rho)^{k^\ast} = (n \rho)^{k^\ast/2+1} \rho^{-1} $ and solving for $k^\ast$, we determine the two forms of $\widetilde{W}_k^\ast$ as claimed.

Our next step is to determine the maximizer of $\Psi_F$ over $W_k \setminus \widetilde{W}_k^\ast$:
\begin{enumerate}
\item If $k \geq 5$ is odd, we repeat the above arguments mutatis mutandis, with the $(k-2,1)$-tadpole $C_{k-2}P_2$ replacing the cycle $C_k$ in Case~1.  This follows by Lemma~\ref{wk-props}, which asserts that it is the unique $F \in W_k$ such that $|v(F)|=|e(F)|=k-1$. Hence $\smash{ \textstyle \max_{F\in W_k \setminus \{ C_k \} } \Psi_F = (n\rho)^{k-1}} $.

\item If $k$ is even and greater than $k^\ast$ we must compare Cases~1 and~3. In this setting, $C_k\in \widetilde{W}_k^\ast$ and $C_{k-2}P_2$ replaces $C_k$ in Case~1, so that it is sufficient to compare $\smash{\Psi_{C_{k-2}P_2}=(n\rho)^{k-1}}$ to $\Psi_F=(n\rho)^{k/2+1}\rho^{-1}$ for $F$ any $(k/2+1)$-tree. Since we have that $(n\rho)^{(k^\ast+2)-1} = (n\rho)^{(k^\ast+2)/2+1}\rho^{-1}$, we obtain the first three lines of~\eqref{W_k_prime}.

\item If $k$ is even and equal to $k^\ast$, for $k^\ast$ an integer, we must compare Cases~1 and~3. In this setting, $C_k\in \widetilde{W}_k^\ast$ and $\mathcal{T}_{k/2+1}\subset \widetilde{W}_k^\ast$, hence $C_{k-2}P_2$ replaces $C_k$ in Case~1 and $\mathcal{T}_{k/2}$ replaces $\mathcal{T}_{k/2+1}$ in Case~3 so that it is sufficient to compare $\Psi_{C_{k-2}P_2}=(n\rho)^{k-1}$ to $\Psi_F=(n\rho)^{k/2}\rho^{-1}$ for $F$ any $(k/2)$-tree. Since we have that $(n\rho)^{k^\ast-1} = (n\rho)^{k^\ast/2}\rho^{-1}$, we obtain the fourth line of~\eqref{W_k_prime}.

\item If $k$ is even and smaller than $k^\ast$ we must compare Cases~1 and~3. In this setting case $\mathcal{T}_{k/2+1} = \widetilde{W}_k^\ast$ and $\mathcal{T}_{k/2}$ replaces $\mathcal{T}_{k/2+1}$ in Case~3, so that it is sufficient to compare $\Psi_{C_k}=(n\rho)^k$ to $\Psi_F=(n\rho)^{k/2}\rho^{-1}$ for $F$ a $(k/2)$-tree. Since we have that $(n\rho)^{(k^\ast-2)} = (n\rho)^{(k^\ast-2)/2}\rho^{-1}$, we obtain the last three lines of~\eqref{W_k_prime}.
\qedhere
\end{enumerate}
\end{proof}

We have now fully described the sets $\widetilde{W}_k^\ast$ and $W_k'$ in the setting of Definition~\ref{model}. In this way Lemma~\ref{dominantscaling} will allow us to show that the number of $k$-walks inducing the elements of $\widetilde{W}_k^\ast$ dominates (in expectation) the number of $k$-walks inducing any other graph.

\subsection{Expected number of walks mapping out trees and cycles}\label{sec:phase}
\begin{Definition}[Walk embedding density $\nu_k(F,\kappa)$]\label{def:nu_k}
Fix a graph $F$, a walk length $k$, and a generalized random graph kernel $\kappa$. In analogy to the kernel embedding density of Definition~\ref{embedding}, define
\begin{equation*}
\nu_k(F,\kappa)=\frac{\mathrm{ind}_k(F,F)}{\mathrm{aut}(F)}s(F,\kappa).
\end{equation*}
\end{Definition}
Recalling $\Psi_F = n^{|v(F)|}\rho^{|e(F)|}$ from Definition~\ref{def:varphi}, we see that $\nu_k(F,\kappa)$ is the limit of $\E\mathrm{ind}_k(F,G)/\Psi_F$. Furthermore, $\nu_k(F,\kappa) > 0 $ whenever $\mathrm{ind}_k(F,F)>0$.

\begin{Theorem}\label{thm:phase}
Fix $k>3$ and $n\in {\mathbb{N}}$ such that $n>k$. Let $G$ be a random graph distributed according to the kernel-based random graph model $G(n,\rho\kappa)$ as given by Definition~\ref{model}. Set $\mu=n\rho$ and assume that $\mu>1$. Then, with $k^\ast = {\log n}/{\log\sqrt\mu}$,
\begin{enumerate}[leftmargin=.2 in]
\item If $k$ is odd or $k$ is even and $k> k^\ast+2$:
\begin{equation*}
\frac{\E\left|\mathcal{W}_k(G)\right|}{\E\left|\{w\in\mathcal{W}_k(G) : F_w \equiv C_k\}\right|}
=1 + \frac{1}{\mu}\left(\frac{\nu_k(C_{k-2}P_2,\kappa)}{\nu_k(C_k,\kappa)}+\epsilon_1(n;k,\kappa)\right).
\end{equation*}
\item If $k$ is even and $k= k^\ast+2$ for $k^\ast$ an integer:
\begin{multline*}
\frac{\E\left|\mathcal{W}_k(G)\right|}{\E\left|\{w\in\mathcal{W}_k(G) : F_w \equiv C_k\}\right|}
= 1\\
+ \frac{1}{\mu}\left(\frac{\nu_k(C_{k-2}P_2,\kappa)+\sum_{T\in\mathcal{T}_{k/2+1}}\nu_k(T,\kappa)}{\nu_k(C_k,\kappa)}
\vphantom{\frac{\sum_{T\in\mathcal{T}_{k/2+1}}\nu_k(T,\kappa)}{\nu_k(C_k,\kappa)}}+\epsilon_2(n;k,\kappa)\right).
\end{multline*}
\item If $k$ is even and $k\in (k^\ast,k^\ast+2)$:
\begin{multline*}
\!\!\!\!\!\!\!\!\!\!\!\!\!\!\!\!\frac{\E\left|\mathcal{W}_k(G)\right|}{\E\left|\{w\in\mathcal{W}_k(G) : F_w \equiv C_k\}\right|}
= 1 + \frac{1}{\sqrt{\mu}^{k-k^\ast}}\left(\frac{\sum_{T\in\mathcal{T}_{k/2+1}}\nu_k(T,\kappa)}{\nu_k(C_k,\kappa)}
+\epsilon_3(n;k,\kappa)\right).
\end{multline*}
\item If $k$ is even and $k = k^\ast$ for $k^\ast$ an integer:
\begin{multline*}
\frac{\E\left|\mathcal{W}_k(G)\right|}{\E\left|\{w\in\mathcal{W}_k(G) : F_w \in \{C_k\}\cup\mathcal{T}_{k/2+1}\}\right|}
= 1\\ + \frac{1}{\mu}\left(\frac{\nu_k(C_{k-2}P_2,\kappa)+\sum_{T\in\mathcal{T}_{k/2}}\nu_k(T,\kappa)}{\nu_k(C_k,\kappa)+\sum_{T\in\mathcal{T}_{k/2+1}}\nu_k(T,\kappa)}\right.
\left.\vphantom{\frac{\sum_{T\in\mathcal{T}_{k/2+1}}\nu_k(T,\kappa)}{\nu_k(C_k,\kappa)}}+\epsilon_4(n;k,\kappa)\right).
\end{multline*}
\item If $k$ is even and $k\in (k^\ast-2,k^\ast)$:
\begin{multline*}
\frac{\E\left|\mathcal{W}_k(G)\right|}{\E\left|\{w\in\mathcal{W}_k(G) : F_w \in \mathcal{T}_{k/2+1}\}\right|}
= 1\\ + \frac{1}{\sqrt{\mu}^{k^\ast-k}}\left(\frac{\nu_k(C_k,\kappa)}{\sum_{T\in\mathcal{T}_{k/2+1}}\nu_k(T,\kappa)}+\epsilon_5(n;k,\kappa)\right).
\end{multline*}
\item If $k$ is even and $k = k^\ast-2$ for $k^\ast$ an integer:
\begin{multline*}
\frac{\E\left|\mathcal{W}_k(G)\right|}{\E\left|\{w\in\mathcal{W}_k(G) : F_w \in \mathcal{T}_{k/2+1}\}\right|}
= 1\\ + \frac{1}{\mu}\left(\frac{\nu_k(C_k,\kappa)+\sum_{T\in\mathcal{T}_{k/2}}\nu_k(T,\kappa)}{\sum_{T\in\mathcal{T}_{k/2+1}}\nu_k(T,\kappa)}+\epsilon_6(n;k,\kappa)\right).
\end{multline*}
\item If $k$ is even and $k < k^\ast-2$:
\begin{multline*}
\!\!\!\!\!\!\!\!\!\!\!\frac{\E\left|\mathcal{W}_k(G)\right|}{\E\left|\{w\in\mathcal{W}_k(G) : F_w \in \mathcal{T}_{k/2+1}\}\right|}
= 1 + \frac{1}{\mu}\left(\frac{\sum_{T\in\mathcal{T}_{k/2}}\nu_k(T,\kappa)}{\sum_{T\in\mathcal{T}_{k/2+1}}\nu_k(T,\kappa)}+\epsilon_7(n;k,\kappa)\right).
\end{multline*}
\end{enumerate}

Error terms $\{\epsilon_i(n;k,\kappa)\}$ for $1\leq i\leq7$ are upper bounded by the quantity $n^k\textstyle{\sum}_{F\in W_k\setminus \widetilde{W}_k^\ast} \nu_k(F,\kappa)/(n)_k\textstyle{\sum}_{F\in \widetilde{W}_k^\ast} \nu_k(F,\kappa)$. Each error term is positive, and furthermore is such that if $\mu$ diverges, then $\lim_{n\to\infty}\epsilon_i(n;k,\kappa) = 0$.
\end{Theorem}

\begin{proof}
Combining Definition~\ref{def:wkast} with~\eqref{tracedecomp} and taking expectations, we have
\begin{equation}
\label{sumsplit}
\E\left|\mathcal{W}_k(G)\right| =
	\sum_{F^\ast\in \widetilde{W}_k^\ast}\E\mathrm{ind}_k(F^\ast,G)+
	\sum_{F\in \overline{\widetilde{W}_k^\ast}}\E\mathrm{ind}_k(F,G).
\end{equation}
Then, from Definition~\ref{embedding} and Lemma~\ref{exp-ind}, we obtain directly that $\E \mathrm{ind}_k(F,G)= \Psi_F\nu_k(F,\kappa)(n)_{|v(F)|}/n^{|v(F)|}$. Thus, from~\eqref{sumsplit} we recover
\begin{align*}
\nonumber
\frac{\E\left|\mathcal{W}_k(G)\right|}{\E\left|\{w\in\mathcal{W}_k(G) : F_w \in \widetilde{W}_k^\ast\}\right|}
& = 1 + \frac{\sum_{F\in \overline{\widetilde{W}_k^\ast}}\E\mathrm{ind}_k(F,G)}{\sum_{F^\ast\in \widetilde{W}_k^\ast}\E\mathrm{ind}_k(F^\ast,G)}\\
&=1+
\frac{\sum_{F\in \overline{\widetilde{W}_k^\ast}}\Psi_F\nu_k(F,\kappa)\frac{(n)_{|v(F)|}}{n^{|v(F)|}}}{\sum_{F^\ast\in \widetilde{W}_k^\ast}\Psi_{F^\ast}\nu_k(F^\ast,\kappa)\frac{(n)_{|v(F^\ast)|}}{n^{|v(F^\ast)|}}}\cdot
\end{align*}
From Lemma~\ref{dominantscaling}, we know that $\Psi$ is constant over $\widetilde{W}_k^\ast$ and $W_k'$. Thus, for any fixed $F^\ast\in \widetilde{W}_k^\ast$ and $F'\in W_k'$, we may write
\begin{equation}
\label{thm:form:rel-error-split}
\frac{\E\left|\mathcal{W}_k(G)\right|}{\E\left|\{w\in\mathcal{W}_k(G) : F_w \in \widetilde{W}_k^\ast\}\right|} = 1+
\frac{\Psi_{F'}}{\Psi_{F^\ast}}\left(\frac{\sum_{F'\in W_k'}\nu_k(F',\kappa)}{\sum_{F^\ast\in \widetilde{W}_k^\ast}\nu_k(F^\ast,\kappa)}+\epsilon(n;k,\kappa)\right),
\end{equation}
where
\begin{multline}\label{thm:epsilon-def}
\epsilon(n;k,\kappa) =
\left(\frac{\sum_{F'\in W_k'}\nu_k(F',\kappa)\frac{(n)_{|v(F)|}}{n^{|v(F)|}}}{\sum_{F^\ast\in \widetilde{W}_k^\ast}\nu_k(F^\ast,\kappa)\frac{(n)_{|v(F^\ast)|}}{n^{|v(F^\ast)|}}}
-
\frac{\sum_{F'\in W_k'}\nu_k(F',\kappa)\vphantom{\frac{(n)_{|v(F)|}}{n^{|v(F)|}}}}{\sum_{F^\ast\in \widetilde{W}_k^\ast}\nu_k(F^\ast,\kappa)\vphantom{\frac{(n)_{|v(F^\ast)|}}{n^{|v(F^\ast)|}}}}\right)\\
+
\frac{\sum_{F\in W_k\setminus (\widetilde{W}_k^\ast\cup W_k')}\frac{\Psi_{F}}{\Psi_{F'}}\nu_k(F,\kappa)\frac{(n)_{|v(F)|}}{n^{|v(F)|}}}{\sum_{F^\ast\in \widetilde{W}_k^\ast}\nu_k(F^\ast,\kappa)\frac{(n)_{|v(F^\ast)|}}{n^{|v(F^\ast)|}}} .
\end{multline}

We first consider the second term in~\eqref{thm:form:rel-error-split}. A direct consequence of Lemma~\ref{dominantscaling} is that for $k>3$ and any pair $F',F^\ast\in W_k'\times \widetilde{W}_k^\ast$, we have
\begin{equation*}
\begin{dcases}
\frac{\Psi_{F^\ast}}{\Psi_{F'}} \mu^{-1}\to1& \textnormal{if $k$ is odd or $k$ is even and either $|k-k^\ast|\geq2$ or $k= k^\ast$,}\\
\frac{\Psi_{F^\ast}}{\Psi_{F'}} \mu^{-|k-k^\ast|/2}\to1 &  \textnormal{if $k$ is even and $0<|k-k^\ast|<2$.}
\end{dcases}
\end{equation*}
Then, enumerating the cases of Lemma~\ref{dominantscaling} in the context of~\eqref{thm:form:rel-error-split}, we match the first- and second-order terms of the expressions in the statement of Theorem~\ref{thm:phase}.

Next, we consider $\epsilon(n;k,\kappa)$ from~\eqref{thm:epsilon-def}, which is the general form of the $\epsilon_i(n;k,\kappa)$ in the statement of Theorem~\ref{thm:phase}. The upper bound follows from noting that: i) all terms are positive; ii) $\Psi_F/\Psi_{F'}< 1$ by construction; and iii) the map $x:\mapsto {(n)_{x}}/{n^x}$ is decreasing with $x$, and bounded above by unity.

We finally consider the limit of $\epsilon(n;k,\kappa)$ as $n$ tends to infinity when $\mu$ diverges. The first term in parentheses  in~\eqref{thm:epsilon-def} is tending to zero, as $(n)_x/n^x\to1$ as $n\to\infty$ for any fixed $x$. The last term in~\eqref{thm:epsilon-def} is of order $\textstyle \max_{F\in W_k\setminus (\widetilde{W}_k^\ast\cup W_k')}\Psi_F/\Psi_{F'}$. Let $F''\in\textstyle\argmax_{F\in W_k\setminus (\widetilde{W}_k^\ast\cup W_k')}\Psi_F$. Using the argument of Lemma~\ref{dominantscaling}, $F''$ is either a tree or has the same order of magnitude $\Psi$ as a tadpole. If $F'$ and $F''$ are both trees or tadpoles, we directly have $\Psi_{F'}/\Psi_{F''} \geq n\rho = \mu\to\infty$, since $F'$ must contain more edges than $F''$. If instead $F'$ is a tree and $F''$ is a tadpole, we have $\Psi_{F'}/\Psi_{F''} = (n\rho)^l/n$ for some $l>0$, since the tadpole must then contain more edges than the tree. On the other hand, if $F'$ is a tadpole and $F''$ is a tree, we have $\Psi_{F'}/\Psi_{F''} = n/(n\rho)^l$, again for some $l>0$. Recalling that $\sqrt\mu^{k^\ast}=n$, we recover that in both cases $\Psi_{F'}/\Psi_{F''} = \mu^{|l-k^\ast/2|}\to\infty$. The case $2l=k^\ast$ is impossible, since then $\Psi_{F'}=\Psi_{F''}$ and $F''$ would be an element of $W_k'$, which is in contradiction with the definition of $F''$. Finally, in all cases considered $\Psi_{F''}/\Psi_{F'} = o(1)$, so that $\lim_{n\to\infty}\epsilon(n;k,\kappa) = 0$, which concludes the proof.
\end{proof}

Theorem~\ref{thm:phase} describes how the total number of closed walks of even length shorter than $k^\ast$ is dominated (in expectation) by trees or by cycles. The balance between trees and cycles that we describe within generalized random graphs extends a number of known results~\cite{bollobas2009metric,bollobas2010cut,bollobas2011sparse}: i) \citet[Lemma~3.10]{bollobas2009metric} show that if the degrees grow faster than $\sqrt n$, cycles dominate for all $k$, which is consistent with our result, since then $k^\ast <4$; ii) \citet[Lemma~2.10]{bollobas2010cut} show that the number of cycles is influenced more strongly by increased sparsity than is the number of trees; and iii) \citet[Eq.~7]{bollobas2011sparse} establish that generalized random graphs are tree-like when $\mu={\mathcal{O}}(1)$. Finally, in the dense graph regime, \citet[p.~63]{lovasz2012large} relates counts of $k$-cycles to the sum of the $k$th power of the eigenvalues of a graph's adjacency matrix, and hence the expected number of closed $k$-walks (see~\eqref{tracedecomp}).

\begin{Remark}
If $k=2$, $\E\left|\mathcal{W}_2(G)\right|=\E\left|\{w\in\mathcal{W}_2(G) : F_w \equiv K_2\}\right|$, while if $k=3$ we have $\E\left|\mathcal{W}_3(G)\right|=\E\left|\{w\in\mathcal{W}_3(G) : F_w \equiv K_3\}\right|$. Thus, in these two cases, $\widetilde{W}_k^\ast$ does not depend on either the size or the density of the graph $G$.
\end{Remark}

The relative magnitudes of different terms presented in Theorem~\ref{thm:phase} are driven by the growth of $\mu$. The rate at which dominance occurs depends on how quickly $\mu = n \rho$ grows relative to $n$. To emphasize this point, we fix $k=4$ and consider the case where $\smash{\mu=2n^\alpha}>1$ for $\alpha\in [0,1]$. From Theorem~\ref{thm:phase} we know that the rate of convergence towards unity of the ratio ${\E|\mathcal{W}_k(G)|}/{\E|\{w\in\mathcal{W}_k(G) : F_w \in \widetilde{W}_k^\ast\}|}$---which we here denote by $\beta(\alpha;k)$ for convenience---is
\begin{equation*}
\beta(\alpha;k) = \alpha\min(1,|k-k^\ast|/2).
\end{equation*}

\begin{figure}[t]
  \centering
  \includegraphics[width=.5\textwidth]{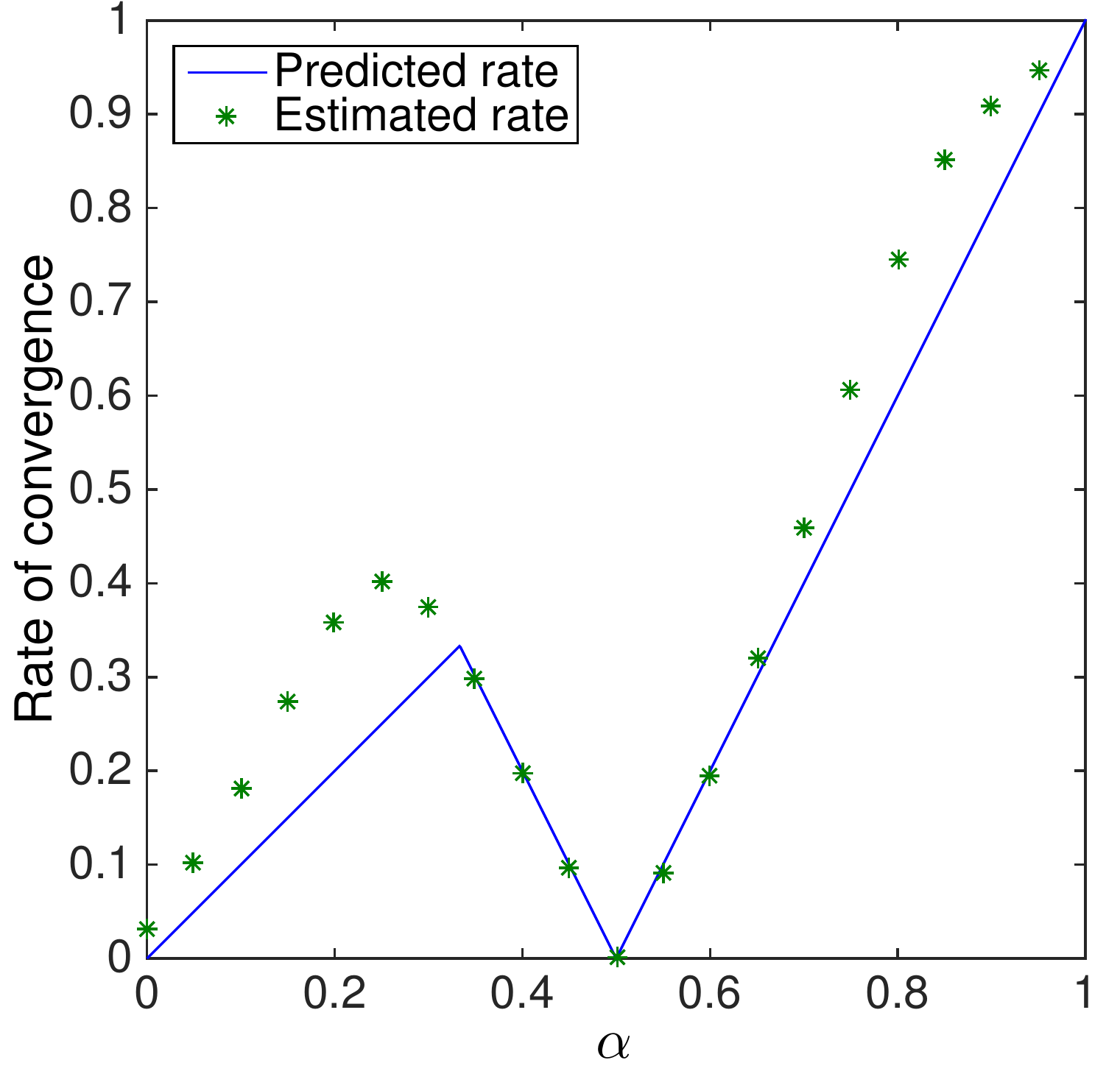}
  \caption{\label{rate-topolo} Predicted asymptotic rates of convergence $\beta(\alpha;4)$ (solid blue line) and rates of convergence estimated in a simulation experiment (green stars).}
\end{figure}

Figure~\ref{rate-topolo} shows an example of $\beta(\alpha;4)$ as a function of $\alpha\in[0,1]$, and compares it to the rate of convergence observed in simulations. For the purposes of this figure, we first fix a kernel $\kappa$ matching a mixed membership model with three communities~\citep{airoldi2008mixed,Onnela2012taxonomies}, in accordance with Definition~\ref{model}. We then let $G_{ij}$ be a random graph distributed according to $\smash{ G(2^j,2(2^j)^{\alpha-1}\kappa) }$. Then, for $\alpha\in\{0,1/20,2/20,\dots,1\}$, we evaluate through repeated simulation the slope of the following set of points:
\begin{equation*}
\left\{\left(\log(2^j),\log\left(\frac{\E|\mathcal{W}_k(G_{ij})|}{\sum_{F\in W_4^\ast}\E\mathrm{ind}_k(F,G_{ij})}-1\right)\right)\right\}_{j\in [7,12]}.
\end{equation*}
We then see from Fig.~\ref{rate-topolo} that the asymptotic rates presented in Theorem~\ref{thm:phase} show reasonable alignment with those estimated through simulation.

\subsection{Domination of walks mapping out cycles}

We now determine dominating sets of non-backtracking, tailless closed walks in the generalized random graph setting. We will see that in this setting cycles dominate, independently of the sparsity the graph, up to the point where the network is no longer connected. As noted in the main text, this result quantifies the observations of~\cite{Krzakala13} for community detection, where non-backtracking, tailless closed walks are observed to improve spectral methods in the sparse graph setting.

We first give definitions for non-backtracking, tailless closed walks that mirror those we have already shown to hold for closed walks in general. We begin by adapting Definition~\ref{closedwalks} to non-backtracking, tailless closed walks. We recall that a closed walk is non-backtracking if it never visits the same edge twice in succession. It is tailless if the first and last edges it traverses are different. We then proceed to adapt all the definitions and results presented in the previous section.

\begin{Definition}[Set $\smash{\mathcal{W}^b_k}(G)$ of non-backtracking, tailless closed $k$-walks]\label{def-non-back}
Fix a simple graph $G$ and $k\in \mathbb{N}$. A non-backtracking, tailless closed walk $w$ of length $k$ in $G$ is a sequence of adjacent vertices in $G$
satisfying
\begin{equation*}
\forall 0\leq i\leq k, \quad v_i\neq v_{(i+2)\!\!\!\!\!\!\!\mod k},
\end{equation*}
where $ \cup_{i=0}^k \{v_i\} \subset v(G) $, $v_0=v_k$ and $ \cup_{i=0}^{k-1} \{ v_i v_{i+1} \} \subset e(G) $.
We denote by $\smash{\mathcal{W}^b_k}(G)$ the set of all non-backtracking, tailless closed $k$-walks.
\end{Definition}

In analogy to Definition~\ref{closedalkdef}, we introduce the set of subgraphs induced by non-backtracking, tailless closed walks.

\begin{Definition}[Set $\smash{W_k^b}$ of unlabeled graphs induced by non-backtracking, tailless closed $k$-walks]\label{defWkb}
We denote by $W_k^b$ the set of all unlabeled graphs induced by non-backtracking, tailless closed walks of length $k$:
\begin{equation*}
W_k^b=\{F\in W_k : \exists w\in\smash{\mathcal{W}^b_k}(K_k),\,F\equiv F_w\}.
\end{equation*}
\end{Definition}

We start by introducing three basic properties of such walks.
\begin{Lemma}[Properties of $W_k^b$]\label{w_k^b}
Fix $k>2$. Then, $W_k^b$ is non-empty, and any $F\in W_k^b$ verifies the following properties:
\begin{enumerate}
\item Every node in $F$ has degree at least $2$, and hence $|e(F)|\geq |v(F)|$.
\item If $|e(F)|=|v(F)|$, then $F\equiv C_{|v(F)|}$ and $|v(F)|$ is a divisor of $k$ greater than or equal to $3$.
\item If $|e(F)|=|v(F)|+1=k$, then $k\geq 6$ and there exists $p\in \{3,4,\dots,\lfloor{k/2}\rfloor\}$ such that $F\equiv C_pC_{k-p}$.
\end{enumerate}
\end{Lemma}
\begin{proof}
We prove the stated results in succession.

\vspace{.5\baselineskip}
\emph{Proof of 1:} 
Since a non-backtracking, tailless closed walk is a walk and therefore the graph it induces connected, we immediately conclude by Lemma~\ref{wk-props} that all nodes in $F$ have positive degree. We proceed to establish the result by contradiction. Assume $F$ contains a node of degree one, and call it $o$. Furthermore call $t$ the only node connected to $o$. Fix $w$ to be any non-backtracking, tailless closed walk such that $F_w=F$. Then $tot$ must be part of $w$, otherwise $w$ cannot be a closed walk such that $F_w=F$. However, since $w$ is non-backtracking, this is not permitted, and we obtain a contradiction. We can therefore conclude that no node can have degree unity; i.e., in a graph where every node can be visited by a non-backtracking, tailless closed walk, no pendant nodes can be present. Hence, since every $F\in W_k^b$  is connected, and $F$ cannot be a tree, it follows that $|e(F)|\geq |v(F)|$.

\vspace{.5\baselineskip}
\emph{Proof of 2:} 
We call $d_i$ the degree of the node labeled $i$ in $F$.
Since $|e(F)|=|v(F)|$, we have
\begin{equation}\label{w_k^b_item2}
|v(F)|^{-1}\sum_{i\in v(F)} d_i = \frac{2|e(F)|}{|v(F)|} = 2.
\end{equation}
On the other hand, from Item~1 of Lemma~\ref{w_k^b}, for all $i\in v(F)$, $d_i\geq2$. Thus, the equality of~\eqref{w_k^b_item2} is possible only if for all $i$, $d_i = 2$. Then all nodes in $F$ have degree two, and since $F$ is connected by Item~2 of Lemma~\ref{wk-props}, we conclude that $F\equiv C_{|v(F)|}$. Finally, for a non-backtracking, tailless closed $k$-walk to induce $C_{|v(F)|}$, it must traverse $C_{|v(F)|}$ one or more times. Hence, $|v(F)|$ must divide $k$.

\vspace{.5\baselineskip}
\emph{Proof of 3:} 
Fix $w$ to be any non-backtracking, tailless closed walk such that $F_w=F$. Since $|v(F)|=k-1$, $w$ must visit one node in $v(F)$ exactly twice, and all the other nodes in $v(F)$ exactly once. Call $o$ the twice visited node. As cyclic permutations of $w$ are still non-backtracking, tailless closed walks, assume without loss of generality that $w$ starts at $o$ and write $w = ov_1\cdots v_{k-1}o$. Let $t\in\{3,\dots,k-3\}$ be such that $v_t=o$. Since $w$ visits $o$ exactly twice, and a non-backtracking, tailless closed walk cannot return to $o$ in less than three steps, $t$ exists and is unique. Since $t$ exists, $k-3\geq3$, and hence $k\geq 6$.

Let $w_1 = ov_1\cdots v_{t-1}o$ and $w_2 = ov_{t+1}\cdots v_{k-1}o$. Then, $w_1$ and $w_2$ are two non-backtracking, tailless closed walks. Since $|e(F)|=k$, $w$ visits a different edge at each step, $w_1$ and $w_2$ also visit different edges at each step. Thus, $|e(F_{w_1})|=t=|v(F_{w_1})|$, and by Item~2 of Lemma~\ref{w_k^b}, $F_{w_1}\equiv C_t$. By the same argument, $F_{w_2}\equiv C_{k-t}$. Since $F_{w_1}\cap F_{w_2} = (\{o\},\emptyset)$, we have that $F\equiv C_tC_{k-t}$. Finally, we conclude that $t\in\{3,\dots,\lfloor{k/2}\rfloor\}$, since $3\leq\min\{t,k-t\}\leq\lfloor{k/2}\rfloor$.
\end{proof}

In the same fashion as for $\smash{W_k}$ in the case of closed $k$-walks, $\smash{W_k^b}$ allows us to partition the set of all non-backtracking, tailless closed $k$-walks $\smash{\mathcal{W}_k^b(G)}$ in any simple graph $G$ as follows:
\begin{equation}
\label{eq:non-back-walk-partition}
\mathcal{W}_k^b(G) = \bigsqcup_{F \in W_k^b} \left\{ w\in\mathcal{W}_k^b(G) : F_w \equiv F \right\} .
\end{equation}
This leads naturally to the following definitions, mirroring Definitions~\ref{defnocop} and~\ref{def:wkast2}.

\begin{Definition}[Number $\mathrm{ind}_k^b(F,G)$ of non-backtracking, tailless closed $k$-walks in $G$ inducing an isomorphic copy of $F$]\label{defnocopb}
Fix a walk length $k$ and two simple graphs $F$ and $G$. We write
\begin{equation*}
\mathrm{ind}_k^b(F,G) = \left| \left\{ w\in\mathcal{W}_k^b(G) : F_w \equiv F \right\} \right|
\end{equation*}
for the number of non-backtracking, tailless closed walks of length $k$ in $G$ that induce a subgraph of $G$ isomorphic to $F$.
\end{Definition}

In analogy to Definition~\ref{def:wkast}, we introduce the sets of dominating non-back\-tracking, tailless walk-induced subgraphs in generalized random graphs. Recall from Definition~\ref{def:varphi} that for fixed $n$ and $\rho$, the map $\Psi$ is such that for a graph $F$, $\Psi_F=n^{|v(F)|}\rho^{|e(F)|}$.

\begin{Definition}[Sets $\widetilde{W}_k^{b\ast},  W_k^b{}'$ of dominating subgraphs induced by non-backtracking, tailless closed walks]\label{def:wkast:back} Fix a walk length $k\in \mathbb{N}$ and let
\begin{align*}
\widetilde{W}_k^{b\ast} &= \{F \in W_k^{b\phantom{\ast}} : \Psi_F = \max_{F'\in W_k^{b\phantom{\ast}}} \Psi_{F'}\},\\
 W_k^b{}' &= \{F \in \overline{\widetilde{W}_k^{b\ast}} : \Psi_F = \max_{F'\in \overline{\widetilde{W}_k^{b\ast}}} \Psi_{F'}\}; \quad \overline{\widetilde{W}_k^{b\ast}} = W_k^b \setminus \widetilde{W}_k^{b\ast}.
\end{align*}
\end{Definition}

From~\eqref{eq:non-back-walk-partition} we recover the partition
\begin{equation}
\label{eq:non-back-walk-sum}
\left| \mathcal{W}_k^b(G) \right| = \sum_{F \in \widetilde{W}_k^{b\ast}} \mathrm{ind}_k^b(F,G) + \sum_{F \in \overline{\widetilde{W}_k^{b\ast}}} \mathrm{ind}_k^b(F,G).
\end{equation}
We now characterize the sets $\smash{\widetilde{W}_k^{b\ast}, W_k^b{}'}$ so as to control the two terms on the right-hand side of~\eqref{eq:non-back-walk-sum}. The following result is analogous to Lemma~\ref{dominantscaling}.

\begin{Lemma}\label{non-back:w} 
Consider the setting of Definition~\ref{model}. Fix an integer $k >3$ and a scalar $\rho\in(0,\|\kappa\|_{\infty}^{-1})$ such that $n\rho>1$ for all integers $n>k$. As in Lemma~\ref{dominantscaling}, set $k^\ast = \log n / \log \sqrt{n\rho} $,
Recall that $C_k$ is the $k$-cycle and $C_kC_l$ is the $(k,l)$-lemniscate.

Finally, denote by $h_k$ the largest divisor of $k$ in $\{3,4,\dots,\max(3,\lfloor k/2 \rfloor\}$, and if there is no such divisor (i.e., if $k$ is prime or equal to $4$), then set $h_k$ to $-\infty$. Then
\begin{align}
\nonumber
\widetilde{W}_k^{b\ast} &= \{C_k\}\\
\label{W_k_b_prime}
 W_k^b{}' &=
	\begin{cases}
	\emptyset&\textnormal{if } k\leq5,\\
	\{C_{h_k}\}&\textnormal{if } k-h_k<k^\ast/2,\\
	\{C_{h_k}\}\cup\{C_pC_{k-p} : p\in \{3,4,\dots,\lfloor{k/2}\rfloor\}\}&\textnormal{if } k-h_k=k^\ast/2,\\
	\{C_pC_{k-p} : p\in \{3,4,\dots,\lfloor{k/2}\rfloor\}\}&\textnormal{otherwise.}
	\end{cases}
\end{align}
\end{Lemma}

\begin{proof}
The proof parallels that of Lemma~\ref{dominantscaling}. The key observation is that $\Psi_F$ is monotone increasing in $|v(F)|$ for fixed $|e(F)|-|v(F)|$, and monotone decreasing in $|e(F)|-|v(F)|$ for fixed $|v(F)|$. We enumerate the possible cases as follows:
\begin{enumerate}
\item Suppose $|e(F)|=|v(F)|$. Then $ \Psi_F = (n\rho)^{|v(F)|} $.  By monotonicity, we have
\begin{equation*}
\max_{F \in W_k^b : |e(F)| = |v(F)|} \Psi_F = (n \rho)^k,
\end{equation*}
achieved uniquely by the $k$-cycle $C_k$, since no other non-backtracking, tailless closed $k$-walk visits $k$ vertices.
\item Suppose $|e(F)|>|v(F)|$. Then, since $\rho<\|\kappa\|_{\infty}^{-1}\leq 1$,
\begin{equation*}
 \Psi_F = (n \rho)^{|v(F)|} \rho^{|e(F)|-|v(F)|} < (n \rho)^{|v(F)|} \leq (n \rho)^k.
\end{equation*}
This last upper bound is not attained, as $|e(F)|>|v(F)|$ and $\rho<1$.
\end{enumerate}
Finally, in contrast to Lemma~\ref{dominantscaling}, we need not consider the case in which $|e(F)|<|v(F)|$. This follows from the observation discussed in Item~1 of Lemma~\ref{w_k^b}: A closed walk inducing a tree will backtrack at each of this tree's leaves, since these have degree $1$. Hence, a non-backtracking, tailless closed walk cannot induce a tree, and thus $|e(F_w)|\geq|v(F_w)|$ for any non-backtracking, tailless closed walk $w$. Thus we conclude $\widetilde{W}_k^{b\ast} = \{C_k\}$ as claimed.

We now consider $W_k^b{}'$. First, if $k\leq 5$, then $W_k^b = \{C_k\}$. Thus, $W_k^b\setminus \widetilde{W}_k^{b\ast} = \emptyset$, and
$ W_k^b{}' = \emptyset$, as per Definition~\ref{def:wkast:back}, and so we obtain the first case in~\eqref{W_k_b_prime}.

Now, if instead $k>5$, we show that there are only two subgraphs to consider. Paralleling Case~1 above, we must consider the largest cycle different from $C_k$ that can be induced by a non-backtracking, tailless closed $k$-walk. From Item~2 of Lemma~\ref{w_k^b}, we know that this subgraph is $C_{h_k}$. As for Case~2, since $|e(F)|-|v(F)|\geq1$ and $|v(F)|<k$, we have
\begin{equation*}
\Psi_F = (n \rho)^{|v(F)|} \rho^{|e(F)|-|v(F)|} \leq (n \rho)^{|v(F)|}\rho \leq (n \rho)^{k-1}\rho.
\end{equation*}
By Item~3 of Lemma~\ref{w_k^b}, this upper bound is attained solely by ${\Psi_{C_pC_{k-p}}}$ for $p\in \{3,4,\dots,\lfloor{k/2}\rfloor\}$. Thus, we only need to consider and compare $\smash{\Psi_{C_{h_k}} = (n\rho)^{h_k}}$ and ${\Psi_{C_pC_{k-p}} = (n \rho)^{k-1}\rho}$. Since $(n\rho)^{k^\ast/2}=n$, we have equality if $(n\rho)^{k-h_k} = (n \rho)^{k^\ast/2}$, and we obtain the last three items in~\eqref{W_k_b_prime}. This completes the proof. \end{proof}

We are now equipped to describe how non-backtracking, tailless closed walks are dominated by cycles, paralleling Theorem~\ref{thm:phase}. Here we will prove that the ratio $\smash{\E|\mathcal{W}^b_k(G)|}/ \smash{\E|\{w\in\mathcal{W}^b_k(G) : F_w\equiv C_k\}|}$ tends to $1$. It is straightforward to show that for $k\in\{3,4,5\}$ we have exact equality---i.e., $\smash{\mathcal{W}^b_k(G)}$ is equal to the set $\smash{\{w\in\mathcal{W}^b_k(G) : F_w\equiv C_k\}}$---and thus in what follows we assume $k\geq 6$.

\begin{Theorem}\label{thm:back}
Fix a walk length $k\geq 6$ and a network size $n>k$. Let $G$ be distributed according to the generalized random graph model $G(n,\rho\kappa)$ from Definition~\ref{model}, and assume $n\rho>1$. Set $\mu=n\rho$ and let $h_k$ be the largest divisor of $k$ in $\{3,4,\dots,\lfloor k/2 \rfloor\}$. If there is no such divisor, i.e., if $k$ is prime, set $h_k$ to $-\infty$. Then:
\begin{multline*}
\frac{\E\left|\mathcal{W}^b_k(G)\right|}
{\E\left|\{w\in\mathcal{W}^b_k(G) : F_w\equiv C_k\}\right|}=1+\\
\begin{dcases}
\frac{1}{\mu^{k-h_k}}
\left(\frac{\nu_k(C_{h_k},\kappa)}{\nu_k(C_k,\kappa)}+\epsilon^b_1(n;k,\kappa)\right)
&
\textnormal{if $k-h_k<k^\ast/2$},\\
\frac{1}{n}
\left(\frac{\nu_k(C_{h_k},\kappa)+\sum_{p\in \{3,4,\dots,\lfloor{k/2}\rfloor\}}\nu_k(C_pC_{k-p},\kappa)}{\nu_k(C_k,\kappa)}\right. &\\ \hspace{6.5cm}+\epsilon^b_2(n;k,\kappa)\bigg)
&
\textnormal{if $k-h_k=k^\ast/2$},\\
\frac{1}{n}
\left(\frac{\sum_{p\in \{3,4,\dots,\lfloor{k/2}\rfloor\}}\nu_k(C_pC_{k-p},\kappa)}{\nu_k(C_k,\kappa)}+\epsilon^b_3(n;k,\kappa)\right)
&
\textnormal{otherwise.}
\end{dcases}
\end{multline*}
In all cases $\epsilon^b_i(n;k,\kappa)$ is bounded by $\smash{n^k\textstyle{\sum}_{F\in W_k^b\setminus \widetilde{W}_k^{b\ast}} \nu_k(F,\kappa)/(n)_k\textstyle{\sum}_{F\in \widetilde{W}_k^{b\ast}} \nu_k(F,\kappa)}$, and if $\mu$ diverges in $n$, then $\lim_{n\to\infty}\epsilon^b_i(n;k,\kappa) = 0$.
\end{Theorem}

\begin{proof}
The proof proceeds in the same fashion as that of Theorem~\ref{thm:phase}. We express the ratio ${\E|\mathcal{W}^b_k(G)|}/
{\E|\{w\in\mathcal{W}^b_k(G) : F_w\equiv C_k\}|}$ using the partition induced by $\widetilde{W}_k^{b\ast}$ and $ W_k^b{}'$, and then control the error terms by using Lemma~\ref{non-back:w}.

To proceed we introduce the following notation:
\begin{equation*}
\nu_k^b(F,\kappa) = \frac{\mathrm{ind}_k^b(F,F)}{\mathrm{aut}(F)}s(F,\kappa).
\end{equation*}
This notation will only be used in this proof, and parallels the introduction of $\nu_k(F,\kappa)$ in Definition~\ref{def:nu_k}. To write the theorem statement using only $\nu_k$ notation, we will use that $\nu_k^b(C_k,\kappa) = \nu_k(C_k,\kappa)$ and that $\nu_k^b(C_pC_{k-p},\kappa) = \nu_k(C_pC_{k-p},\kappa)$ for $p\in \{3,4,\dots,\lfloor{k/2}\rfloor\}$.

We first write ${\E|\mathcal{W}^b_k(G)|}/
{\E|\{w\in\mathcal{W}^b_k(G) : F_w\equiv C_k\}|}$ using~\eqref{eq:non-back-walk-sum}, which we recall here:
\begin{equation*}
\left| \mathcal{W}_k^b(G) \right| = \sum_{F \in \widetilde{W}_k^{b\ast}} \mathrm{ind}_k^b(F,G) + \sum_{F \in \overline{\widetilde{W}_k^{b\ast}}} \mathrm{ind}_k^b(F,G).
\end{equation*}
By taking expectations and considering the ratio of the terms on the right-hand side of this expression, we obtain:
\begin{align}
\nonumber
\frac{\E\left| \mathcal{W}_k^b(G) \right|}{\E\left|\{w\in\mathcal{W}^b_k(G) : F_w\equiv C_k\}\right|}
&= 1 + \frac{ \sum_{F \in \overline{\widetilde{W}_k^{b\ast}}} \E\mathrm{ind}_k^b(F,G)}{\E\mathrm{ind}_k^b(C_k,G)}\\
\label{thm:back:form:rel-error-split}
&= 1 + \frac{ \sum_{F \in \overline{\widetilde{W}_k^{b\ast}}} \Psi_F\nu_k^b(F,\kappa)}{\Psi_{C_k}\nu_k^b(C_k,\kappa)}.
\end{align}

Then, we use Lemma~\ref{non-back:w} to control $\Psi_{F^\ast}/\Psi_{F'}$ for all pairs of subgraphs $F^\ast\in \widetilde{W}_k^{b\ast}$ and $F'\in W_k^b{}'$. This parallels our use of Lemma~\ref{dominantscaling} in the proof of Theorem~\ref{thm:phase}. From Lemma~\ref{non-back:w}, we have that for all pairs of subgraphs $F^\ast\in \widetilde{W}_k^{b\ast}$ and $F'\in W_k^b{}'$, the ratio $\Psi_{F^\ast}/\Psi_{F'}$ is determined by the number of edges and vertices of $F^\ast$ and $F'$. We note from Lemma~\ref{non-back:w} (referring back to the definition of $h_k$) that it follows that $|v(F_w)|=|e(F_w)|=k$. For $k\geq 6$, if $k<k^\ast/2+h_k$ then $|v(F_w')|=|e(F_w')|=h_k$, while if $k> k^\ast/2+h_k$ then $|v(F_w')|=k-1$ and $|e(F_w')|=k$. In conjunction with the form of $\Psi_F$, this implies that
\begin{equation}\label{phi-ratio}
\begin{dcases}
\frac{\Psi_{F^\ast}}{\Psi_{F'}}\mu^{h_k-k}\to 1&\textnormal{if $k<k^\ast/2+h_k$},\\
\frac{\Psi_{F^\ast}}{\Psi_{F'}}n^{-1}\to 1 & \textnormal{otherwise}.
\end{dcases}
\end{equation}

We then fix $F^\ast\in \widetilde{W}_k^{b\ast}$ and $F'\in W_k^b{}'$. Using both~\eqref{thm:back:form:rel-error-split} and~\eqref{phi-ratio} yields---in the same fashion as in~\eqref{thm:form:rel-error-split} and~\eqref{thm:epsilon-def}---the following expression:
\begin{equation}
\label{thm:back:form:rel-error-split-2}
\frac{\E\left| \mathcal{W}_k^b(G) \right|}{\E\left|\{w\in\mathcal{W}^b_k(G) : F_w\equiv C_k\}\right|}
= 1+ \frac{\Psi_{F^\ast}}{\Psi_{F'}}\left(\frac{\sum_{F' \in  W_k^b{}'} \nu_k^b(F',\kappa)}{\nu_k^b(C_k,\kappa)}+\epsilon^b(n;k,\kappa)\right),
\end{equation}
where
\begin{multline}\label{thm:epsilon-b-def}
\epsilon^b(n;k,\kappa) =
\left(\frac{\sum_{F'' \in  W_k^b{}'}\nu_k^b(F'',\kappa)\frac{(n)_{|v(F'')|}}{n^{|v(F'')|}}}{\nu_k^b(C_k,\kappa)\frac{(n)_k}{n^k}}-\frac{\sum_{F'' \in  W_k^b{}'}\nu_k^b(F'',\kappa)}{\nu_k^b(C_k,\kappa)}\right)
+\\
\frac{\sum_{F \in W_b^k\setminus(\{C_k\}\bigcup  W_k^b{}')}\frac{\Psi_F}{\Psi_{F'}} \nu_k^b(F,\kappa)\frac{(n)_{|v(F)|}}{n^{|v(F)|}}}{\nu_k^b(C_k,\kappa)\frac{(n)_k}{n^k}}.
\end{multline}
This is the general form of the $\epsilon_i^b(n;k,\kappa)$ in the statement of Theorem~\ref{thm:back}. We note that~\eqref{thm:back:form:rel-error-split-2} provides all of the cases of Theorem~\ref{thm:back}, using the forms of $\widetilde{W}_k^{b\ast}$ and $ W_k^b{}'$ provided in Lemma~\ref{non-back:w}.

We now bound $\epsilon^b(n;k,\kappa)$ and compute its order in the case where $\mu\to\infty$. Our method of bounding exactly parallels that of $\epsilon(n;k,\kappa)$ in the proof of Theorem~\ref{thm:phase}. The upper bound follows from noting that: i) all terms are positive; ii) $\Psi_F/\Psi_{F'}< 1$ by construction for $F\in W_b^k\setminus(\{C_k\}\cup  W_k^b{}')$; and iii) the map $x\mapsto {(n)_{x}}/{n^x}$ is decreasing in $x>0$, and bounded above by unity.  Now consider the order of $\epsilon^b(n;k,\kappa)$ as $n$ tends to infinity when $\mu$ diverges. The first ratio in~\eqref{thm:epsilon-b-def} tends to zero, since $(n)_x/n^x\to1$ as $n\to\infty$ for any fixed $x>0$. The last term in~\eqref{thm:epsilon-b-def} has the same order of magnitude as $\smash{ \textstyle{\max}_{F\in W_b^k\setminus(\{C_k\}\bigcup  W_k^b{}')}\Psi_F/\Psi_{F'}}$. Following the arguments of Lemma~\ref{non-back:w}, in the same fashion as in Theorem~\ref{thm:phase}, this ratio is smaller than $\mu^{-1}$. Thus we obtain the claimed order of $\epsilon^b(n;k,\kappa)$, thereby concluding the proof.
\end{proof}

\section{Dominating walks in scale-free random graphs}\label{sec:S5}
\subsection{Scale-free inhomogeneous random graphs}
We next determine the set of dominating closed walks for the case of inhomogeneous random graphs with power-law degree distributions~\citep{Soderberg2002inhomogeneous,olhede2012degree}.

\begin{Definition}[Inhomogeneous random graph with power-law degrees]\label{powerlaw}
Fix an exponent $\gamma>0$ and a scalar $\theta \in (0,1]$. Let $G$ be the simple random graph whose adjacency matrix $A$ has elements generated independently according to
\begin{equation*}
\mathbb{P}(A_{ij}=1)=\theta^2 \cdot \left( ij\right)^{-\gamma},\quad 1\leq i<j\leq n,
\end{equation*}
with $A_{ji} =A_{ij}$ for $1\leq i<j\leq n$, and $A_{ii} =0$ for $1\leq i\leq n$.
\end{Definition}

The following proposition describes the expected number of unlabeled copies of a fixed graph in networks generated according to Definition~\ref{powerlaw}.

\begin{Proposition}\label{ihrg-ptop}
Fix an exponent $\gamma > 0$ and a sequence $\{\theta_n\}$ taking values in $(0,1]$. Consider a sequence of random graphs $\{G_n\}$ on $n$ nodes, generated according to Definition~\ref{powerlaw}. Fix a graph $F$ on $v$ vertices with strictly positive degrees $d=\left(d_1,\dots,d_v\right)$ enumerated in non-decreasing order, with
\begin{equation*}
\begin{cases}
d_t \gamma <1 & \text{for $t\in \{1,\dots,q\}$, setting $q=0$ if $\{ t: d_t \gamma < 1\} = \emptyset$,} \\
d_t \gamma =1 & \text{for $t\in\{q+1,\dots,h\}$, setting $h=q$ if $\{ t: d_t \gamma = 1\} = \emptyset$,} \\
d_t \gamma >1 & \text{for $t\in\{h+1,\dots,v\}$, setting $v=h$ if $\{ t: d_t \gamma > 1\} = \emptyset$.}
\end{cases}
\end{equation*}
Then, the expected number of unlabeled copies of $F$ in $G_n$ takes the following form:
\begin{equation*}
\E X_F(G_n)\!=\!\frac
		{\left\{C_\gamma(d_{h+1},\dots,d_v)\right\}^{1_{\{h<v\}}}}
		{\mathrm{aut}(F)\left\{\prod_{t=1}^q(1-d_t\gamma)\right\}^{1_{\{q>0\}}}}
	\theta_n^{\sum_{i=1}^v d_i} n^{q-\sum_{t=1}^qd_t \gamma} (\log n)^{h-q}
	\left\{1\!+\!\epsilon_\gamma(d)\right\},
\end{equation*}
with a multiplicative error term given by
\begin{equation*}\epsilon_\gamma(d)
	= {\mathcal{O}}\left(
		1_{\{q>0\}} n^{-(1-d_q\gamma )}+1_{\{h<v\}} n^{-(d_{h+1}\gamma-1)}
		+1_{\{q<h<v\}} (\log n)^{-1}
	\right).
\end{equation*}
Whenever $h<v$, $C_{\gamma}(\cdot)$ is defined inductively, with $\zeta(\cdot)$ the Riemann zeta function:
\begin{align*}
C_\gamma(d_{h+1}) &= \zeta(d_{h+1}\gamma),\\
C_\gamma(d_{h+1},\dots,d_v)
	&=
	C_\gamma(d_{h+1},\dots,d_{v-1})\zeta(d_v\gamma)\\
	&\qquad-\sum_{j=h+1}^{v-1}
	C_\gamma(d_{h+1},\dots,d_{j-1},d_j+d_v,d_{j+1}, \dots,d_v).
\end{align*}
\end{Proposition}

\begin{proof}
Since $X_F(G_n) = \left| \left\{ F'\subset G_n : F'\equiv F \right\} \right|$ for any fixed graph $F$, we may write
\begin{align}
\nonumber
X_F(G_n) & = \sum_{F'\subset K_n}1_{\{F' \equiv F\}}1_{\{F'\subset G_n\}}
\\ \Rightarrow \E X_F(G_n) & = \sum_{F'\subset K_n \,:\, F'\equiv F} \operatorname{\mathbb{P}} (F'\subset G_n),
\label{eq:XFGirg}
\end{align}
with $K_n$ the complete graph on $v(G_n)$. As in the case of generalized random graphs with bounded kernels, we calculate $\operatorname{\mathbb{P}} (F'\subset G_n)$ using the fact that edges form independently under the assumed model of Definition~\ref{powerlaw}.

Specifically, if we label the vertex set of $F'\equiv F$ as $\{i_1, \ldots, i_v\}$, so as to correspond to our chosen ordering $\left(d_1,\dots,d_v\right)$ of the degrees of $F$, then for $n \geq |v(F)|$,
\begin{align*}
\operatorname{\mathbb{P}} (F'\subset G_n)
&= \prod_{i_si_t\in e(F')}\theta_n^2(i_si_t)^{-\gamma}\\
&= \theta_n^{\sum_{j=1}^v d_j}\prod_{i_si_t\in e(F')}(i_si_t)^{-\gamma}\\
&= \theta_n^{\sum_{j=1}^v d_j}\prod_{t=1}^v\prod_{s \,:\, i_ti_s\in e(F')} i_t^{-\gamma}\\
&= \theta_n^{\sum_{j=1}^v d_j}\prod_{t=1}^vi_t^{-d_t\gamma}.
\end{align*}
Summing $\operatorname{\mathbb{P}} (F'\subset G_n)$ over all $F'\subset K_n : F' \equiv F$ as per~\eqref{eq:XFGirg}, we obtain
\begin{equation}\label{EXFG}
\E X_F(G_n) = \frac{\theta_n^{\sum_{j=1}^v d_j}}{{\mathrm{aut}}(F)} \sum_{r_1\in[n]} \sum_{r_2\in[n]/\{r_1 \}} \dots \sum_{r_v\in[n]/\{r_1,\dots, r_{v-1} \}} \prod_{i=1}^v  r_i ^{-d_i\gamma}.
\end{equation}

To derive from~\eqref{EXFG} the stated result of the proposition, we establish by induction the following result:
\begin{align}
\nonumber
S_\gamma(d_1,\ldots,d_v)&=
	\sum_{r_1\in[n]}
    \sum_{r_2\in[n]/\{r_1 \}}
	\dots
	\sum_{r_v\in [n]\backslash \{r_1,\dots, r_{v-1}\}}\,
	\prod_{i=1 }^v  r_i^{-d_t \gamma}\\
\label{induct}
&=\left(\frac{\left[C_\gamma(d_{h+1},\dots,d_v)\right]^{1_{\{h<v\}}}}
	{\left[\prod_{t=1}^q(1-d_t\gamma)\right]^{1_{\{q>0\}}}}\right)
n^{q-\sum_{t=1}^qd_t \gamma} (\log n)^{h-q}
	\left\{1+\epsilon_\gamma(d)\right\},
\end{align}
where $\epsilon_\gamma(d)$ is defined in the statement of the proposition.

To prove~\eqref{induct} by induction, we define the statement $P(v)$ to be the following: $S_\gamma(d_1,\ldots,d_v)$ takes the form of~\eqref{induct} for any tuple $(d_1,\ldots,d_v)$ of $v$ strictly positive integers in non-decreasing order. 

We start by proving $P(1)$. From~\cite[Equation~3.5]{olhede2012degree},$\!$%
\begin{equation}
\label{eq:pows}
S_\gamma(d_1)=\sum_{r_1\in[n]} r_1^{-d_1 \gamma}=
\begin{dcases}
\left(1- d_1 \gamma\right)^{-1} n^{1-d_1 \gamma}+\mathcal{O}(1)
	& \text{if } d_1 \gamma<1,\\
\log n+\gamma_E+\mathcal{O}\left(n^{-1}\right)
	& \text{if } d_1 \gamma=1,\\
\zeta( d_1 \gamma)+\mathcal{O}\left(n^{-(d_1 \gamma-1)}\right)
	& \text{if } d_1 \gamma>1;
\end{dcases}
\end{equation}
where $\gamma_E$ is the Euler--Mascheroni constant and $\zeta(\cdot)$ is the Riemann zeta function. Thus we conclude from~\eqref{eq:pows} directly that $P(1)$ holds. 

To make the general form of our inductive argument clear, we show how $P(1)$ implies $P(2)$. To do so, we first express $P(2)$ as follows: 
\begin{align*}
S_\gamma\left(d_1,d_2\right)
&=
\sum_{r_1\in[n]} r_1^{-d_1 \gamma}\left[\sum_{r_2\in[n]\setminus\{r_1\}} r_2^{-d_2 \gamma}\right]\\
&=\sum_{r_1\in[n]} r_1^{-d_1 \gamma}
	\left[\sum_{r_2\in[n]} r_2^{-d_2 \gamma}-r_2^{-d2\gamma}\right]\\
&=\left[\sum_{r_1\in[n]} r_1^{-d_1 \gamma}\right]
	\left[\sum_{r_2\in[n] } r_2^{-d_2 \gamma}\right] -\left[\sum_{r_3\in[n]} r_3^{-(d_1+d_2)\gamma}\right].
\end{align*}
We now apply~\eqref{eq:pows} to each term to obtain:
\begin{align*}
S_\gamma\left(d_1,d_2\right) =& 
\Bigg[
\Big\{\left(1-d_1\gamma\right)^{-1} n^{1-d_1 \gamma}+\mathcal{O}(1)\Big\}1_{\{d_1 \gamma<1\}}\\
&\qquad\qquad+
\Big\{\log n +\gamma_E+\mathcal{O}\left(n^{-1}\right)\Big\}1_{\{d_1 \gamma=1\}}\\
&\qquad\qquad+
\Big\{\zeta(d_1\gamma)+\mathcal{O}\left(n^{-(d_1 \gamma-1)}\right)\Big\}1_{\{d_1 \gamma>1\}}
\Bigg]\\
&\qquad\cdot
\Bigg[
\Big\{\left(1- d_2 \gamma\right)^{-1} n^{1-d_2 \gamma}+\mathcal{O}(1)\Big\}1_{\{d_2 \gamma<1\}}\\
&\qquad\qquad+
\Big\{\log n +\gamma_E+\mathcal{O}(n^{-1})\Big\}1_{\{d_2 \gamma=1\}}\\
&\qquad\qquad+
\Big\{\zeta( d_2 \gamma)+\mathcal{O}\left(n^{-(d_2 \gamma-1)}\right)\Big\}1_{\{d_2 \gamma>1\}}
\Bigg]\\
&\qquad-
\Bigg[
\Big\{\left(1- (d_1+d_2) \gamma\right)^{-1} n^{1-(d_1+d_2) \gamma}+\mathcal{O}(1)\Big\}1_{\{(d_1+d_2) \gamma<1\}}\\
&\qquad\qquad+
\Big\{\log n +\gamma_E+\mathcal{O}\left(n^{-1}\right)
	\Big\}1_{\{(d_1+d_2) \gamma=1\}}\\
&\qquad\qquad+
\Big\{\zeta( (d_1+d_2) \gamma)+\mathcal{O}\left(n^{-((d_1+d_2) \gamma-1)}\right)\Big\}1_{\{(d_1+d_2 )\gamma>1\}}\Bigg].
\end{align*}
Since $d_1\leq d_2$ by hypothesis, when multiplying the terms above we may reduce the number of distinct cases to six. We then obtain that $S_\gamma(d_1,d_2)$ is equal to
\begin{equation*}
\begin{dcases}
\left(1- d_1 \gamma \right)^{-1}
\left(1- d_2 \gamma \right)^{-1} n^{2-(d_1+d_2) \gamma}\left\{1+\mathcal{O}\left(n^{-(1-d_2 \gamma)}\right) \right\}
	& {\!\!\mathrm{if}}\ d_1\gamma<1,\; d_2\gamma<1
,\\
\left(1- d_1 \gamma \right)^{-1} n^{1-d_1 \gamma}
\log n \left\{1+\mathcal{O}\left((\log n)^{-1}\right) \right\}
	& {\!\!\mathrm{if}}\ d_1\gamma<1,\; d_2\gamma=1
,\\
\left(1- d_1 \gamma \right)^{-1} n^{1-d_1 \gamma}\zeta(d_2\gamma)
\left\{1+\mathcal{O}\left(n^{d_1\gamma-1}+n^{-(d_2\gamma-1)}\right)\right\}
	& {\!\!\mathrm{if}}\ d_1\gamma<1,\; d_2\gamma>1
,\\
(\log n)^2\left\{1+\mathcal{O}\left((\log n)^{-1}\right) \right\}
	& {\!\!\mathrm{if}}\ d_1\gamma=1,\; d_2\gamma=1
,\\
(\log n) \zeta(d_2\gamma)\left\{1+\mathcal{O}\left( (\log n)^{-1}\right)\right\}
	& {\!\!\mathrm{if}}\ d_1\gamma=1,\; d_2\gamma>1
,\\
\zeta(d_1\gamma)\zeta(d_2 \gamma)-\zeta((d_1+d_2) \gamma) +\mathcal{O}\left(n^{-(d_1 \gamma-1)}\right)
	& {\!\!\mathrm{if}}\ d_1\gamma>1,\; d_2\gamma>1.
\end{dcases}
\end{equation*}
By inspection, this expression agrees with the hypothesized form of~\eqref{induct}, establishing that $P(1)$ implies $P(2)$. Our general inductive step will follow the same logic as this argument.

We now assume that~\eqref{induct} holds for $P(v)$, and will show that this implies in turn that $P(v+1)$ is also true. We begin by rewriting $S_\gamma(d_1,\dots,d_{v+1})$ in terms of $S_\gamma(d_1,\dots,d_v)$ as follows:
\begin{align*}
S_\gamma(d_1,\dots,d_{v+1})&=\sum_{r_1\in[n]} \dots \sum_{r_{v+1}\in [n]\backslash \{r_1,\dots, r_{v}\}} \prod_{i=1 }^{v+1}  r_i^{-d_t \gamma}\\
&=\sum_{r_1\in[n]} \dots \sum_{r_{v}\in [n]\backslash \{r_1,\dots, r_{v-1}\}} \prod_{i=1 }^{v}  r_i^{-d_t \gamma}\\
&\pushright{\qquad\quad \cdot
\left(\sum_{r_{v+1}\in[n]}r_{v+1}^{-d_{v+1} \gamma}-\left\{r_{1}^{-d_{v+1} \gamma} +\dots
+r_{v}^{-d_{v+1} \gamma}\right\}\right)}\\
&= \left(\sum_{r_{v+1}\in[n]}r_{v+1}^{-d_{v+1} \gamma}\right)S_\gamma(d_1,\dots,d_v)\\
&\pushright{-\sum_{j=1}^vS_\gamma(d_1,\dots,d_{j-1},d_j+d_{v+1},d_{j+1},\dots,d_v).}
\end{align*}
We may now apply~\eqref{eq:pows} to each term in $\textstyle{\sum}_{r_{v+1}\in[n]}r_{v+1}^{-d_{v+1} \gamma}$, leading to the expression
\begin{align}
\nonumber
S_\gamma(d_1,\dots,d_{v+1})&=
\Bigg[
\Big\{(1- d_{v+1} \gamma)^{-1} n^{1-d_{v+1} \gamma}+{\mathcal{O}}(1)\Big\} \cdot 1_{\{d_{v+1} \gamma<1\}}\\
\nonumber
&\quad\ \ +
\Big\{\log n +\gamma_E+{\mathcal{O}}\left(n^{-1}\right)\Big\} \cdot 1_{\{d_{v+1} \gamma=1\}}\\
\nonumber
&\quad\ \ +
\Big\{\zeta(d_{v+1} \gamma)+{\mathcal{O}}\left(n^{-(d_{v+1} \gamma-1)}\right)\Big\} \cdot 1_{\{d_{v+1} \gamma>1\}}
\Bigg]
S_\gamma(d_1,\dots,d_v)\\
\label{secondterm-zeta}
&-\sum_{j=1}^vS_\gamma(d_1,\dots,d_{j-1},d_j+d_{v+1},d_{j+1},\dots,d_v).
\end{align}

To evaluate the right-hand side of~\eqref{secondterm-zeta}, we replace $S_\gamma(d_1,\dots,d_v)$ with the form given by our inductive assumption that~\eqref{induct} holds for $P(v)$. We simplify the result for three exhaustive cases as follows, thereby completing the proof:

{\bf Case 1)}
We first assume that $d_{v+1} \gamma<1$. Then by hypothesis, $d_t \gamma < 1$ for all $t\leq v+1$, and $h=q=v$. Combining~\eqref{induct} with~\eqref{secondterm-zeta}, it follows that the first term in~\eqref{secondterm-zeta} is of order $n^{v+1-\sum_{t=1}^{v+1}d_t\gamma}$, while the last term is at most of order $n^{v-\sum_{t=1}^{v+1}d_t\gamma}$:
\begin{align*}
S_\gamma(d_1&,\dots,d_{v+1})\\
&=
\Big\{(1- d_{v+1} \gamma)^{-1} n^{1-d_{v+1} \gamma}+{\mathcal{O}}(1)\Big\}
S_\gamma(d_1,\dots,d_v)\!+\!\mathcal{O}\left(n^{v-\sum_{t=1}^{v+1}d_t\gamma}\right)\\
&=(1- d_{v+1} \gamma)^{-1} n^{1-d_{v+1} \gamma} S_\gamma(d_1,\dots,d_{v}) \left\{1+{\mathcal{O}}\left(n^{-(1-d_{v+1} \gamma)}\right) \right\}.
\end{align*}
Substituting~\eqref{induct} for $S_\gamma(d_1,\dots,d_{v})$ in this expression and comparing to the claimed result of the proposition, we see that
\begin{equation*}
\epsilon_\gamma(d_1,\dots,d_{v+1})
	= {\mathcal{O}}\left(n^{-(1-d_{v+1} \gamma)}\right),
\end{equation*}
and so our claimed result is directly verified for the case in which $h=q=v$.

{\bf Case 2)} We next assume that $d_{v+1} \gamma=1$, whence from~\eqref{secondterm-zeta} we conclude that
\begin{multline*}
S_\gamma(d_1,\dots,d_{v+1})=
\Big\{\log n +\gamma_E+{\mathcal{O}}\left(n^{-1}\right)\Big\}
S_\gamma(d_1,\dots,d_v)\\
-\sum_{j=1}^vS_\gamma(d_1,\dots,d_{j-1},d_j+d_{v+1},d_{j+1},\dots,d_v).
\end{multline*}
To determine the dominant error term in this expression, we note that since $d_{v+1} \gamma=1$, it follows from our inductive assumption that~\eqref{induct} holds for $P(v)$ that
\begin{multline*}
S_\gamma(d_1,\dots,d_{j-1},d_j+d_{v+1},d_{j+1},\dots,d_v)
=
S_\gamma(d_1,\dots,d_{j-1},d_j,d_{j+1},\dots,d_v)
\\ \cdot \left\{ \Theta\left(1_{\{d_{j} \gamma=1\}}(\log n)^{-1}+ 1_{\{d_{j} \gamma<1\}}n^{d_j\gamma-1}\right)\right\}.
\end{multline*}
Therefore, the Euler--Mascheroni constant $\gamma_E$ dominates the error, and
\begin{equation*}
\epsilon_\gamma(d_1,\dots,d_{v+1})
	= {\mathcal{O}}\left((\log n)^{-1}\right).
\end{equation*}
Comparing to the claimed result of the proposition, we then see that it is directly verified for the case in which $d_{v+1} \gamma=1$.

{\bf Case 3)} We finally assume that $d_{v+1} \gamma>1$, and from~\eqref{secondterm-zeta} we see that
\begin{multline}
\label{eq:case3}
S_\gamma(d_1,\dots,d_{v+1})=
\Big\{\zeta(d_{v+1} \gamma)+{\mathcal{O}}\left(n^{-(d_{v+1} \gamma-1)}\right)\Big\}
S_\gamma(d_1,\dots,d_v)\\
-\sum_{j=1}^vS_\gamma(d_1,\dots,d_{j-1},d_j+d_{v+1},d_{j+1},\dots,d_v).
\end{multline}
Recalling our inductive assumption that~\eqref{induct} holds for $P(v)$, we note that:
\begin{multline*}
S_\gamma\left(d_1,\dots,d_{j-1},d_j+d_{v+1},d_{j+1},\dots,d_v\right)\\
=S_\gamma\left(d_1,\dots,d_v\right)\cdot\begin{cases}
\Theta\left(n^{-\left(1-d_j \gamma\right)}\right)
&\text{if $q>0$ and $j\in[q]$,}\\
\Theta\left( (\log n)^{-1} \right)
&\text{if $h>q$ and $j\in[h]\setminus[q]$,}\\
\Theta(1)
&\text{if $v>h$ and $j\in[v]\setminus[h]$,}
\end{cases}
\end{multline*}
and therefore we may simplify~\eqref{eq:case3} to
\begin{multline*}
S_\gamma(d_1,\dots,d_{v+1})=
S_\gamma(d_1,\dots,d_v)\cdot
\Big\{\zeta(d_{v+1}\gamma)+
\\{\mathcal{O}}\left(n^{-(d_{v+1} \gamma-1)}+1_{\{q>0\}}n^{-\left(1-d_q \gamma\right)}+1_{\{q<h<v\}}(\log n)^{-1}\right)\Big\}\\
-\sum_{j=h+1}^vS_\gamma(d_1,\dots,d_{j-1},d_j+d_{v+1},d_{j+1},\dots,d_v).
\end{multline*}
Finally, by substituting~\eqref{induct} for $S_\gamma(d_1,\dots,d_{v})$ in this last expression, we recover the claimed result.
\end{proof}

It follows from Proposition~\ref{ihrg-ptop} that in the setting of Definition~\ref{powerlaw}, the degree sequence of any fixed graph can affect its expected number of unlabeled copies in $G_n$. In particular, as we will now show, for all $\gamma \in (0,1]$, the number of copies of the $(e+1)$-star $K_{1,e}$ grows at least as quickly as that of any other $(e+1)$-tree. Noting that the $2$-star or $2$-path $K_{1,1} \equiv P_2$ is the only element of $\mathcal{T}_2$, and the $3$-star or $3$-path $K_{1,2} \equiv P_3$ is the only element of $\mathcal{T}_3$, we formalize this result as follows.

\begin{Corollary}\label{cor:irg_star_dom}
Fix a number of edges $e \geq 3$ and a tree $T \not\equiv S \in \mathcal{T}_{e+1}$, where $S = K_{1,e}$ is the $(e+1)$-star. Then, for all $\gamma \in (0,1]$ in the setting of Proposition~\ref{ihrg-ptop}:
\begin{equation*}
\frac{\E X_T(G_n)}{\E X_S(G_n)} = 
\begin{cases}
\Theta\left(1\right) & 0 < \gamma < 1/e, \\
\Theta\left( (\log n)^{-1} \right) & \gamma = 1/e, \\
\mathcal{O}\left( (\log n)^{d^\ast} n^{-\min\{\gamma, e\gamma-1\}}\right) & 1/e < \gamma < 1/2, \\
\mathcal{O}\left( (\log n)^{d^\ast} n^{-1} \right) & 1/2 \leq \gamma < 1, \\
\mathcal{O}\left( (\log n)^{-1} \right) & \gamma = 1;
\end{cases}
\end{equation*}
where $d^\ast$ is the number of degrees of $T$ equal to $\gamma^{-1}$.
\end{Corollary}

\begin{proof}
Let $(d_1,\dots, d_{e+1})$ be the degrees of $T\not\equiv S$, enumerated in non-decreasing order. If $\gamma=1$, then all degrees of $T$ and $S$ are at least as large as $\smash{\gamma^{-1}} = 1$, and so from Proposition~\ref{ihrg-ptop}, we deduce that only logarithmic growth terms are present:
\begin{equation*}
\E X_T(G_n) = \Theta\left(\theta_n^{2e}(\log n)^{\left|\{t \,:\, d_t=1\}\right|}\right)
\text{ vs.\ }
\E X_S(G_n) = \Theta\left(\theta_n^{2e}(\log n)^{e + 1_{\{e=1\}} }\right).
\end{equation*}
Then, since for $e \geq 3$, the $(e+1)$-star $S=K_{1,e}$ is distinct in having the maximum number $e$ of leaves among all trees of order $e+1$, we recover the claimed result.

By this same reasoning, if $1/2 < \gamma < 1$, then $\E X_T(G_n) / \E X_S(G_n) = \smash{ \mathcal{O}\big( n^{-1} \big)}$:
\begin{equation*}
\E X_T(G_n) = \Theta\left(\theta_n^{2e} n^{\left|\{t \,:\, d_t=1\}\right|}\right)
\text{ vs.\ }
\E X_S(G_n) = \Theta\left(\theta_n^{2e} n^{e + 1_{\{e=1\}} }\right),
\end{equation*}
and if $\gamma = 1/2$, then $\E X_T(G_n)$ will gain an extra factor of $\smash{ (\log n)^{\left|\{t \,:\, d_t=2\}\right|} }$.

If $1/e < \gamma < 1/2$, then we consider two cases.  First, assume that $d_{e+1} \gamma < 1$. By Proposition~\ref{ihrg-ptop} we deduce that $\E X_T(G_n) / \E X_S(G_n) = \smash{\Theta\big(n^{-(e\gamma-1)}\big)}$, since:
\begin{equation*}
\E X_T(G_n) = \Theta\left(\theta_n^{2e}n^{e(1-2\gamma)+1}\right)
\text{ vs.\ }
\E X_S(G_n) = \Theta\left(\theta_n^{2e}n^{e(1-2\gamma) + e\gamma}\right).
\end{equation*}
If $d_{e+1} \gamma \geq 1$, then for $q = \left| \{ t: d_t \gamma < 1\} \right| \in \{2, \ldots, e\}$ and $d^\ast = \left| \{ t: d_t \gamma = 1\} \right|$,
\begin{align*}
\E X_T(G_n)
& = \Theta\left(\theta_n^{2e}n^{q-\gamma\sum_{t=1}^{q}d_t}(\log n)^{d^\ast}\right) \\
& = \mathcal{O}\left(\theta_n^{2e}n^{q(1-\gamma)-\gamma\, \cdot 1_{\{d_q\geq2\}}}(\log n)^{d^\ast}\right).
\end{align*}
where the lower bound $\smash{\textstyle \sum_{t=1}^{q}d_t \geq q+1_{\{d_q\geq2\}}}$ follows since $\textstyle\sum_{t=1}^{q}d_t=q$ if $d_q=1$, whereas if $d_q\geq 2$, then $\smash{\textstyle\sum_{t=1}^{q}d_t\geq q+1}$.

Thus we have that $\E X_T(G_n) / \E X_S(G_n) = \smash{\mathcal{O}\big(n^{(q-e)(1-\gamma)-\gamma \,\cdot 1_{\{d_q\geq2\}}}(\log n)^{d^\ast}\big)}$. As a result, if $2 \leq q\leq e-1$, this rate is $\smash{\mathcal{O}( n^{-(1-\gamma)} (\log n)^{d^\ast} )}$. If instead $q=e$, then we must have $d_q\geq 2$ (otherwise $T \equiv S$), yielding the rate $\smash{\mathcal{O}( n^{-\gamma}(\log n)^{d^\ast})}$.

Finally, if $0 < \gamma\leq 1/e$, then all degrees of $T\not\equiv S$ and $S$ are strictly less than $\smash{\gamma^{-1}}$, except possibly the largest degree $e$ of $S$. From Proposition~\ref{ihrg-ptop} we then deduce:
\begin{equation*}
\E X_T(G_n) = \Theta\left(\theta_n^{2e}n^{e(1-2\gamma)+1}\right)
\text{ vs.\ }
\E X_S(G_n) = \Theta\left(\theta_n^{2e}n^{e(1-2\gamma)+1} (\log n)^{1_{\{e\gamma=1\}}} \right).
\end{equation*}
Thus if $\gamma< 1/e$, then $\E X_T(G_n) / \E X_S(G_n) = \Theta(1)$, recovering the bounded kernel setting as $\gamma \rightarrow 0$. If $\gamma=1/e$, then $\E X_T(G_n) / \E X_S(G_n) = \Theta\big((\log n)^{-1}\big)$.
\end{proof}

The following result, analogous to Corollary~\ref{cor:irg_star_dom}, holds for unicyclic graphs.

\begin{Corollary}\label{cor:irg_octo_dom}
Fix a number of edges $e \geq 4$ and a connected graph $U \not\equiv C_3K_{1,e-3}$ containing exactly one cycle, where $C_3K_{1,e-3}$ is the graph obtained by identifying any node of the 3-cycle $C_3$ with the unique node of degree $e-3$ in the $(e-2)$-star $K_{1,e-3}$. Then, for all $\gamma \in (0,1/2]$ in the setting of Proposition~\ref{ihrg-ptop}:
\begin{equation*}
\frac{\E X_U(G_n)}{\E X_{C_3K_{1,e-3}}(G_n)} = 
\begin{cases}
\Theta\left(1\right) & 0 < \gamma < 1/(e-1), \\
\Theta\left( (\log n)^{-1} \right) & \gamma = 1/(e-1), \\
\mathcal{O}\left( (\log n)^{d^\ast} n^{-\min\{\gamma,1-2\gamma, (e-1)\gamma-1\}}\right) & 1/(e-1) < \gamma < 1/2, \\
\mathcal{O}\left( (\log n)^{-1} \right) & \gamma = 1/2;
\end{cases}
\end{equation*}
where $d^\ast$ is the number of degrees of $U$ equal to $\gamma^{-1}$. Furthermore, if $\smash{ \mathcal{U}_3^{e-3} }$ is the set of unlabeled graphs formed by attaching $e-3$ edges to $C_3$ to form $e-3$ pendant nodes, then for $\smash{ U' \not\in \mathcal{U}_3^{e-3} }$ and $\smash{ F, F'\in \mathcal{U}_3^{e-3} }$, $\E X_{F'}(G_n) / \E X_F(G_n) = \Theta(1)$ and
\begin{equation*}
\frac{\E X_{U'}(G_n)}{\E X_F(G_n)} = 
\begin{cases}
\mathcal{O}\left( n^{-1}\right) & 1/2 < \gamma < 1, \\
\mathcal{O}\left( (\log n)^{-1} \right) & \gamma = 1.
\end{cases}
\end{equation*}
\end{Corollary}

\begin{proof}
Let $(d_1,\dots, d_e)$ be the degrees of $U \not\equiv C_3K_{1,e-3}$ in non-decreasing order. If $\gamma=1$, then all degrees of $U$ and $C_3K_{1,e-3}$ are at least as large as $\smash{\gamma^{-1}} = 1$, and so from Proposition~\ref{ihrg-ptop}, we deduce that only logarithmic growth terms are present:
\begin{equation*}
\E X_U(G_n) = \Theta\left(\theta_n^{2e}(\log n)^{\left|\{t \,:\, d_t=1\}\right|}\right)
\text{ vs.\ }
\E X_{C_3K_{1,e-3}}(G_n) = \Theta\left(\theta_n^{2e}(\log n)^{e-3}\right).
\end{equation*}
Attaching $e-3$ pendant nodes to $C_3$ will satisfy $\smash{ \E X_U(G_n) = \Theta\big(\E X_{C_3K_{1,e-3}}(G_n)}$, whereas attaching $e-4$ pendant nodes to $C_4$ will lose a power of $\log n$. Thus $\smash{ \E X_{U'}(G_n) / \E X_F(G_n) = \mathcal{O}\big( (\log n)^{-1} \big)}$ for $\smash{ U' \not\in \mathcal{U}_3^{e-3} }$ and $\smash{ F \in \mathcal{U}_3^{e-3} }$ as claimed.

By this same reasoning, if $1/2 < \gamma < 1$, then $\smash{ \E X_{U'}(G_n) / \E X_F(G_n) = \mathcal{O}\big( n^{-1} \big)}$:
\begin{equation*}
\E X_{U'}(G_n) = \Theta\left(\theta_n^{2e} n^{\left|\{t \,:\, d_t=1\}\right|}\right)
\text{ vs.\ }
\E X_F(G_n) = \Theta\left(\theta_n^{2e}n^{e-3}\right).
\end{equation*}
If $\gamma = 1/2$, then $\E X_{C_3K_{1,e-3}}(G_n)$ will gain an extra $\log n$ factor relative to any other $\smash{ F \in \mathcal{U}_3^{e-3} }$, since $C_3K_{1,e-3}$ is unique in $\smash{\mathcal{U}_3^{e-3} }$ in containing two nodes of degree 2.

If $1/(e-1) < \gamma < 1/2$, then we consider two cases.  First, assume that $d_e \gamma < 1$. By Proposition~\ref{ihrg-ptop} we deduce that $\E X_U(G_n) / \E X_{C_3K_{1,e-3}}(G_n) = \smash{\Theta\big(n^{-[(e-1)\gamma - 1]}\big)}$:
\begin{equation}
\label{eq:C3S-bnd}
\E X_U(G_n) = \Theta\left(\theta_n^{2e}n^{e(1-2\gamma)}\right)
\text{ vs.\ }
\E X_{C_3K_{1,e-3}}(G_n) = \Theta\left(\theta_n^{2e}n^{e(1-2\gamma) + (e-1)\gamma - 1}\right).
\end{equation}

If $d_e \gamma \geq 1$, then for $q = \left| \{ t: d_t \gamma < 1\} \right| \in \{3, \ldots, e-1\}$ and $d^\ast = \left| \{ t: d_t \gamma = 1\} \right|$,
\begin{align}
\nonumber
\E X_U(G_n) & = \Theta\left(\theta_n^{2e}n^{
q-\gamma\sum_{t=1}^{q}d_t}(\log n)^{d^\ast}\right)
\\ & = \mathcal{O}\left(\theta_n^{2e}n^{q(1-\gamma) - \gamma (d_q-1) \cdot 1_{\{d_q\geq2\}} - \gamma (d_{q-1}-1) \cdot 1_{\{d_{q-1}\geq2\}} } (\log n)^{d^\ast}\right),
\label{eq:U-bnd}
\end{align}
where the lower bound $\smash{ \textstyle \sum_{t=1}^{q}d_t \geq q + (d_q-1) \cdot 1_{\{d_q\geq2\}} + (d_{q-1}-1) \cdot 1_{\{d_{q-1}\geq2\}} }$ follows from that fact that $d_q \geq 1$. If $d_q=1$, then comparing~\eqref{eq:U-bnd} and~\eqref{eq:C3S-bnd} yields
\begin{equation*}
\frac{\E X_U(G_n)}{\E X_{C_3K_{1,e-3}}(G_n)} =
\mathcal{O}\left(n^{
(q-e)(1-\gamma) + \gamma + 1} (\log n)^{d^\ast}
\right).
\end{equation*}
By hypothesis, $U$ contains exactly one cycle, each vertex of which is of degree at least $2$. Therefore $U$ must possess at least three vertices of degree larger than $1$, and so if $d_q=1$, then we must have $q \leq e - 3$. Consequently,    
\begin{equation*}
\frac{\E X_U(G_n)}{\E X_{C_3K_{1,e-3}}(G_n)} =
\mathcal{O}\left(n^{
3(\gamma-1) + \gamma + 1} (\log n)^{d^\ast}
\right)
=
\mathcal{O}\left(n^{
4(\gamma-1/2)} (\log n)^{d^\ast}
\right).
\end{equation*}
By this same reasoning, if $q = e - 2$, then $d_q\geq 2$ and so $\textstyle\sum_{t=1}^{q}d_t\geq q+1$. Thus,
\begin{equation*}
\frac{\E X_U(G_n)}{\E X_{C_3K_{1,e-3}}(G_n)} 
=
\mathcal{O}\left(n^{
2(\gamma-1)-\gamma+\gamma+1} (\log n)^{d^\ast}
\right)
=
\mathcal{O}\left(n^{
2(\gamma-1/2)} (\log n)^{d^\ast}
\right).
\end{equation*}
Likewise, if $q=e-1$, then $d_q\geq 2$ and $d_{q-1} \geq 2$. If $d_q\geq 3$ then $\textstyle\sum_{t=1}^{q}d_t\geq q+3$. If instead $d_q = 2$, then necessarily $d_{q-1} = 2$. In this case, since $C_3K_{1,e-3}$ is the only unicyclic graph with $d_{e-2} = d_{e-1} = 2$ and $d_e = e-1$, it follows that $d_e \leq e - 2$, and so again we conclude that $\textstyle\sum_{t=1}^{q}d_t = 2e - d_e \geq q+3$. Thus we obtain
\begin{equation*}
\frac{\E X_U(G_n)}{\E X_{C_3K_{1,e-3}}(G_n)} 
=
\mathcal{O}\left(n^{
(\gamma-1)-3\gamma+\gamma+1} (\log n)^{d^\ast}
\right)
=
\mathcal{O}\left(n^{-\gamma} (\log n)^{d^\ast}\right).
\end{equation*}

Finally, if $0 < \gamma\leq 1/(e-1)$, then all degrees of $U \not\equiv C_3K_{1,e-3}$ and $C_3K_{1,e-3}$ are strictly less than $\smash{\gamma^{-1}}$, except possibly the largest degree $e-1$ of $C_3K_{1,e-3}$. Thus:
\begin{equation*}
\E X_U(G_n)\!= \Theta\left(\!\theta_n^{2e}n^{e(1-2\gamma)}\right)
\text{\! vs.\!\ }
\E X_{C_3K_{1,e-3}}\!(G_n)\!=\!\Theta\left(\!\theta_n^{2e}n^{e(1-2\gamma)} (\log n)^{1_{\{(e-1)\gamma=1\}}}\!\right).
\end{equation*}
If $\gamma< 1/(e-1)$, $\E X_U(G_n) / \E X_{C_3K_{1,e-3}}(G_n) = \Theta(1)$, recovering the bounded kernel setting as $\gamma \rightarrow 0$. If $\gamma= 1/(e-1)$, $\smash{ \E X_U(G_n) / \E X_{C_3K_{1,e-3}}(G_n) = \Theta\big((\log n)^{-1}\big) }$.
\end{proof}

We next establish the form of the logarithm of the expected graph walk density in the setting of Proposition~\ref{ihrg-ptop}.
\begin{Corollary}\label{cor:irg_walk_dens}
Fix an integer $k\geq 2$ and a closed $k$-walk $w\in\mathcal{W}_k(K_k)$ whose induced subgraph has degrees $(d_1,\dots, d_v)$. Then, in the setting of Proposition~\ref{ihrg-ptop}, the graph walk density $\varphi(w,G_n)$ satisfies
\begin{equation*}
\log \E\varphi(w,G_n)= - \left[ 
\left(\sum_{i=1}^v\min\{1, d_i\gamma \}\right)\log n
+ \left(\sum_{i=1}^v d_i\right)\left|\log\theta_n\right|
\right]\left\{1+o(1)\right\}.
\end{equation*}
\end{Corollary}

\begin{proof}
We first recall Definition~\ref{varphi}, and combine it with Lemma~\ref{nb-copies} to obtain
\begin{align}
\nonumber
\E\varphi(w,G_n)
=\frac{\E\mathrm{ind}_k(F_w,G_n)}{\mathrm{ind}_k(F_w,K_n)}
&=\frac{\mathrm{ind}_k(F_w,F_w) \E X_{F_w}(G_n)}{\mathrm{ind}_k(F_w,F_w) X_{F_w}(K_n)}\\
\nonumber
&= \frac{\E X_{F_w}(G_n)}{(n)_v/{\mathrm{aut}}(F_w)}\\
\label{eq:evarphi_irg}
&= {\mathrm{aut}}(F_w)\,n^{-v}\E X_{F_w}(G_n)\left\{1+\mathcal{O}\left(n^{-1}\right)\right\}.
\end{align}
Then, replacing $\E X_{F_w}(G_n)$ in~\eqref{eq:evarphi_irg} with the result of Proposition~\ref{ihrg-ptop}, we obtain
\begin{multline}
\label{eq:evarphi_irg_ii}
\E\varphi(w,G_n) = \frac
	{\left[C_\gamma(d_{h+1},\dots,d_v)\right]^{1_{\{h<v\}}}}
	{\left[\prod_{t=1}^q(1-d_t\gamma)\right]^{1_{\{q>0\}}}}
	\theta_n^{\sum_{i=1}^v d_i}
	n^{-v +q-\sum_{t=1}^qd_t \gamma} (\log n)^{h-q}
	\\ \cdot \left\{1+\mathcal{O}\left(n^{-1}\right)+\epsilon_\gamma(d)\right\}.
\end{multline}
The logarithm of~\eqref{eq:evarphi_irg_ii} will yield the claimed result, after noting $\theta_n \in (0,1]$ and
\begin{equation*}
-v +q-\sum_{t=1}^qd_t \gamma
	=-\sum_{t=q+1}^v 1-\sum_{t=1}^qd_t \gamma
	=-\sum_{t=1}^v\min(1, d_t\gamma ).
\end{equation*}
These substitutions allow us to write $\log\E\varphi(w,G_n)$, following on from~\eqref{eq:evarphi_irg_ii}, as 
\begin{multline}
\label{eq:full_varphi}
\log\E\varphi(w,G_n)
=\log
\frac{\left\{C_\gamma(d_{h+1},\dots,d_v)\right\}^{1_{\{h<v\}}}}{\left\{\prod_{t=1}^q(1-d_t\gamma)\right\}^{1_{\{q>0\}}}}
-\left(\sum_{i=1}^v\min\{1, d_i\gamma\}\right)\log n\\
	-\left(\sum_{i=1}^v d_i\right)\left|\log\theta_n\right|
	+(h-q)\log \log n
	+\log\left(1+\mathcal{O}\left(n^{-1}\right)+\epsilon_\gamma(d)\right).
\end{multline}
Noting that $\textstyle \sum_{i=1}^v\min\{1, d_i\gamma\} \geq 2\gamma > 0$ for $k \geq 2$, we conclude by absorbing the first, fourth, and fifth terms of~\eqref{eq:full_varphi} into a multiplicative $\left\{1+o(1)\right\}$ error term.
\end{proof}

\subsection{Dominating closed walks in scale-free random graphs}

Above we have compared the expected numbers of unlabeled copies of trees (Corollary~\ref{cor:irg_star_dom}) and unicylic graphs (Corollary~\ref{cor:irg_octo_dom}) in the setting of Proposition~\ref{ihrg-ptop}. We now characterize dominant walks in this setting. To do so we will study the model of Definition~\ref{powerlaw} for a sequence $\{G_n\}$ of random graphs, in conjunction with a scaling sequence $\{\theta_n\}$ taking values in $(0,1]$. We therefore define a measure of the speed of decay of $\theta_n$ in $n$ as follows, recalling from Definition~\ref{powerlaw} that $\theta_n \in (0,1]$:
\begin{equation}
\label{betadef}
\beta_n=\frac{\left| \log \theta_n \right|}{\log n }.
\end{equation}
Thus $\theta_n = n^{-\beta_n}$, and so $\beta_n$ measures the rate of decay of $\log \theta_n$ relative to $\log n$.

\begin{Theorem}\label{propn:irg_thm_validity}
Fix a walk length $k\geq 2$, and consider the setting of Proposition~\ref{ihrg-ptop} with the additional assumption that $\textstyle \lim_{n\rightarrow\infty}\beta_n = \beta \geq 0$ exists. If $\beta+\gamma<1/2$, so that a randomly chosen degree in this setting will grow in expectation, then the set of asymptotically dominating walk-induced subgraphs ($W_k^\ast(\{G_n\})$) is as follows.

If $\beta+\gamma<(1-\gamma)/2$, then Theorem~\ref{thm:dominant-walk} holds, and with $\,\smash{k^\ast = [1/2-(\beta_n+\gamma)]^{-1}}$,%
\begin{equation}\label{ihrg-prop-W}
\E|\mathcal{W}_k(G_n)| \sim
\begin{cases}
\E\mathrm{ind}_k(C_k,G_n) & \text{if $k$ is odd or if $k>k^\ast$,}\\
\E\mathrm{ind}_k(C_k,G_n) & \\ \quad + \textstyle \sum_{T\in\mathcal{T}_{k/2+1}} \E\mathrm{ind}_k(T,G_n)
	& \text{if $k$ is even and $k=k^\ast$,}\\
\textstyle \sum_{T\in\mathcal{T}_{k/2+1}} \E\mathrm{ind}_k(T,G_n)
	& \text{if $k$ is even and $k<k^\ast$.}
\end{cases}\end{equation}

If $\beta + \gamma=(1-\gamma)/2$, then $k^\ast=2/\gamma$, and
\begin{equation}\label{ihrg-prop-W-2}
W_k^\ast(\{G_n\}) = \begin{cases}
\{C_k\}&\text{if $k$ is odd,}\\
\mathcal{T}_{k/2+1}&\text{if $k$ is even and }k<k^\ast,\\
\{C_k,K_{1,k/2}\}&\text{if $k$ is even and }k\geq k^\ast.
\end{cases}\end{equation} 

If $(1-\gamma)/2 < \beta + \gamma < 1/2$, then with $k^\dagger = \tfrac2\gamma$, $k^\circ=2\tfrac{1/2 - \gamma}{\beta + \gamma - (1-\gamma)/2}+3$, and $k^{+}=\max\{k^\circ,k^\dagger+1\}$,
\begin{equation}\label{ihrg-prop-W-1}
W_k^\ast(\{G_n\})   =
\begin{cases}
\{C_k\} & \text{if $k$ is odd and $k<k^{+}$},\\
\{C_k\} & \text{if $k$ is odd, $k=k^{+}$ and $k^{+}\leq k^{\dagger}+1$},\\
\{C_k,C_3K_{1,(k-3)/2}\} & \text{if $k$ is odd, $k=k^{+}$ and $k^{+}>k^\dagger+1$}, \\
\{C_3K_{1,(k-3)/2}\} & \text{if $k$ is odd and $k>k^{+}$},\\
\mathcal{T}_{k/2+1} & \text{if $k$ is even and $k<k^\dagger$},\\
\{K_{1,k/2}\} & \text{if $k$ is even and $k\geq k^\dagger$}.
\end{cases}
\end{equation} 

Furthermore, for all $k\geq 3$, non-backtracking, tailless closed $k$-walks are dominated by the $k$-cycle: 
\begin{equation}
\label{ihrg-prop-Wb}
\E|\mathcal{W}_k^b(G_n)| \sim \E\mathrm{ind}_k(C_k,G_n),
\end{equation}
if and only if $\beta+\gamma<1/2$ or $\gamma= 1/2$ and $\beta_n \log n / \log \log n < 1/2$ eventually in $n$.
\end{Theorem}

\begin{proof}
We will prove each of~\eqref{ihrg-prop-W}--\eqref{ihrg-prop-Wb} in order. To begin, note that for $F=K_2$, Proposition~\ref{ihrg-ptop} immediately implies that for $\gamma \in (0,1)$, the probability of an edge satisfies $ \textstyle \smash{ \E X_{K_2}(G_n) / \binom{n}{2} = \Theta\big( n^{-2( \beta_n+\gamma )} \big) }$, and thus the expected value $ \textstyle \smash{ \E X_{K_2}(G_n) / n }$ of a randomly chosen degree diverges if and only if $\beta_n+\gamma<1/2$. 

\vspace{.5\baselineskip}
\emph{Proof of~\eqref{ihrg-prop-W}:} 
We begin by showing that Assumptions~\ref{asump1.1} and~\ref{asump1.3} are satisfied, and thus that Theorem~\ref{thm:dominant-walk} holds, whenever $\beta_n+\gamma<(1-\gamma)/2$. Proposition~\ref{ihrg-ptop}, specifically~\eqref{EXFG}, implies that $\E X_F(G_n)$ is strictly positive for any $F$ with strictly positive degrees whenever $|v(F)| \leq n$. Thus  $\E\left|\mathcal{W}_k(G_n)\right|>0$ eventually in $n$ for all $k\geq 2$, and so Assumption~\ref{asump1.1} is satisfied.

If we can show that if $\beta_n+\gamma<(1-\gamma)/2$, then Assumption~\ref{asump1.3} is satisfied, then we can apply Theorem~\ref{thm:dominant-walk}. To this end, we fix $w$ to be an arbitrary closed $k$-walk that admits an extension, and is such that for $n$ sufficiently large, $\E\varphi(w,G_n)>0$. To satisfy Assumption~\ref{asump1.3}, it is sufficient to exhibit an extension $w'$ of $w$ such that
\begin{equation}
\label{exten-eq-pl}
\log \left( \frac{n \E\varphi(w',G_n)}{\E\varphi(w,G_n)} \right) = \omega(1).
\end{equation}

Let $w'$ be any extension of $w$, and denote the degree sequences of $F_w$ and $F_{w'}$ by $d=(d_1,\dots,d_v)$ and $d' = (d'_1,\dots, d'_{v+1})$, respectively. We write
\begin{equation*}
\lambda=\sum_{t=1}^v \min(1,d_t \gamma)
\quad\text{and}\quad
\lambda'=\sum_{t=1}^{v+1} \min(1,d_t' \gamma).
\end{equation*}
Then, since $F_{w'}$ has one more vertex and up to one more edge than $F_w$, it follows from Corollary~\ref{cor:irg_walk_dens} that
\begin{align*}
\log\left(\frac{n\E\varphi(w',G_n)}{\E\varphi(w,G_n)}\right)
&=\left[1-2\beta_n \left(|e(F_{w'})|-|e(F_w)|\right) -\left(\lambda'-\lambda\right) \right] \log n \left\{1+o(1)\right\} \\
&\geq\left(1-2\beta_n+\lambda-\lambda'\right) \log n \left\{1+o(1)\right\}.
\end{align*}
Thus,~\eqref{exten-eq-pl} will hold if $w$ and $w'$ are such that, eventually in $n$,
\begin{equation}
\label{pow-law-lambda-cond}
\lambda+1-2\beta_n-\lambda'>0.
\end{equation}

We now appeal to Lemma~\ref{extens-prop} to exhibit a $w'$ such that~\eqref{pow-law-lambda-cond}---and consequently~\eqref{exten-eq-pl}---is satisfied. Lemma~\ref{extens-prop} characterizes the degree sequence of $F_{w'}$ as a function of that of $F_w$ through four exhaustive cases. First, by inspection of each case in Lemma~\ref{extens-prop}, we see that $d_t=d'_t$ for $t\leq v-2$, as we have added one node and at most one edge, and therefore
\begin{multline*}
\lambda-\lambda'
=[\min(1,d_{v-1}\gamma)-\min(1,d_{v-1}'\gamma)]\\
+[\min(1,d_v\gamma)-\min(1,d_v'\gamma)]
-\min(1,d_{v+1}'\gamma).
\end{multline*}
Again by inspection of each case, we see that $d_{v-1}'\leq d_{v-1}$ and $d_v'\leq d_v+1$, so that
\begin{align*}
\min(1,d_{v-1}\gamma)-\min(1,d_{v-1}'\gamma) & \geq 0,
\\ \min(1,d_v\gamma)-\min(1,d_v'\gamma) & \geq\min(1,d_v\gamma)-\min(1,(d_v+1)\gamma).
\end{align*}
Finally, since $\min(1,d_v \gamma) - \min(1,(d_v+1) \gamma) \geq - \gamma$, we conclude that
\begin{equation*}
\lambda - \lambda'\geq
\begin{cases}
-\min(1,2\gamma) & \text{Case 1 of Lemma~\ref{extens-prop},}\\
-\gamma - \min(1,2\gamma) & \text{Case 2 of Lemma~\ref{extens-prop},}\\
-\gamma - \min(1,\gamma) & \text{Case 3 of Lemma~\ref{extens-prop},}\\
-\min(1,\gamma) & \text{Case 4 of Lemma~\ref{extens-prop}.}
\end{cases}
\end{equation*}
The bound $\lambda - \lambda' \geq -3\gamma$ therefore holds in all four cases. Then, since $\textstyle \beta = \lim_{n\rightarrow\infty}\beta_n $ exists by hypothesis, we see that the condition $\beta+\gamma<(1-\gamma)/2$ ensures that~\eqref{pow-law-lambda-cond} holds eventually in $n$---and thus that Assumption~\ref{asump1.3} holds.

We now establish~\eqref{ihrg-prop-W} under the hypothesis $\beta_n+\gamma<(1-\gamma)/2$. Since Assumptions~\ref{asump1.1} and~\ref{asump1.3} hold under this condition, Theorem~\ref{thm:dominant-walk} is in force, and we may use it to deduce that walks mapping out trees and cycles at maximal scales are dominant. In particular, combining Theorem~\ref{thm:dominant-walk} with Item~\ref{item:walksE} of Proposition~\ref{w-k-ast-prop} yields directly that
\begin{equation*}
\E|\mathcal{W}_k(G_n)|
\sim
\E\mathrm{ind}_k(C_k,G_n)+
1_{\{k\text{ even}\}}\sum_{T\in\mathcal{T}_{k/2+1}} \E\mathrm{ind}_k(T,G_n).
\end{equation*}
This result, along with Proposition~\ref{w-k-ast-prop}, implies~\eqref{ihrg-prop-W} when $k=2$ or when $k\geq 3$ is odd. When $k\geq 4$ is even, we must instead compare the behavior of trees and cycles. Corollary~\ref{cor:irg_star_dom} establishes that the expected number of closed walks inducing the $(k/2+1)$-star $K_{1,k/2}$ grows at least as quickly as that of any other tree in $\mathcal{T}_{k/2+1}$, and thus it is sufficient to compare the expected counts of walks inducing the $k$-cycle $C_k$ and the $(k/2+1)$-star $K_{1,k/2}$. Corollary~\ref{cor:irg_star_dom} then asserts that if $k<2/\gamma$, all trees in $\mathcal{T}_{k/2+1}$ will exhibit the same rate of growth.

Since $\beta_n+\gamma<(1-\gamma)/2$ implies $\gamma<1/2$, Proposition~\ref{ihrg-ptop} implies that
\begin{equation}
\label{eq:star-cycle}
\frac{\E X_{K_{1,k/2}}(G_n)}{\E X_{C_k}(G_n)}
=\Theta\left(
n^{\max\{1,k\gamma/2\}-\frac k2(1-2\beta_n-2\gamma)}
(\log n)^{1_{\{k\gamma=2\}}}
\right).
\end{equation}
The map $k\mapsto\max\{1,k\gamma/2\}-k(1-2\beta_n-2\gamma)/2$ is monotone decreasing in $k$ if $\beta_n+\gamma<(1-\gamma)/2$, and also admits a unique zero at $\smash{k^\ast = [1/2-(\beta_n+\gamma)]^{-1}}$. Therefore, when $k\geq 4$ is even, we conclude the following:
\begin{equation*}
\frac{\E X_{C_k}(G_n)}{\E X_{K_{1,k/2}}(G_n)}=
\begin{cases}
	o(1) & \text{if }k<\smash{k^\ast},\\
	\Theta(1) & \text{if }k=\smash{k^\ast},\\
	\omega(1) & \text{if }k>\smash{k^\ast}.
\end{cases}
\end{equation*}
Finally, the hypothesis $\beta_n+\gamma<(1-\gamma)/2$ implies that $k^\ast<2/\gamma$, and so by Corollary~\ref{cor:irg_star_dom}, we conclude that all trees in $\mathcal{T}_{k/2+1}$ exhibit the same rate of growth whenever $k\leq k^\ast$. Together with Proposition~\ref{w-k-ast-prop}, these results imply~\eqref{ihrg-prop-W} when $k\geq 4$.

\vspace{.5\baselineskip}
\emph{Proof of~\eqref{ihrg-prop-W-2} and~\eqref{ihrg-prop-W-1}:} 
We now treat the case $(1-\gamma)/2 \leq \beta_n+\gamma < 1/2$. Fix $F\in W_k$, let $v=|v(F)|$ and $e=|e(F)|$, write $(d_1,\dots, d_v)$ for the degree sequence of $F$, enumerated in non-decreasing order, and set $v_{(s)} = |\{t\,:\,d_t=s\}|$. 

First, assume that $e-v \geq 0$. Fix $F\in W_k\setminus\{C_k, C_3K_{1,(k-3)/2}\}$, where the graph $C_3K_{1,(k-3)/2}$ is obtained by identifying any node of the 3-cycle $C_3$ with the unique node of degree $(k-3)/2$ in the $[(k-3)/2+1]$-star $K_{1,(k-3)/2}$.

To bound $\E X_F(G_n)$, we start from the expression for $\E X_F(G_n)$ in Proposition~\ref{ihrg-ptop}, with its definitions of $q$ and $h$. Then, with $v_{(s)} = |\{t\,:\,d_t=s\}|$, we have 
\begin{align}
\label{eav2}
q-\sum_{t=1}^qd_t \gamma &= \sum_{t=1}^{v} \max\{1-d_t\gamma,0\} = \sum_{s=1}^{d_v}v_{(s)}\max\{1-s\gamma,0\},
\\
\label{eav}
v&=v_{(1)}+v_{(2)}+\dots+v_{(d_v)}.
\end{align}
We may use~\eqref{eav2} and~\eqref{eav}, along with the fact that $\gamma < 1/2$, to obtain
\begin{align*}
\E & \,X_F(G_n)=
\frac
		{\left\{C_\gamma(d_{h+1},\dots,d_v)\right\}^{1_{\{h<v\}}}}
		{\mathrm{aut}(F)\left\{\prod_{t=1}^q(1-d_t\gamma)\right\}^{1_{\{q>0\}}}}
	\theta_n^{2e} n^{q-\sum_{t=1}^qd_t \gamma} (\log n)^{h-q}
	\left\{1\!+\!\epsilon_\gamma(d)\right\}
\\
&= \Theta\left(n^{-2e\beta_n+
 \sum_{s=1}^{d_v}v_{(s)}\max\{1-s\gamma,0\}}
(\log n)^{v_{(1/\gamma)}}
\right)
\\ 
&=\Theta\left(n^{
v_{(1)}(1-\gamma)
	+v_{(2)}(1-2\gamma)
	+ \sum_{s=3}^{d_v}v_{(s)}\max\{1-s\gamma,0\}
	-2\beta_n v-2\beta_n(e-v)}
(\log n)^{v_{(1/\gamma)}}
\right)
\\ &=\Theta\left(n^{
v_{(1)}(1-2\beta_n-\gamma)
	+v_{(2)}(1-2\beta_n-2\gamma)
	+\sum_{s=3}^{d_v}v_{(s)}[\max\{1-s\gamma,0\}-2\beta_n]
	-2\beta_n(e-v)}
(\log n)^{v_{(1/\gamma)}}
\right).
\end{align*}
Then, as we have assumed $(1-\gamma)/2 \leq \beta_n+\gamma $, or equivalently $1 \leq 2\beta_n+3\gamma $, it follows that $\max\{1-2\beta_n-s\gamma,-2\beta_n\}\leq0$ for $s\geq 3$. As a consequence, we obtain
\begin{equation}
\label{ihrg-prop-W-1x-odd}
\E X_F(G_n)
=\mathcal{O}\left(n^{
v_{(1)}(1-2\beta_n-\gamma)
	+v_{(2)}(1-2\beta_n-2\gamma)
	-2\beta_n(e-v)}
(\log n)^{v_{(1/\gamma)}}
\right).
\end{equation}

Now suppose that $v=v_{(1)}+v_{(2)}$. Then, since $e\geq v$, we deduce that $2e=v_{(1)}+2v_{(2)}\geq 2v=2v_{(1)}+2v_{(2)}$. Thus, $v_{(1)}=0$, and it follows that $v_{(2)}=v$.  The only graphs in $W_k$ satisfying these constraints are $C_v$, for $v \in \{3,\ldots,k\}$. Since $F \not\equiv C_k$ by hypothesis, it follows by Item~8 of Lemma~\ref{wk-props} that $v_{(2)}\leq k - 1$, and hence, since $\beta + \gamma < 1/2$, we conclude that $\E X_F(G_n)  / \E X_{C_k}(G_n) = o(1)$, since
\begin{equation*}
  \E X_{C_k}(G_n)=\Theta\left( n^{
	k(1-2\beta_n-2\gamma)}(\log n)^{v_{(1/\gamma)}}\right).
\end{equation*}

If instead $v\geq v_{(1)}+v_{(2)}+1$, then
\begin{equation*}
\E X_F(G_n)=\mathcal{O}\left(
n^{v_{(1)}(1-2\beta_n-\gamma)
	+(v-v_{(1)}-1)(1-2\beta_n-2\gamma)
	-2\beta_n(e-v)}
(\log n)^{v_{(1/\gamma)}}\right).
\end{equation*}
We now bound $v$, using the fact that $F\in W_k\setminus\{C_k, C_3K_{1,(k-3)/2}\}$ is induced by a closed $k$-walk. To this end, observe that at least $2v_{(1)}$ walk steps are used to traverse all $v_{(1)}$ edges attached to pendant nodes, and at least $v-v_{(1)}$ steps are used to visit the remaining vertices. Combining these observations, we conclude that $v+v_{(1)}\leq k$. Thus:
\begin{align}
\nonumber
\!\!\E X_F(G_n)&=\mathcal{O}\left(
n^{v_{(1)}(1-2\beta_n-\gamma)
	+(k-2v_{(1)}-1)(1-2\beta_n-2\gamma)
	-2\beta_n(e-v)}
 (\log n)^{v_{(1/\gamma)}}\right)\\
&=\mathcal{O}\left(
n^{v_{(1)}(-1+2\beta_n+3\gamma)
	+(k-1)(1-2\beta_n-2\gamma)
	-2\beta_n(e-v)}
 (\log n)^{v_{(1/\gamma)}}\right).
\label{boundedgerich}
\end{align}

Next, we bound $v_{(1)}$ in~\eqref{boundedgerich} using the fact that $F$ is induced by a closed $k$-walk. To this end, first observe that since $F$ is not a tree, it must contain a cycle. It follows that at least three walk steps are required to traverse the edges of this cycle.  At the same time, $2v_{(1)}$ steps are required to traverse all $v_{(1)}$ edges attached to pendant nodes. Thus $2v_{(1)}\leq k-3$, with equality only if $k$ is odd.

If $2v_{(1)} \leq k-4$, then since $-1+2\beta_n+3\gamma\geq0$, we obtain from~\eqref{boundedgerich} that
\begin{align}
\nonumber
\E X_F(G_n)
&=\mathcal{O}\left(
n^{\frac{k-4}{2}(-1+2\beta_n+3\gamma)
	+(k-1)(1-2\beta_n-2\gamma)
	-2\beta_n(e-v)}
(\log n)^{v_{(1/\gamma)}}\right)
\\
&=\mathcal{O}\left(n^{
	\frac{k-4}{2}(1-2\beta_n-\gamma)
	+2(1-2\beta_n-2\gamma)
	-2\beta_n(e-v)}(\log n)^{v_{(1/\gamma)}}\right).
\label{eq:irg-3gamma-4odd}
\end{align}
Comparing with the expression for $\E X_{C_3 K_{1,(k-3)/2}}(G_n)$ from Proposition~\ref{ihrg-ptop}:
\begin{equation}
\label{eqn:summary}
\E X_{C_3 K_{1,(k-3)/2}}(G_n)=\Theta\left(n^{
	\frac{k-3}{2}(1-2\beta_n-\gamma)+2(1-2\beta_n-2\gamma)
	+\max\{1-(k+1)\gamma/2,0\}}(\log n)^{v_{(1/\gamma)}}\right),
\end{equation}
we see that $\E X_F(G_n) / \E X_{C_3K_{1,(k-3)/2}}(G_n)=o(1)$ if $k$ is odd.  If $k$ is even, then similarly $\E X_F(G_n) / \E X_{K_{1,k/2}}(G_n)=o(1)$, since
\begin{equation*}
  \E X_{K_{1,k/2}}(G_n) =\Theta\left(n^{
	\frac{k}{2}(1-2\beta_n-\gamma)
	+\max\{1-k\gamma/2,0\}}(\log n)^{v_{(1/\gamma)}}\right).
\end{equation*}

If instead $2v_{(1)}=k-3$, then $F$ must take the form of a triangle with $(k-3)/2$ pendant nodes and degree sequence $(1,\dots,1,2+l_1,2+l_2,2+l_3)$ for natural numbers $l_1,l_2,l_3$ such that $l_1+l_2+l_3=(k-3)/2$. Proposition~\ref{ihrg-ptop} then asserts that~\eqref{eq:irg-3gamma-4odd} holds unless $l_1=l_2=0$, in which case $F\equiv C_3K_{1,(k-3)/2}$ and we obtain a contradiction. 

From the above, we conclude that if $F\in W_k\setminus\{C_k, C_3K_{1,(k-3)/2}\}$ and furthermore $F$ satisfies $e\geq v$, then
\begin{equation*}
\E X_F(G_n)=o\left(\E X_{C_k}(G_n)+\E X_{K_{1,k/2}}(G_n)+\E X_{C_3K_{1,(k-3)/2}} (G_n)\right).
\end{equation*}
If $F$ is such that $e<v$, then Lemma~\ref{wk-props} asserts that $e = v - 1$ and that $k$ must be even. By Corollary~\ref{cor:irg_star_dom}, the expected number of unlabeled copies in $G_n$ of a tree with $e$ edges is bounded by that of the star $K_{1,e}$. If $e < 1/\gamma$, then all trees $T \in \mathcal{T}_{e+1}$ will exhibit similar behavior; otherwise, if $e \geq 1/\gamma$, then $\smash{ \E X_{K_{1,e}}(G_n) }$ will dominate. Furthermore, Proposition~\ref{ihrg-ptop} shows that $\E X_T(G_n)$ increases in $e$ and $n$ whenever $\beta_n+\gamma<1/2$, implying that dominant trees must have $e=k/2$ edges.

We now turn to comparing the rates of growth of $C_k$, $K_{1,k/2}$, and $C_3K_{1,(k-3)/2}$. If  $k\gamma/2<1$, then by Corollary~\ref{cor:irg_star_dom}, all trees in $\mathcal{T}_{k/2+1}$ grow at the same rate as $K_{1,k/2}$. This leads to the same comparison as in the case of Lemma~\ref{dominantscaling}, and so trees dominate cycles if $k$ is even and $\smash{k < k^\ast = [1/2-(\beta_n+\gamma)]^{-1}}$. If $k$ is odd and $(k+1)\gamma/2<1$, then once again replicating the calculations of Lemma~\ref{dominantscaling}, we conclude that cycles dominate.

If $k\gamma/2\geq 1$ and if $k$ is even, then we compare $\E X_{K_{1,k/2}}(G_n)$ to $\E X_{C_k}(G_n)$, as per~\eqref{eq:star-cycle}. Then, if $(1-\gamma)/2 < \beta_n + \gamma$, the map $k\mapsto \max\{1,k\gamma/2\}-k(1-2\gamma-2\beta_n)/2$ is monotone increasing, and for all $k\geq 3$, $\smash{{\E X_{K_{1,k/2}}}/{\E X_{C_k}(G_n)} = \omega(1)}$. If $(1-\gamma)/2 = \beta_n + \gamma$, then $\smash{{\E X_{K_{1,k/2}}}/{\E X_{C_k}(G_n)} = \omega(1)}$ for $\smash{k\leq k^\dagger}$, and otherwise $\smash{{\E X_{K_{1,k/2}}}/{\E X_{C_k}(G_n)} = \Theta(1)}$. This recovers the last two lines of~\eqref{ihrg-prop-W-2}.

Finally, if $(k+1)\gamma/2\geq 1$ and $k$ is odd, then we compare
\begin{align}
\nonumber
\frac{\E X_{C_3K_{1,(k-3)/2}}(G_n)}{\E X_{C_k}(G_n)}
&=\frac{\Theta\left(n^{
	\frac{k-3}{2}(1-2\beta_n-\gamma)+2(1-2\beta_n-2\gamma)
	+\max\{1-(k+1)\gamma/2,0\}}(\log n)^{v_{(1/\gamma)}}\right)}
	{\Theta\left(n^{
	k(1-2\beta_n-2\gamma)}\right)}\\
	&=\Theta\left(n^{
	\frac{k-3}{2}(-1+2\beta_n+3\gamma)-(1-2\beta_n-2\gamma)}(\log n)^{v_{(1/\gamma)}}\right).
\end{align}
When $1<2\beta_n+3\gamma$, then we observe that $C_3K_{1,(k-3)/2}$ dominates whenever $k>3+{2(1-2\beta_n-2\gamma)}/{(-1+2\beta_n+3\gamma)}$. Otherwise, or if $(1-\gamma)/2 = \beta_n + \gamma$, then $C_k$ dominates.

\vspace{.5\baselineskip}
\emph{Proof of~\eqref{ihrg-prop-Wb}:} 
Finally, we show that~\eqref{ihrg-prop-Wb} holds if and only if $\gamma + \beta_n < 1/2$ or $\gamma = 1/2$ and $\beta_n \log n/ \log\log n < 1/2$ eventually in $n$. We recast the result of Proposition~\ref{ihrg-ptop} to refine~\eqref{ihrg-prop-W-1x-odd} as follows, where again we write $\smash{v_{(i)}}$ as the number of nodes of $F$ having degree $i$. We use the identity $\smash{ 2e = \textstyle \sum_{t=1}^v d_t = \sum_{i=1}^{v-1} i v_{(i)} }$: 
\begin{align}
\nonumber
\E X_F(G_n)
& = \Theta\left( n^{ \sum_{t=1}^v \max\{1 - d_t \gamma,0\} }\left(\log n\right)^{\left| t'\,:\,d_{t'} \gamma = 1 \right|} n^{-2e\beta_n} \right)
\\ \nonumber
& = \Theta\left( n^{ \sum_{i=1}^{d_v} v_{(i)}  \max\{1 - i\gamma,0\}  -2e\beta_n} \left(\log n\right)^{ v_{(1/\gamma)} } \right)
\\ \nonumber
& = \Theta\left( n^{ \sum_{i=1}^{d_v} v_{(i)} \left[ \max\{1 - i\gamma,0\}+ x_n \cdot 1_{\{i=1/\gamma\}} \right] -2e\beta_n} \right)
\\
& = \Theta\left( n^{ v \sum_{i=1}^{d_v} \frac{ v_{(i)} }{ v } \left[ \max\{1 - i\gamma,0\}+ x_n \cdot 1_{\{i=1/\gamma\}} - i\beta_n \right] } \right),
\label{eq:EXF-convex-form-i}
\end{align}
where $x_n = \log \log n / \log n = o(1)$. For any fixed $v$, the summation in~\eqref{eq:EXF-convex-form-i} is a convex combination of the values $\big[ \max\{1 - i\gamma,0\} + x_n 1_{\{i=1/\gamma\}} - i\beta_n \big]$ for $i \in \{1,\ldots,d_v\}$. These values are strictly decreasing in $i$ (until $i$ exceeds $1/\gamma$ if $\beta_n = 0$). 

Now, suppose that $v_{(1)} = 0$ and that $0 < \gamma \leq 1/2$. Since the values in the summation of~\eqref{eq:EXF-convex-form-i} are then strictly decreasing in $i$ until at least $i=3$, the maximum of this sum over the entire standard $(v-3)$-simplex is attained uniquely by the extreme point $v_{(2)} = v$. This point is achieved by $F \equiv C_v$ for any fixed $v \in \{3,4,\ldots\}$, and so $\smash{ \E X_{C_v}(G_n) = \Theta\big( n^{v\left[ \max\{1 - 2\gamma,0\}+ x_n \cdot 1_{\{\gamma=1/2\}} - 2\beta_n \right]} \big) }$. Thus $\E X_{C_v}(G_n) \rightarrow \infty$ if and only if $\gamma + \beta_n < 1/2$ eventually in $n$ or $\gamma = 1/2$ and $\beta_n < x_n/2$ eventually in $n$, implying then that $\E X_{F_v}(G_n) = o\big( \E X_{C_v}(G_n) \big)$ for any graph $F_v$ on $v$ vertices with strictly positive degrees and no pendant nodes.

Therefore, if we consider non-backtracking, tailless closed $k$-walks inducing elements of $\smash{W_k^b}$ in this setting, we see that $\smash{ \E X_{C_k}(G_n) = \textstyle \max_{\{v\,:\, C_v \in W_k^b\}} \E X_{C_v}(G_n) }$ grows strictly more quickly than any other element of $\smash{W_k^b}$.  Thus we conclude that $k$-walks inducing $C_k$ will grow in prevalence and uniquely dominate amongst all non-backtracking, tailless closed $k$-walks whenever $  \beta_n +\gamma< 1/2$ or $\gamma = 1/2$ and $\beta_n / x_n < 1/2$ eventually in $n$.
\end{proof}

\section{Methods and algorithms}\label{sec:S6}

\begin{algorithm}[!t]
  \SetAlgoLined
    \KwInput{Simple graph $G$ on $n$ nodes; $N_s$ subgraph sizes $ 1 < s_1, \ldots, s_{N_s} < n$; maximal scale $ \textstyle 2 < k_{\max} <\min_l s_l$; level $0<\alpha<1$.}
    \KwOutput{Violin plots~\cite{hintze1998violin} based on $\{t_k(r)\}_{r=1}^R$ for each $k \in \{2,\dots, k_{\max}\}$.}
    {$ \textstyle R \gets \left\lceil \left\{\frac{1}{2\alpha}\Phi^{-1}\left(1-\frac{\alpha}{2(k_{\max}-1)} \right)\right\}^2\right\rceil$, for $\Phi(\cdot)$ the CDF of a standard normal\;}
    \For(\tcp*[h]{Subsample independently in parallel}){$r \gets 1$ \KwTo $R$}
    {
    {Sample $u_r$ uniformly at random over all $s_{\lceil N_s r/R \rceil}$-subsets of nodes of $G$\;}
    {Construct the node-induced subgraph $G[u_r]$ on $|u_r|$ nodes of $G$\;}
    {$ t_2(r) \gets \#\{\textnormal{\Textpathtwo\!\! in }G[u_r]\} / \#\{\textnormal{\Textpathtwob in }G[u_r]\}$ as defined in~\eqref{count1}\;}
    \For(\tcp*[h]{Iterate over matrix powers}){$k \gets 3$ \KwTo $k_{\max}$}
    {$ t_k(r) \gets \smash{ \big( \#\{C_{k}\text{ in }G[u_r]\} / \#\{C_{k}\text{ in }K_{|u_r|}\} \big)^{1/k} }$ as defined in~\eqref{count2}\;}}
   \For(\tcp*[h]{Construct density estimates}){$k \gets 2$ \KwTo $k_{\max}$}
   {
   {Form a kernel density estimate from the set $\smash{\{t_k(r) : t_k(r)>0\}}$ using any user-defined automatic bandwidth selection method (see, e.g.,~\cite{silverman1986density})\;}
   \If{$\#\{t_k(r) : t_k(r)=0\}/R>\alpha$}
   {
   {Re-weight the kernel density estimate by $\#\{t_k(r) : t_k(r)>0\}/R$\;}
   {Add a mass at zero of height $\#\{t_k(r) : t_k(r)=0\}/R$\;}
   }
   }
   {}
{Display all $k_{\max}-1$ kernel density estimates together as violin plots.}
\caption{\label{alg:violins}Network summarization routine used in the main text}
\end{algorithm}

\begin{algorithm}[!t]
  \SetAlgoLined
    \KwInput{Simple graph $G$ on $n$ nodes; maximal scale $2 < k_{\max} < n/2$; level $0<\alpha<1$, number of subgraph sizes $N_s$, size increment $\delta > 0$.}
    \KwOutput{Multiset of $N_s$ subgraph sizes $ 1 < s_1, \ldots, s_{N_s} < n$.}
	{$s^\ast \gets \min \left\{ \max \left\{k_{\max}+1,\min\left\{\lfloor n/4\rfloor,3(k_{\max}+1)\right\}\right\} , n \right\} / \left(1+\delta\right) $\;}
    {$\mathcal{K}^\ast \gets \{2,\dots,k_{\max}\}$\;}
    {$s_{\max} \gets n$\;}
	\For(\tcp*[h]{Increment subgraph size}){$i \gets 1$ \KwTo $\lceil \log n / \log \left(1+\delta\right) \rceil$}{
	{$s^\ast \gets \min\{ [\left(1+\delta\right) s^\ast ], n\} $\;}
	{$\{t_k(r) : 2 \leq k \leq k_{\max} \}_{r=1}^R \gets $ Algorithm~\ref{alg:violins} with inputs $G$, $s^\ast$, $k_{\max}$, $\alpha$\;}
    \ForEach(\tcp*[h]{Check summary separated from 0}){$k \in \mathcal{K}^\ast$}
    {
    \eIf{$\#\{t_k(r) : t_k(r)>0\} > 0$}
    {
    {$\textstyle p_k \gets 1 - \Phi\left(\frac{1}{R} \sum_r t_k(r) / \sqrt{ \frac{1}{R-1} \sum_r t_k^2(r)  - \frac{1}{R(R-1)} \big[ \sum_r t_k(r)  \big]^2 } \right)$\;}
	}
    {$p_k \gets 1/2$\;}
    }
    {$\textstyle p \gets \max_{k\in\mathcal{K}^\ast} p_k$ \tcp{Quantify least-separated summary}}
    \If(\tcp*[h]{Reset, then halt}){$ p \leq \alpha / (k_{\max}-1) $ or $s^\ast \geq s_{\max}$}{
    \If{$s_{\max} = \lfloor 0.8 n\rfloor$}{Terminate and return $N_s$ equispaced values $s_1,\dots,s_{N_s}$ over the range $0.9 s^\ast$--$1.1 s^\ast$, rounding each to the nearest integer.}
	{$s^\ast \gets \min \left\{ \max \left\{k_{\max}+1,\min\left\{\lfloor n/4\rfloor,3(k_{\max}+1)\right\}\right\} , n \right\} / \left(1+\delta\right) $\;}
	{$\mathcal{K}^\ast \gets \{k : p_k < 1/2 \}$ \tcp{Ignore all-zero summaries}}
    {$s_{\max} \gets \lfloor 0.8 n\rfloor$ \tcp{Restrict maximum subgraph size}}
    }
    }
\caption{\label{alg:sizes}Automatic selection routine for network subsampling sizes}
\end{algorithm}

The approach we use to calculate network summaries in this study is described in Algorithms~\ref{alg:violins} and~\ref{alg:sizes}. It is based on repeatedly subsampling the network of interest, and then determining the prevalence and local variability of trees and $k$-cycles within each sub-network. The number of sub-network replicates sampled in this manner is set via Algorithm~\ref{alg:violins}, and the range of replicate sizes via Algorithm~\ref{alg:sizes}.

\subsection{Algorithm~\ref{alg:violins}: Subsampling to summarize trees and cycles}

Given an input graph $G$, Algorithm~\ref{alg:violins} randomly subsamples its nodes, yielding for each replicate $r$ a set of nodes $u_r$ and a corresponding node-induced subgraph $G[u_r]$. Summaries of trees and cycles in each $G[u_r]$ are then calculated via operations on the adjacency matrix of $G[u_r]$~\cite{alon1997finding,perepeshko2012cycle}. This procedure is repeated independently, and the resulting samples are used to form kernel density estimates~\cite{silverman1986density} summarizing the prevalence and local variability of trees and cycles throughout $G$.  These density estimates comprise the violin plots~\cite{hintze1998violin} in Figs.~\ref{fig:tadpole}--\ref{fig:twitter} in the main text.

Algorithm~\ref{alg:violins} summarizes trees in $G[u_r]$ through the ratio of the following quantities, with $|u_r|$ the number of nodes in $G[u_r]$ and $d[u_r]$ its degree sequence:
\begin{subequations}\label{count1}
\begin{align}
\#\{\textnormal{\Textpathtwo\!\! in }G[u_r]\} & = \frac{1}{2}\sum_{i=1}^{|u_r|} d[u_r]_i \, (d[u_r]_i-1), \\
\#\{\textnormal{\Textpathtwob in }G[u_r]\} & = \frac{|u_r|-2}{2}\sum_{i=1}^{|u_r|} d[u_r]_i .
\end{align}
\end{subequations}
Algorithm~\ref{alg:violins} summarizes $k$-cycles in $G[u_r]$ as the $k$th root of the proportion of cycles present, recalling that $(|u_r|)_k = |u_r|!/(|u_r|-k)!$ is the falling factorial:
\begin{subequations}\label{count2}
\begin{align}
\#\{C_{k}\text{ in }G[u_r]\} & = \text{the number of $k$-cycles in $G[u_r]$}, \\
\smash{\#\{C_{k}\text{ in }K_{|u_r|}\}} & = \frac{1}{2k}(|u_r|)_k.
\end{align}
\end{subequations}

To determine the total number of replicates $R$, Algorithm~\ref{alg:violins} relies on the following idea. Consider the cumulative distribution function (CDF) $F_{t_k}(\cdot)$ of a given $t_k(r)$.  Conditionally upon the observed graph of interest $G$, it follows from the central limit theorem that when $R$ is sufficiently large, the corresponding empirical CDF $\hat{F}_{t_k}(\cdot)$ behaves like a normal distribution with mean $F_{t_k}(\cdot)$ and variance $\smash{F_{t_k}(\cdot)\big[1-F_{t_k}(\cdot)\big]/R}$. To constrain departures of $\smash{\hat{F}_{t_k}(\cdot)}$ from $F_{t_k}(\cdot)$ greater than $\alpha$ to occur with probability bounded by $\alpha$, simultaneously taking into account via the Bonferroni correction that we wish this to hold for all $k_{\max}-1$ scales under consideration, we may formulate the following requirement:
\begin{equation}
\label{eq:normalpha}
\mathbb{P}\left(\left|\hat{F}_{t_k}(\cdot)-{F}_{t_k}(\cdot)\right|>\alpha\right)\leq \frac{\alpha}{k_{\max}-1}.
\end{equation}
Replacing the left-hand side of~\eqref{eq:normalpha} by the corresponding normal CDF and noting that $F_{t_k}(\cdot)\big[1-F_{t_k}(\cdot)\big] \leq 1/4$, we obtain the expression for $R$ given in Algorithm~\ref{alg:violins}.

Algorithm~\ref{alg:violins} requires as input a set of subgraph sizes $s_1, \ldots, s_{N_s}$, a maximal scale $k_{\max}$, and a level $\alpha$.  While it is possible to use just a single subgraph size, this can lead to discretization effects in the values of the summaries obtained when $G$ is small. To avoid this eventuality, we vary subgraph sizes around an average value; both the selection of this value and the variation around it are done automatically in this study using Algorithm~\ref{alg:sizes}.

\subsection{Algorithm~\ref{alg:sizes}: Selecting network sizes for subsampling}

Algorithm~\ref{alg:sizes} must determine subgraph sizes which are neither too small, in which case $k$-cycles may be absent simply because of network edge sparsity~\cite{rito2010threshold}---nor too large, in which case independently chosen subgraph replicates will overlap significantly, and be highly correlated as a result. To automate this procedure, Algorithm~\ref{alg:sizes} steps through an increasing sequence of trial values for the number of subgraph nodes $s^\ast$. For each trial value, Algorithm~\ref{alg:sizes} checks if summaries of $G$ obtained from Algorithm~\ref{alg:violins} have stabilized, using the fact that if $k$-cycles are present, they will appear in the output of Algorithm~\ref{alg:violins} only after $s^\ast$ exceeds a sparsity-dependent threshold~\cite{rito2010threshold}.

If not all $k_{\max}-1$ summaries based on a given value of $s^\ast$ are sufficiently well separated from zero, as determined by a one-sided $Z$-test at level $\alpha$ with Bonferroni correction, then Algorithm~\ref{alg:sizes} increments $s^\ast$ and the procedure repeats. Once all summaries are sufficiently well separated from zero, or if eventually $s^\ast$ reaches the full number of nodes in $G$, then the procedure restarts with the initial value of $s^\ast$---but any scales whose summaries were previously identically zero are then excluded from testing. The remaining summaries are once again tested and $s^\ast$ is incremented up to a maximum of $\lfloor 0.8 n\rfloor$, with Algorithm~\ref{alg:sizes} returning equispaced values $s_1,\dots,s_{N_s}$ over the range $0.9 s^\ast$--$1.1 s^\ast$, each rounded to the nearest integer.

\section{Network datasets}\label{sec:S7}

We analyzed all networks in this study using Algorithms~\ref{alg:violins} and~\ref{alg:sizes}, with $k_{\max} = 9$, $\alpha=0.01$, $N_s=21$, and $\delta = 0.05$. We now describe these networks.

\subsection{Networks analyzed in Fig.~\ref{fig:compnets} of the main text}

The networks analyzed in Fig.~\ref{fig:twitter} of the main text are standard benchmark datasets in the literature (save for Fig.~\ref{fig:compnets}d, synthetically generated to illustrate triadic closure), which we now describe. Archetypical synthetic examples corresponding to Figs.~\ref{fig:compnets}a--c are shown in Fig.~\ref{fig:typexamp} for comparison.

\begin{figure}[!t]
\centering
  \subfloat[]{%
    \label{fig:PrefAttach}%
    \includegraphics[width=.25\textwidth]
      {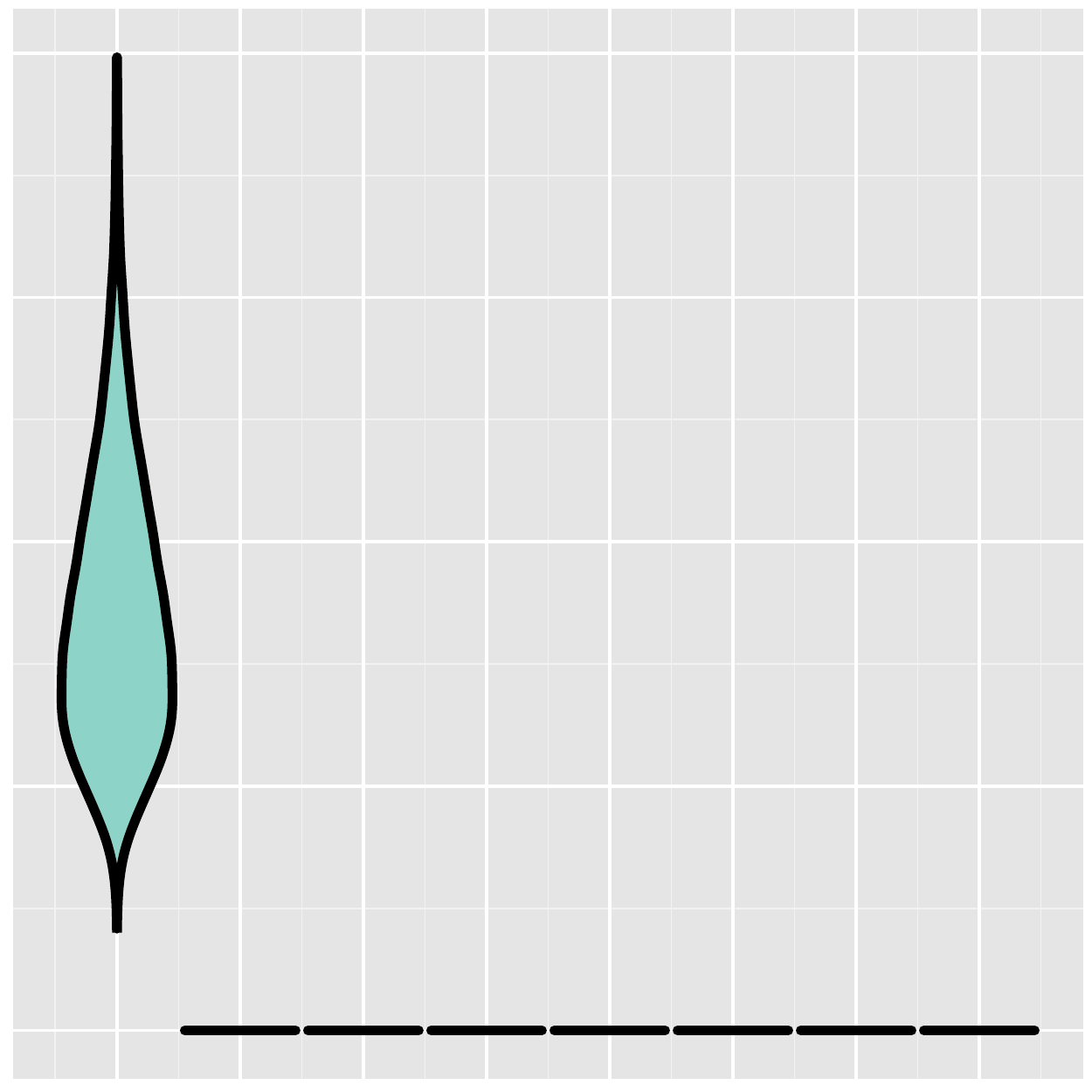}}\ \ %
  \subfloat[]{%
    \label{fig:WatStro}%
    \includegraphics[width=.25\textwidth]
      {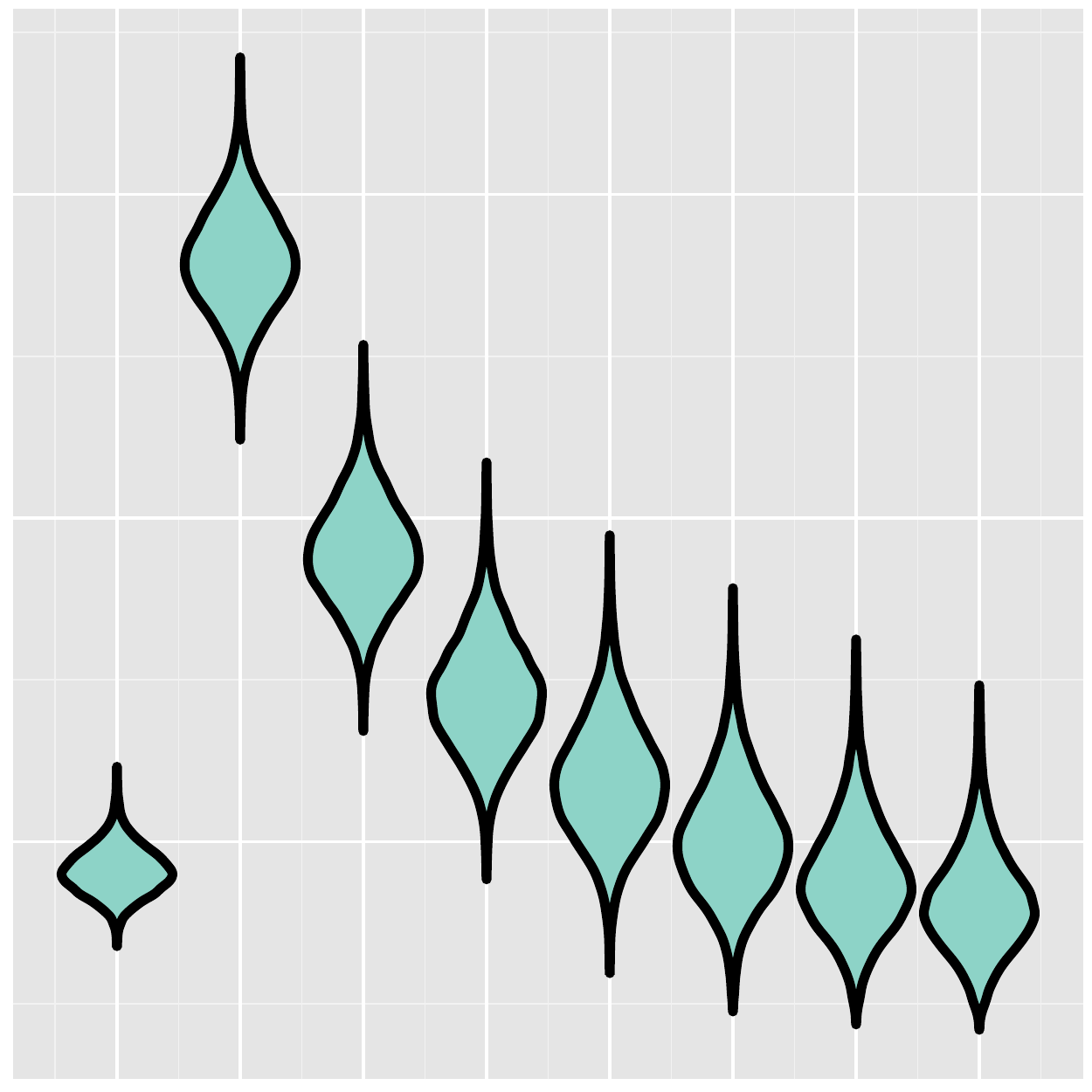}}\ \ %
  \subfloat[]{%
    \label{fig:Blockmodel}%
    \includegraphics[width=.25\textwidth]
      {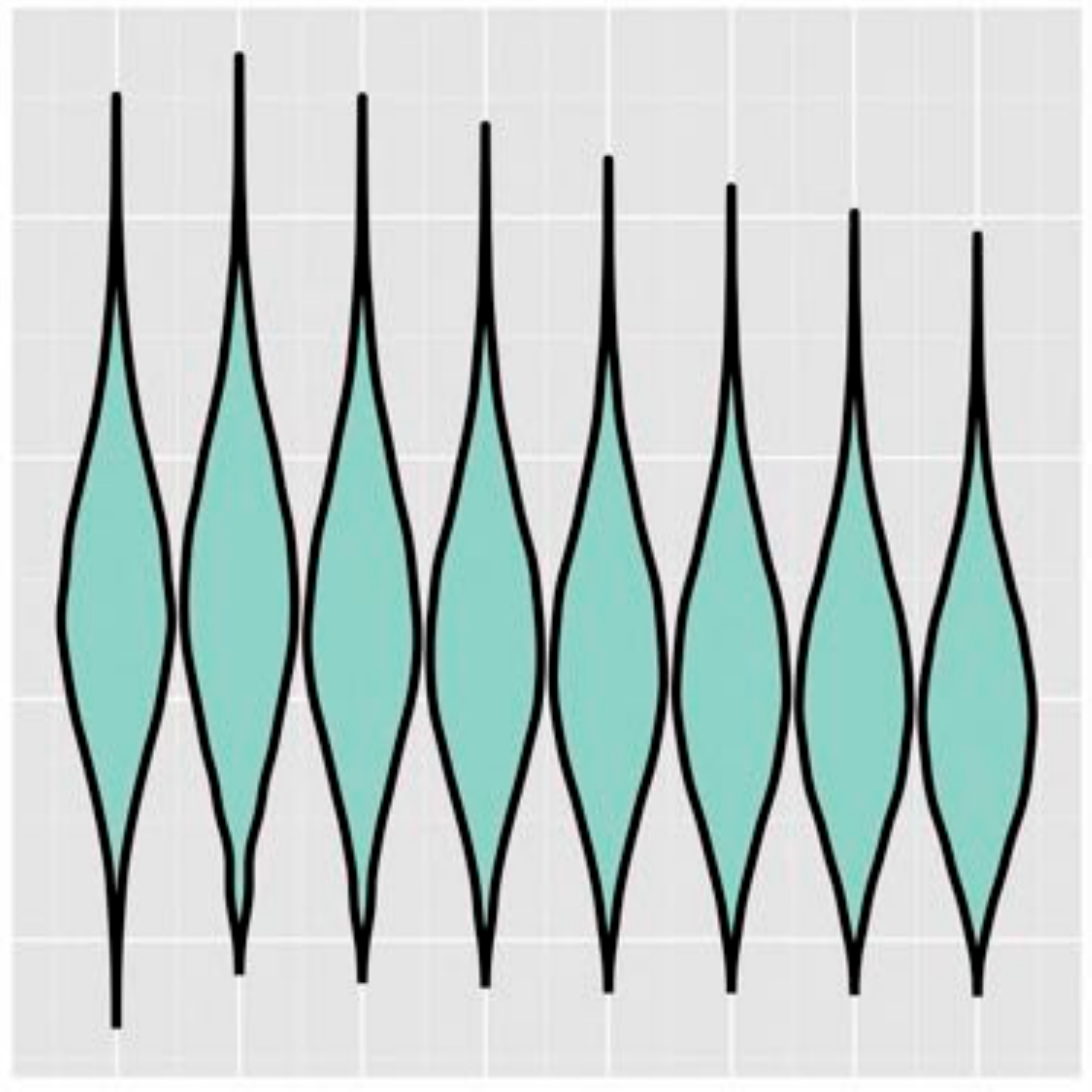}}%
  \caption{\label{fig:typexamp} \textbf{Summaries of synthetic networks.} Generated according to archetypical network models: \protect\textbf{a}, Preferential attachment model (1000 nodes, $s^\ast=322$). \protect\textbf{b} Watts--Strogatz model (500 nodes, $s^\ast=150$). \protect\textbf{c} Blockmodel (100 nodes, $s^\ast=33$).}
\end{figure}

\paragraph{Figure~\ref{fig:compnets}a} This 453-node metabolic network of \emph{C. elegans} was introduced in~\cite{jeong2000large} as an example of a network with scale-free properties. It may be obtained from \url{http://www.cise.ufl.edu/research/sparse/matrices/Arenas/}. For purposes of comparison, Fig.~\ref{fig:typexamp}a shows the scale-based summary of a preferential attachment network (a tree), generated with default parameter settings in the R package igraph~\cite{igraph}.

\paragraph{Figure~\ref{fig:compnets}b} This 4,941-node electrical power grid network of the western US was introduced in~\cite{watts1998small} as an example of a network with small-world properties. It may be obtained from \url{http://www.cise.ufl.edu/research/sparse/matrices/Newman/power}. For comparison, Fig.~\ref{fig:typexamp}b shows the scale-based summary of a synthetically generated small-world network, generated using the R package igraph~\cite{igraph} by perturbing a one-dimensional lattice to connect each group of five neighbors and then rewiring at random with probability 0.05.

\paragraph{Figure~\ref{fig:compnets}c} This 1,224-node network of political weblogs was introduced in~\cite{adamic2005political} as an example of network with known two-community structure. It may be obtained from \url{http://www.cise.ufl.edu/research/sparse/matrices/Newman/polblogs}. Here we study the version analyzed in~\cite{olhede2013network}. For comparison, Fig.~\ref{fig:typexamp}b shows the scale-based summary of a synthetically generated degree-corrected stochastic blockmodel with two communities.
 
\paragraph{Figure~\ref{fig:compnets}d} These two 512-node networks demonstrate triadic closure. The first network shown in Fig.~\ref{fig:compnets}d is formed by connecting pairs of nodes independently with constant edge probability $0.0380$. The second network is derived from the first as follows. We select a triplet of nodes uniformly at random, mark one of the three nodes, and then check if it is connected to the remaining two. If it is, and no third edge is present, then we close the triangle and delete another edge chosen at random from within the network---thus keeping the overall number of edges fixed. We repeat this procedure independently, 10,000 times. As a result, from the first to the second network, the number of triangles increases from 1303 to 3478.

\paragraph{ Figure~\ref{fig:compnets}e} This 1,117-node network of student friendships derives from the US National Longitudinal Study of Adolescent Health (Add Health)~\cite{resnick1997protecting}. We study the version analyzed in~\cite{olhede2013network}, after removing an additional 5 students with missing grade covariates. This yields six networks for middle school (grades 7--8) and high school (grades 9--12): 209 nodes (grade 7), 242 nodes (grade 8), 237 nodes (grade 9), 161 nodes (grade 10), 137 nodes (grade 11), and 131 nodes (grade 12).

\subsection{Networks analyzed in Fig.~\ref{fig:twitter} of the main text}

The social media discussion dataset analyzed in Fig.~\ref{fig:twitter} of the main text was constructed by the now-defunct private company FSwire Limited~\cite{FSwireCompaniesHouse} as follows. First, from all tweets broadcasted via the social media platform Twitter over a given time period, FSwire applied proprietary algorithms to retain only the most relevant tweets ``to extract content that impacts the capital markets''~\cite{FSwireBloomberg}. Then, FSwire extracted the tweets related to Apple Inc., labeled using the AAPL ticker in its database~\cite{FSwireFinanceMagnates}. This dataset consisted of 10,000 tweets shared online between $t_1 = $ 4:18:19pm on April 8, 2014 and $t_2 = $ 9:54:45am on May 15, 2014. Finally, of these 10,000 tweets, only those containing the keyword ``iPhone'' were retained, yielding $N=1986$ tweets from 939 unique users.

We use these 1986 tweets to construct two types of time-varying networks. In both cases, nodes in the networks are users broadcasting tweets. The differences between the two constructions we consider lie in how and when edges are considered to be present between pairs of users. In both cases we remove nodes without any connections to other users during the time period under consideration.

\begin{figure}[!t]
\centering
  \subfloat[]{%
    \includegraphics[width=0.19\textwidth]
      {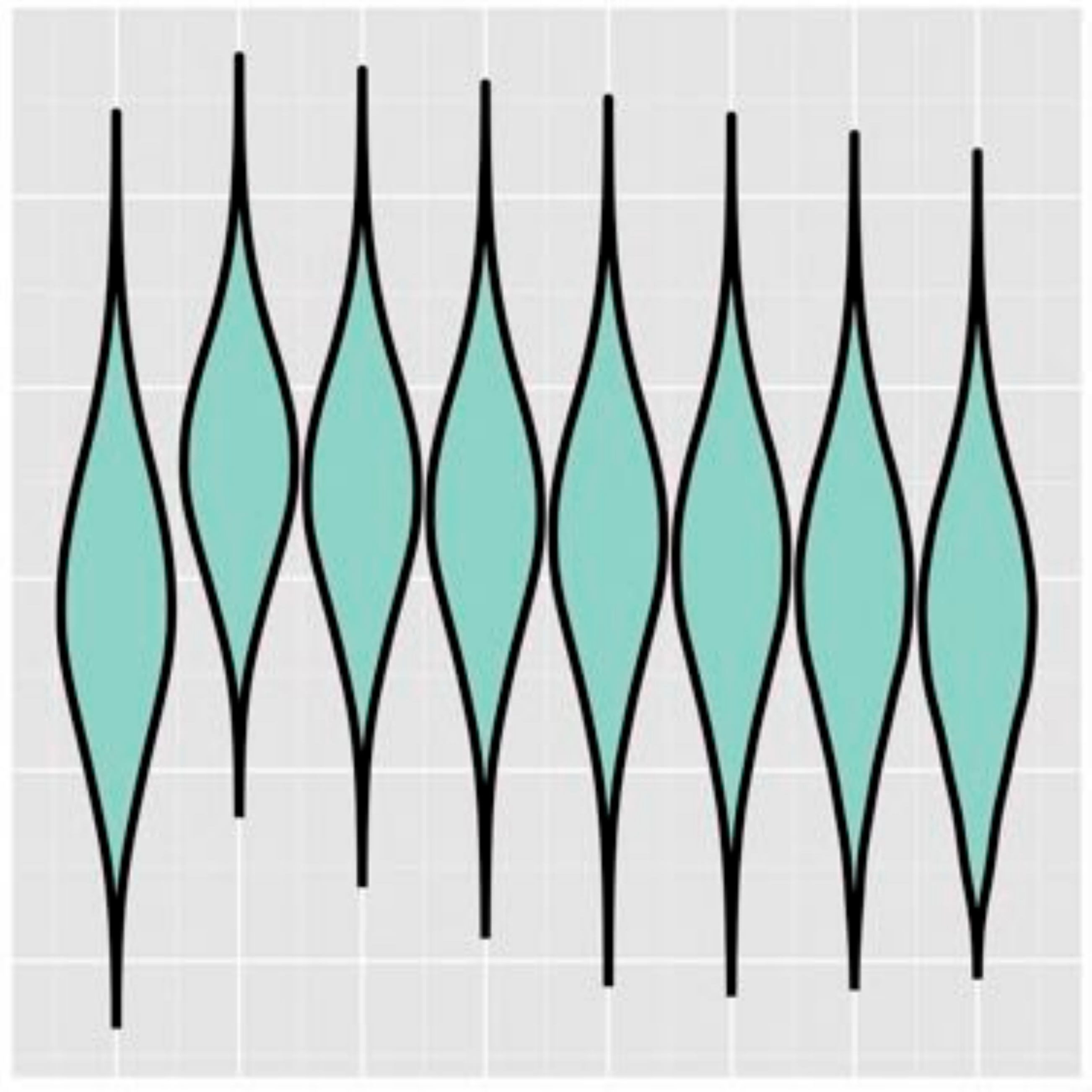}}\ \ %
  \subfloat[]{%
    \includegraphics[width=0.19\textwidth]
      {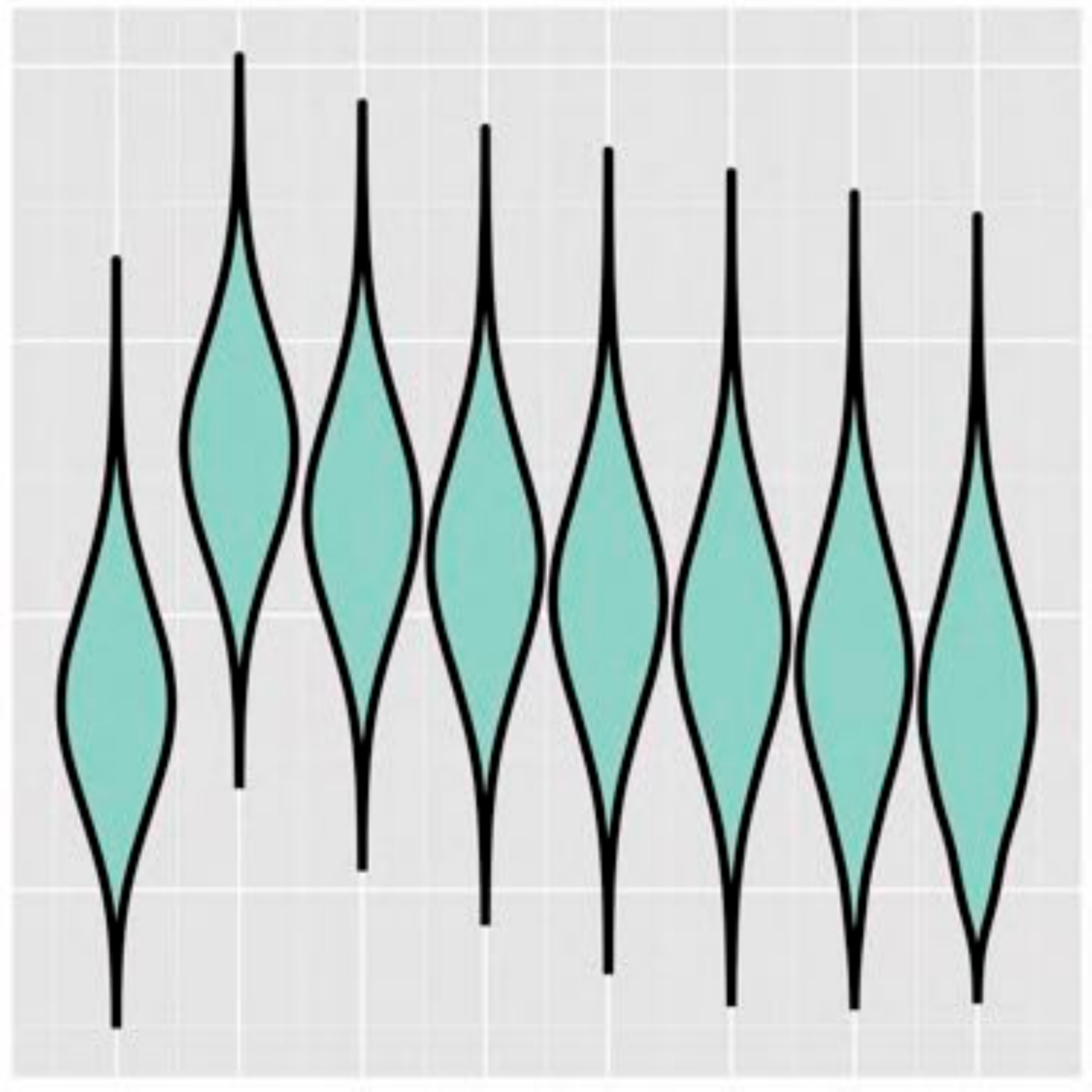}}\ \ %
  \subfloat[]{%
    \includegraphics[width=0.19\textwidth]
      {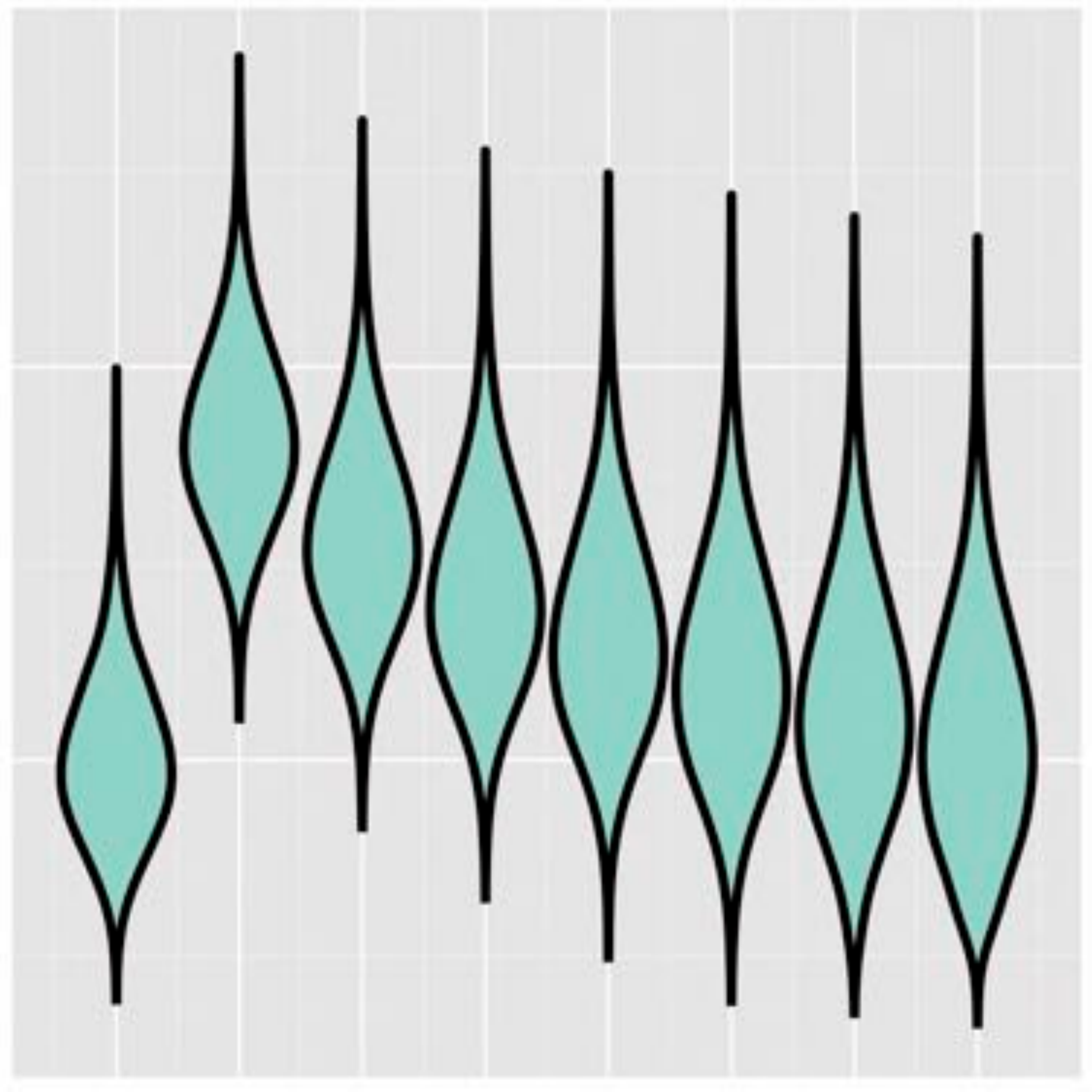}}\ \ %
  \subfloat[]{%
    \includegraphics[width=0.19\textwidth]
      {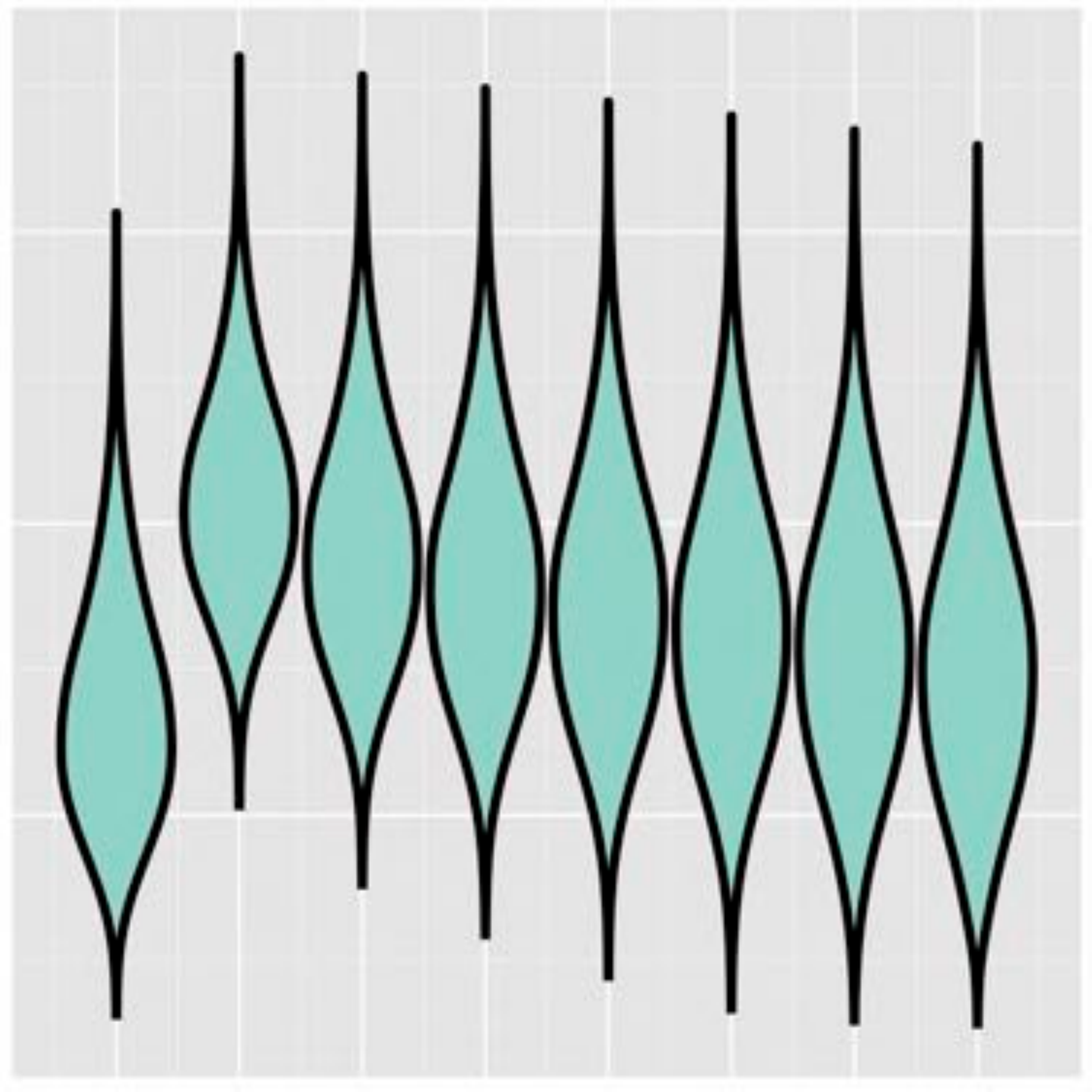}}\ \ %
      \\
  \subfloat[]{%
    \includegraphics[width=0.19\textwidth]
      {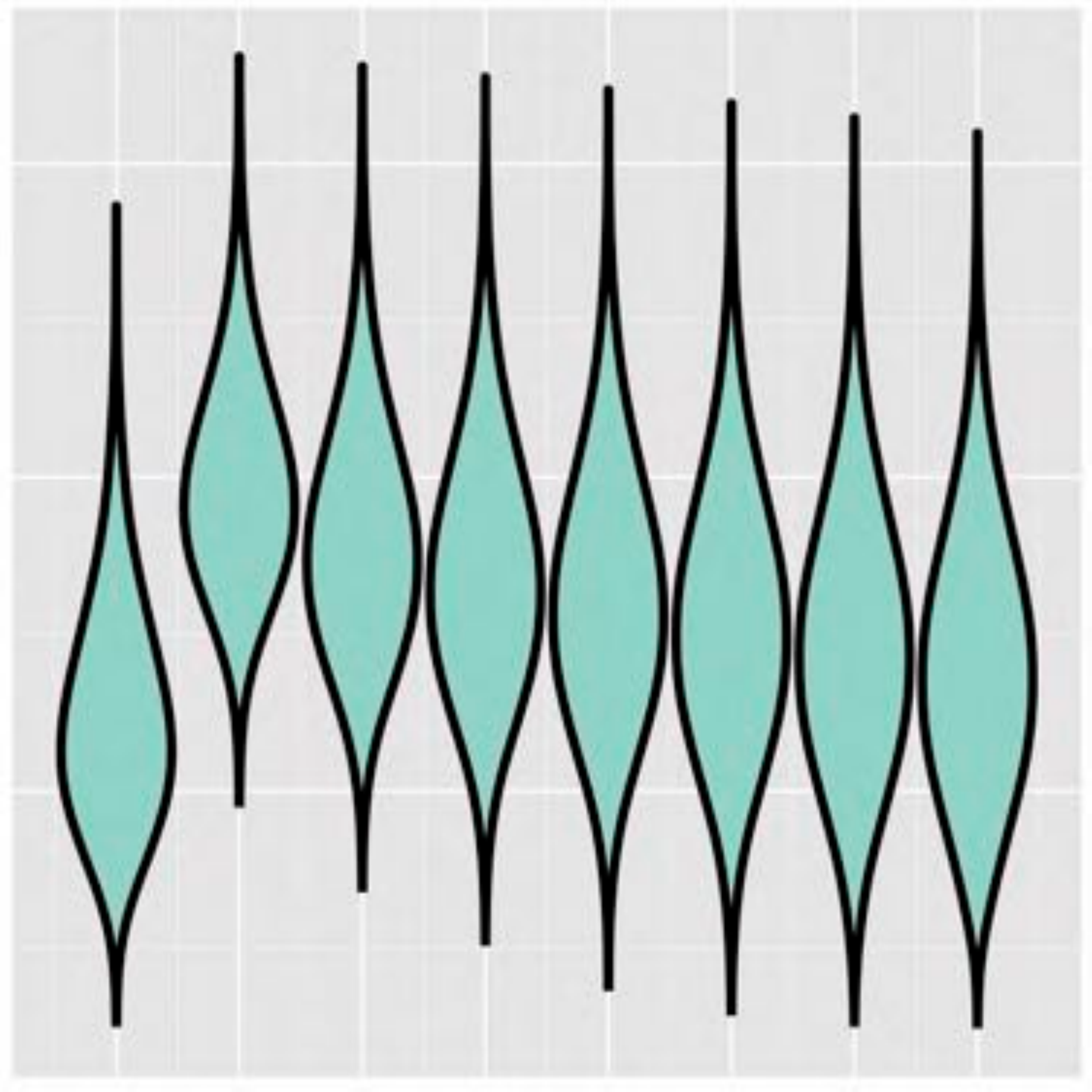}}\ \ %
  \subfloat[]{%
    \includegraphics[width=0.19\textwidth]
      {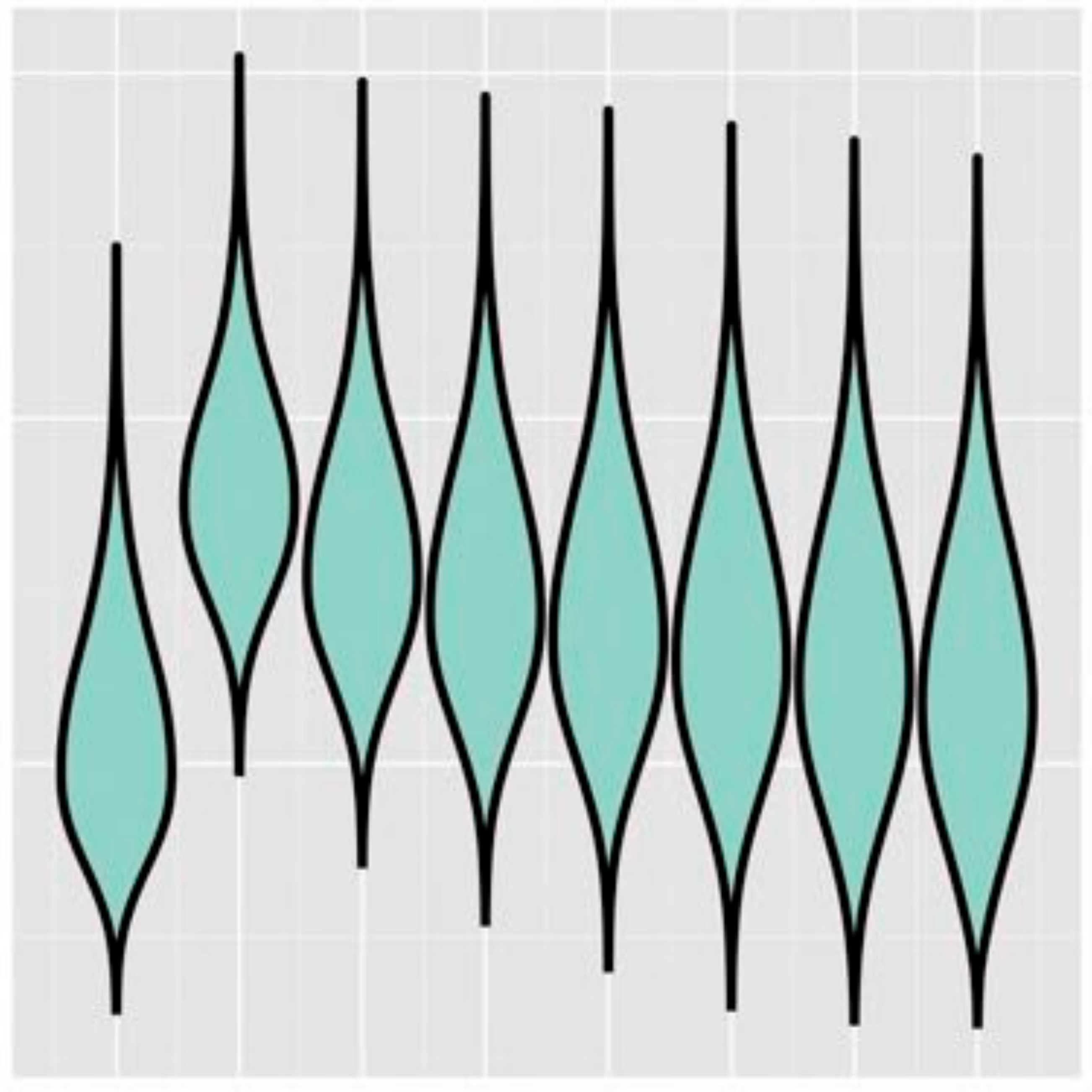}}\ \ %
  \subfloat[]{%
    \includegraphics[width=0.19\textwidth]
      {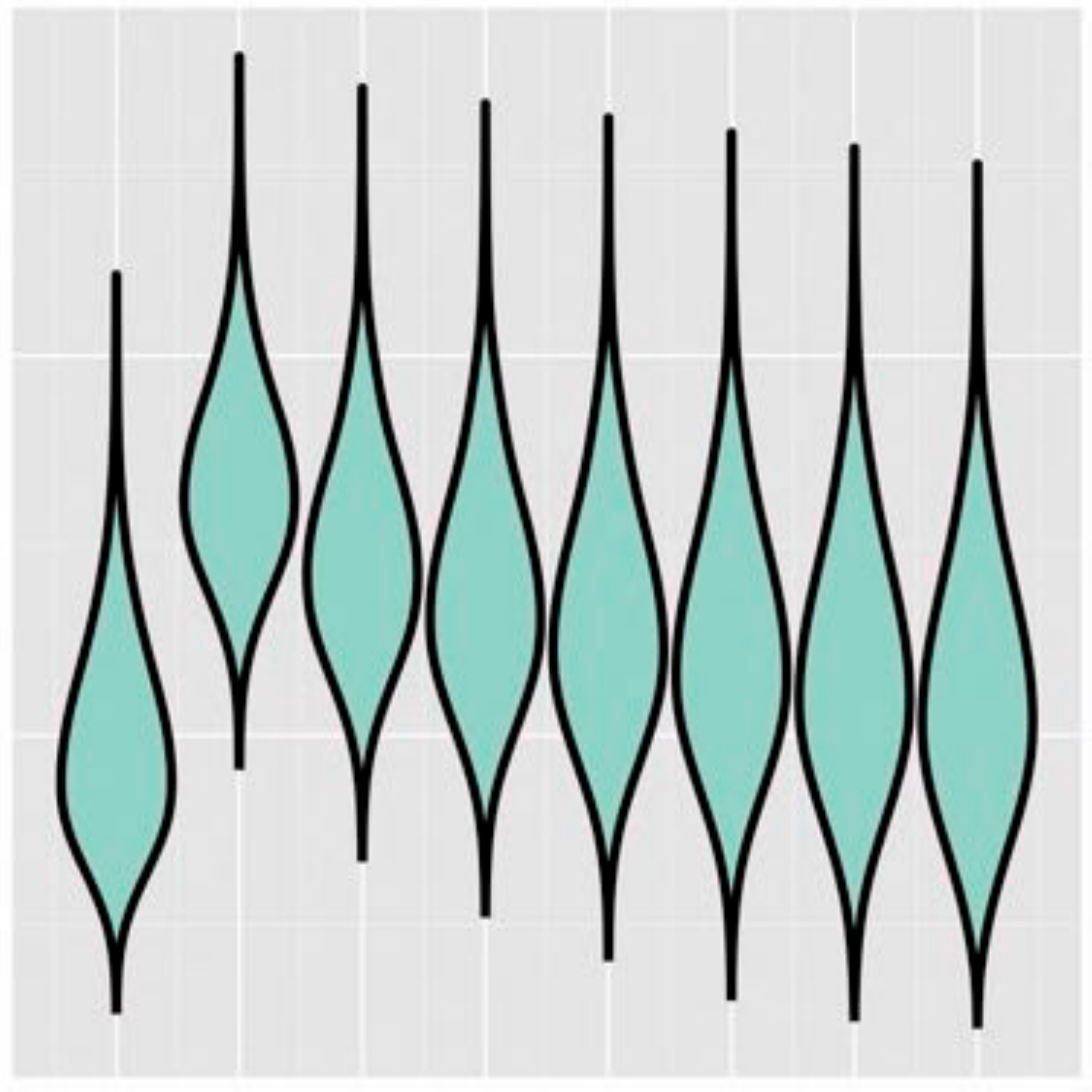}}\ \ %
  \subfloat[]{%
    \includegraphics[width=0.19\textwidth]
      {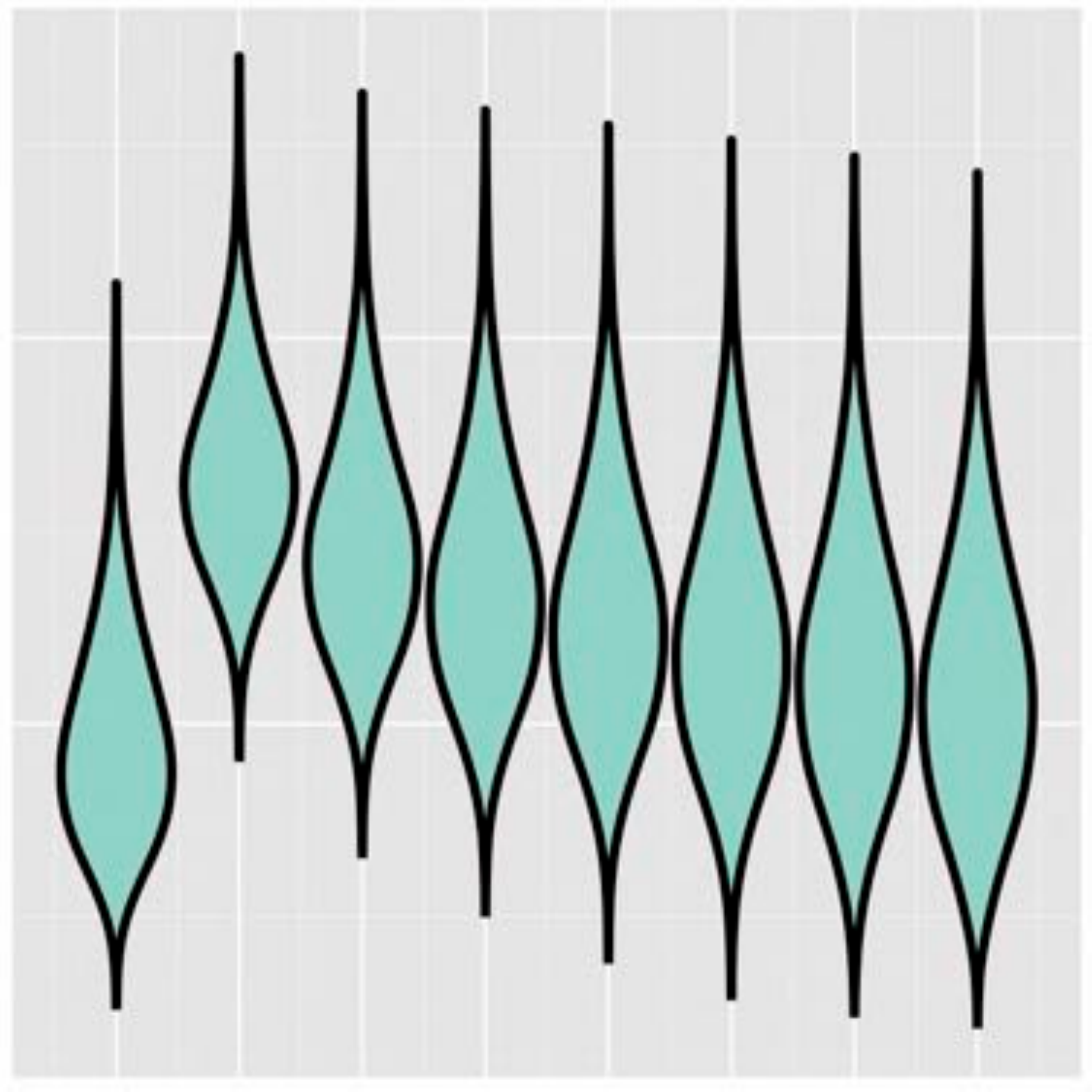}}%
  \caption{\label{fig:tweetevol} \textbf{Summaries of the aggregate network in Fig.~\ref{fig:twitter} in the main text.} Constructed using Algorithm~\ref{alg:sizes}, starting from $t_1 = $ 4:18:19pm on April 8, 2014 through: \protect\textbf{a}, April 14 (4:29:55pm). \protect\textbf{b} April 19 (pm). \protect\textbf{c} April 24 (pm). \protect\textbf{d}, May 2 (4:29:55am). \protect\textbf{e}, May 4 (pm). \protect\textbf{f}, May 9 (pm). \protect\textbf{g}, May 14 (pm). \protect\textbf{h}, May 15 (am). Summaries \protect\textbf{a} and \protect\textbf{h} also appear as Figs.~\ref{fig:twitter}b and~\ref{fig:twitter}f, respectively, in the main text.}
\end{figure}

\paragraph{An aggregate network} The first network we consider is one in which users are connected if they broadcasted the same Uniform Resource Locator (URL) in tweets. The use of common URLs to determine edges in the network follows from the observation that tweets often contain a URL linking to an article or blog post. Thus, at time point $t\in[t_1,t_2]$, two users are connected if they broadcasted any tweets containing the same URL at {\emph{any}} point in the time-window $[t_1,t]$. Any edges that appear then remain in the network, so that it accumulates edges over time.

Figure~\ref{fig:tweetevol} shows scaled-based summaries of this aggregate network at eight different time points. These summaries evolve quickly from a relatively uniform set of contributions across scales (cf.~Fig.~\ref{fig:compnets}c) to a steady state configuration which is reminiscent of small-world connectivity (cf.~Fig.~\ref{fig:compnets}b).

\begin{figure}[!t]
\centering
  \subfloat[]{%
    \includegraphics[width=0.19\textwidth]
      {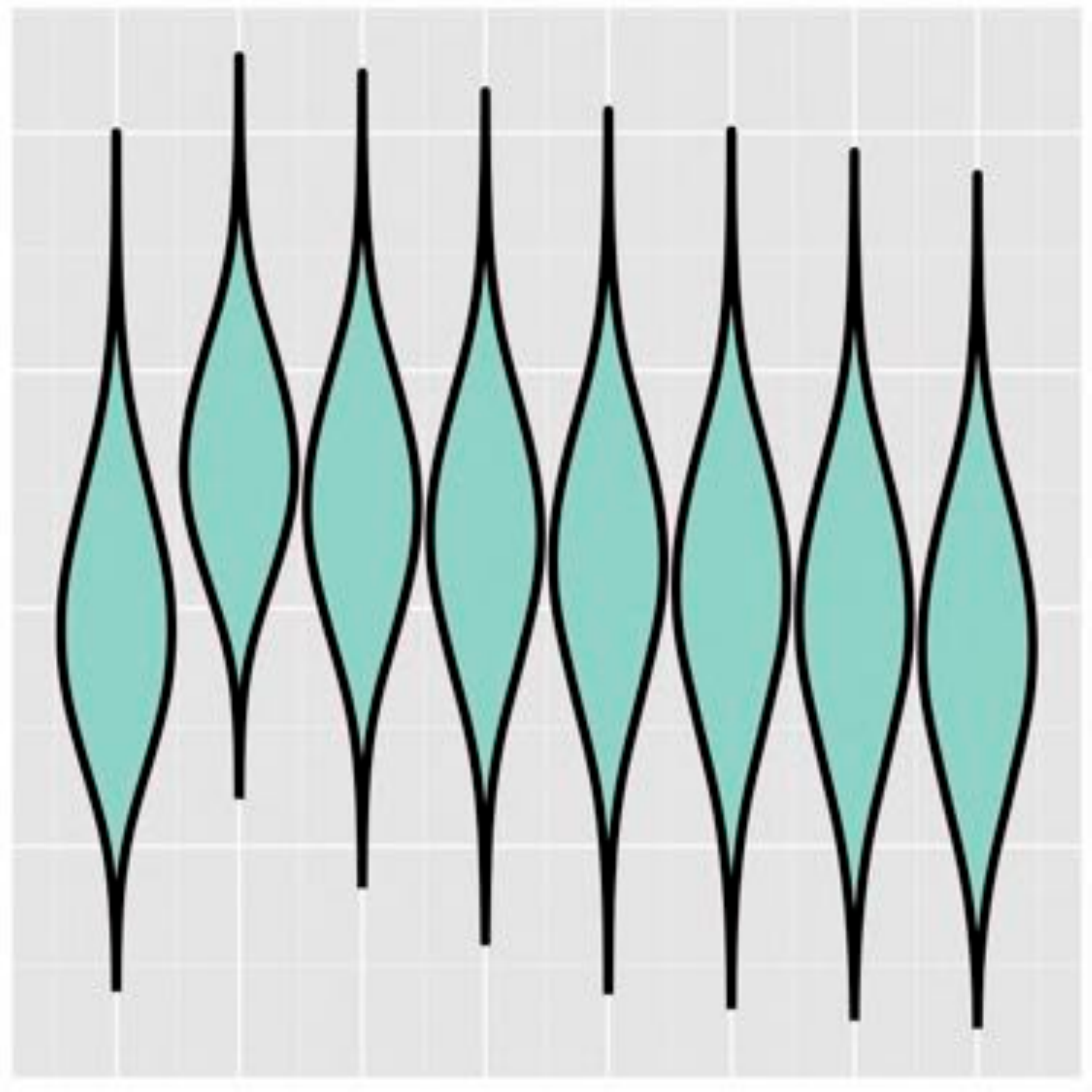}}\ \ %
  \subfloat[]{%
    \includegraphics[width=0.19\textwidth]
      {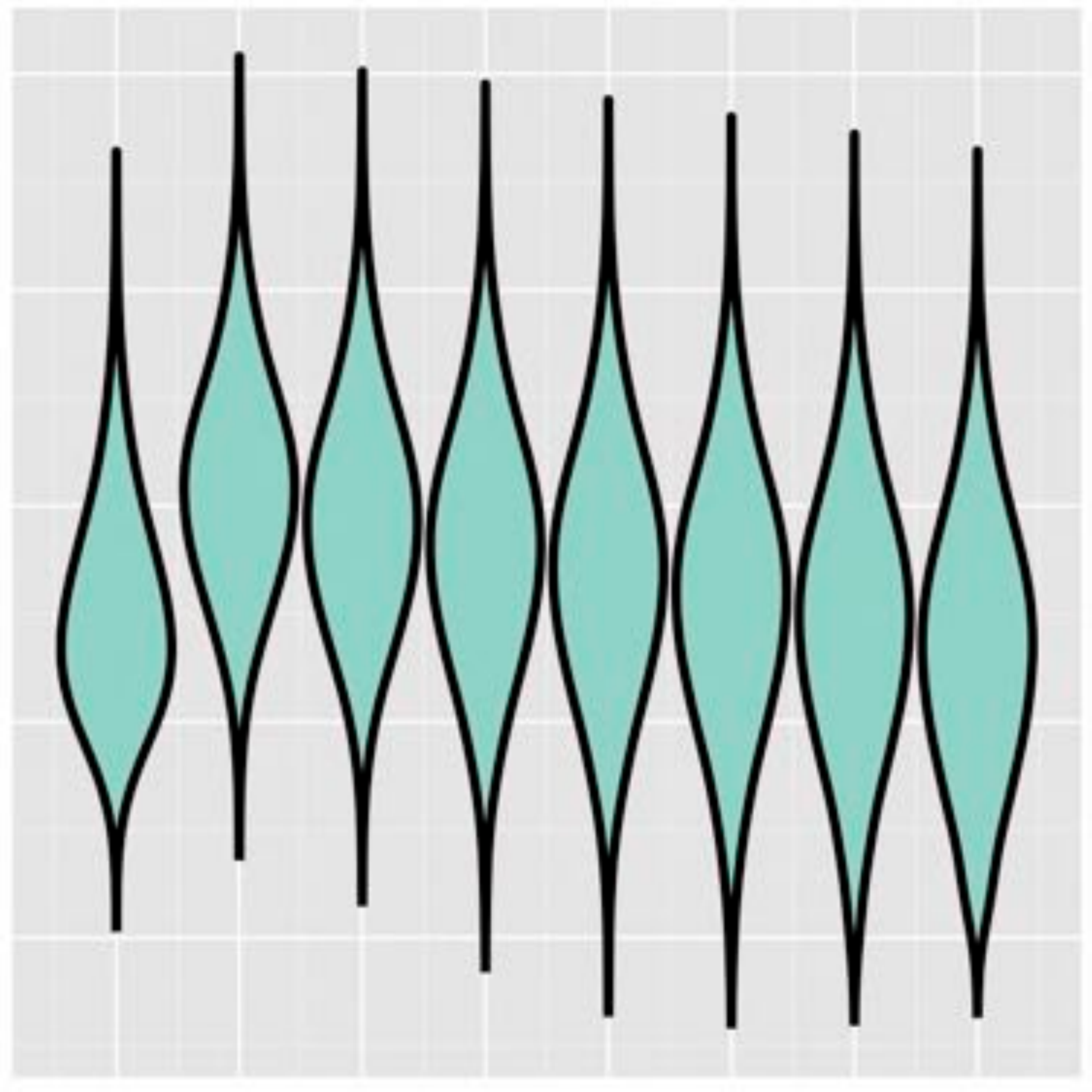}}\ \ %
  \subfloat[]{%
    \includegraphics[width=0.19\textwidth]
      {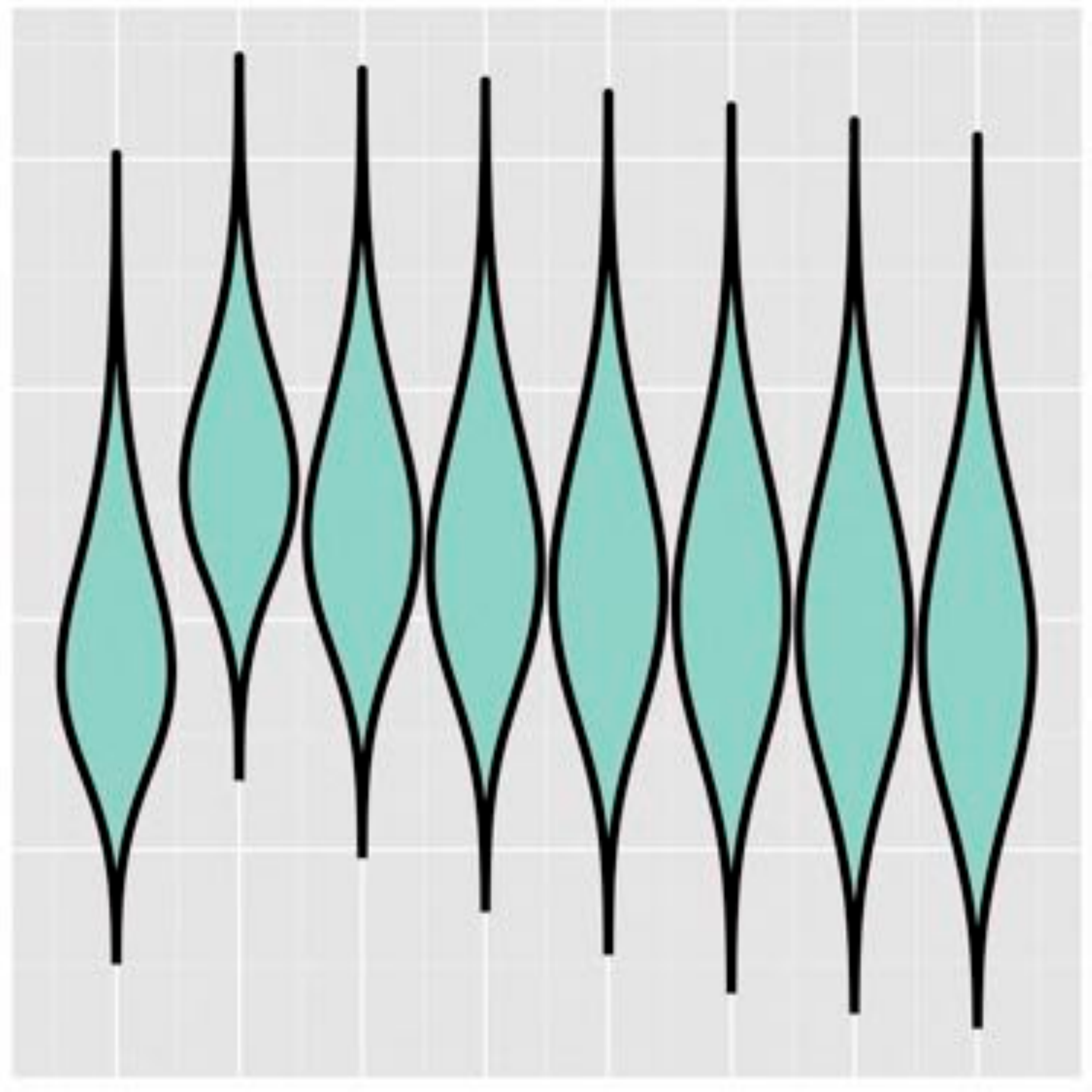}}\ \ %
  \subfloat[]{%
    \includegraphics[width=0.19\textwidth]
      {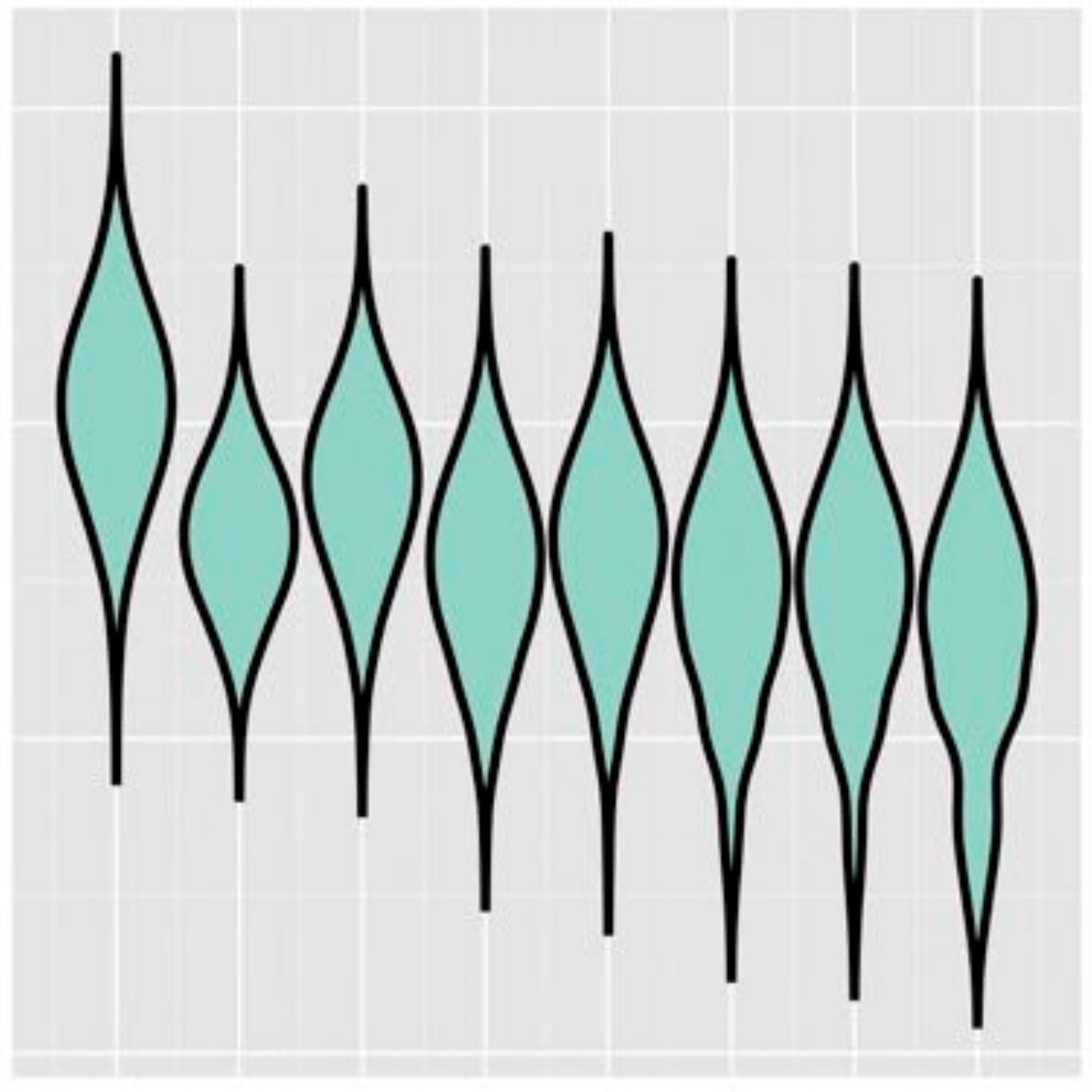}}\ \ %
      \\
  \subfloat[]{%
    \includegraphics[width=0.19\textwidth]
      {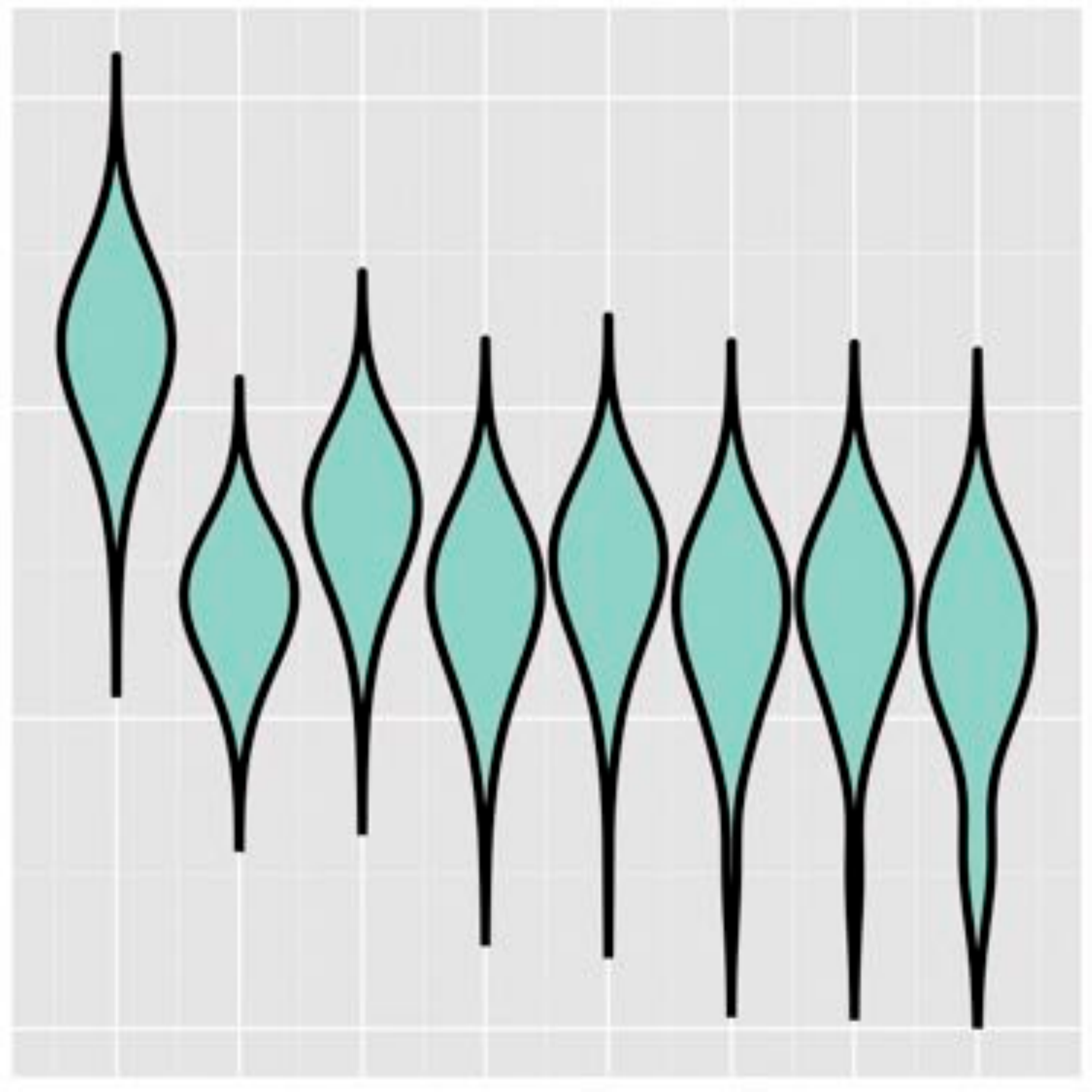}}\ \ %
  \subfloat[]{%
    \includegraphics[width=0.19\textwidth]
      {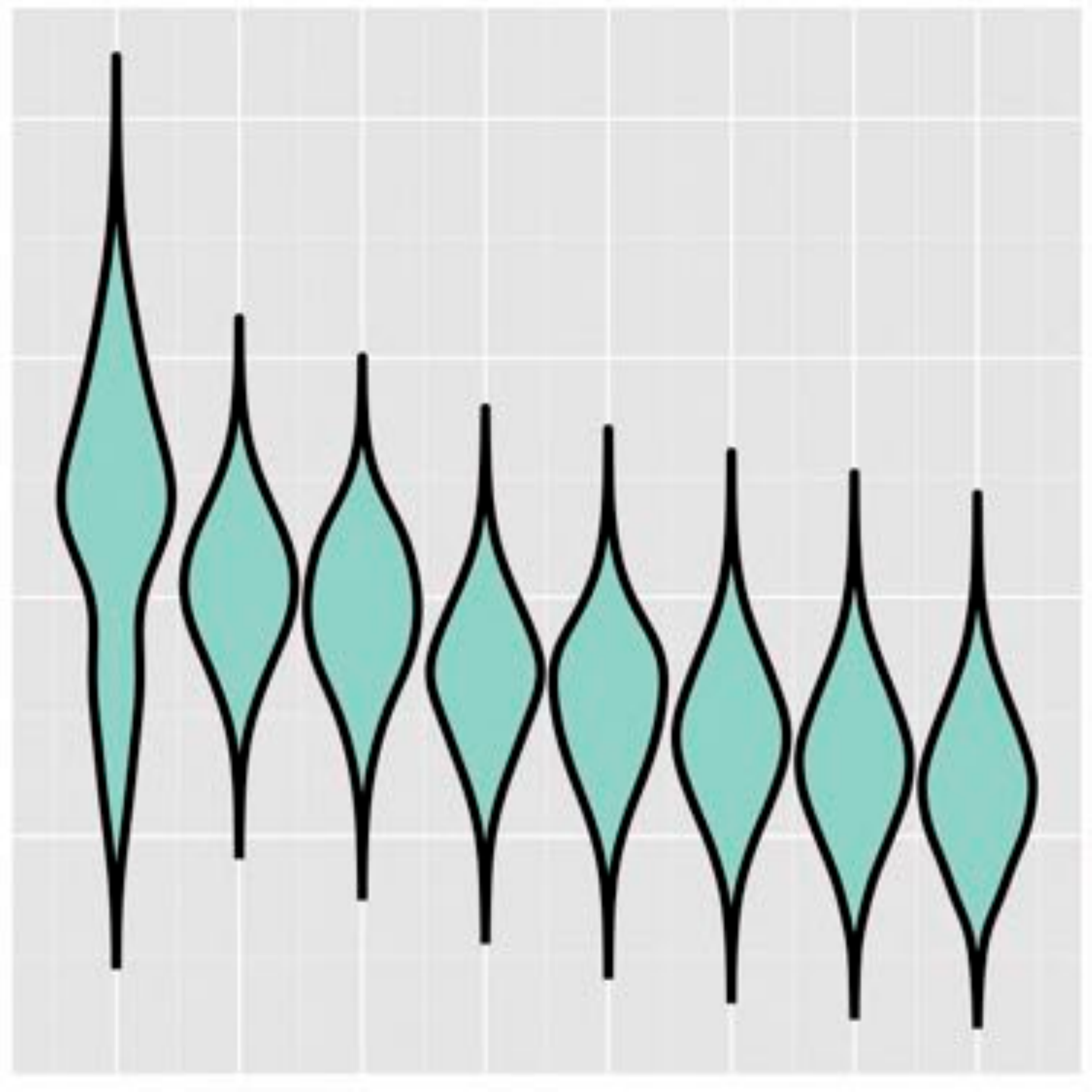}}\ \ %
  \subfloat[]{%
    \includegraphics[width=0.19\textwidth]
      {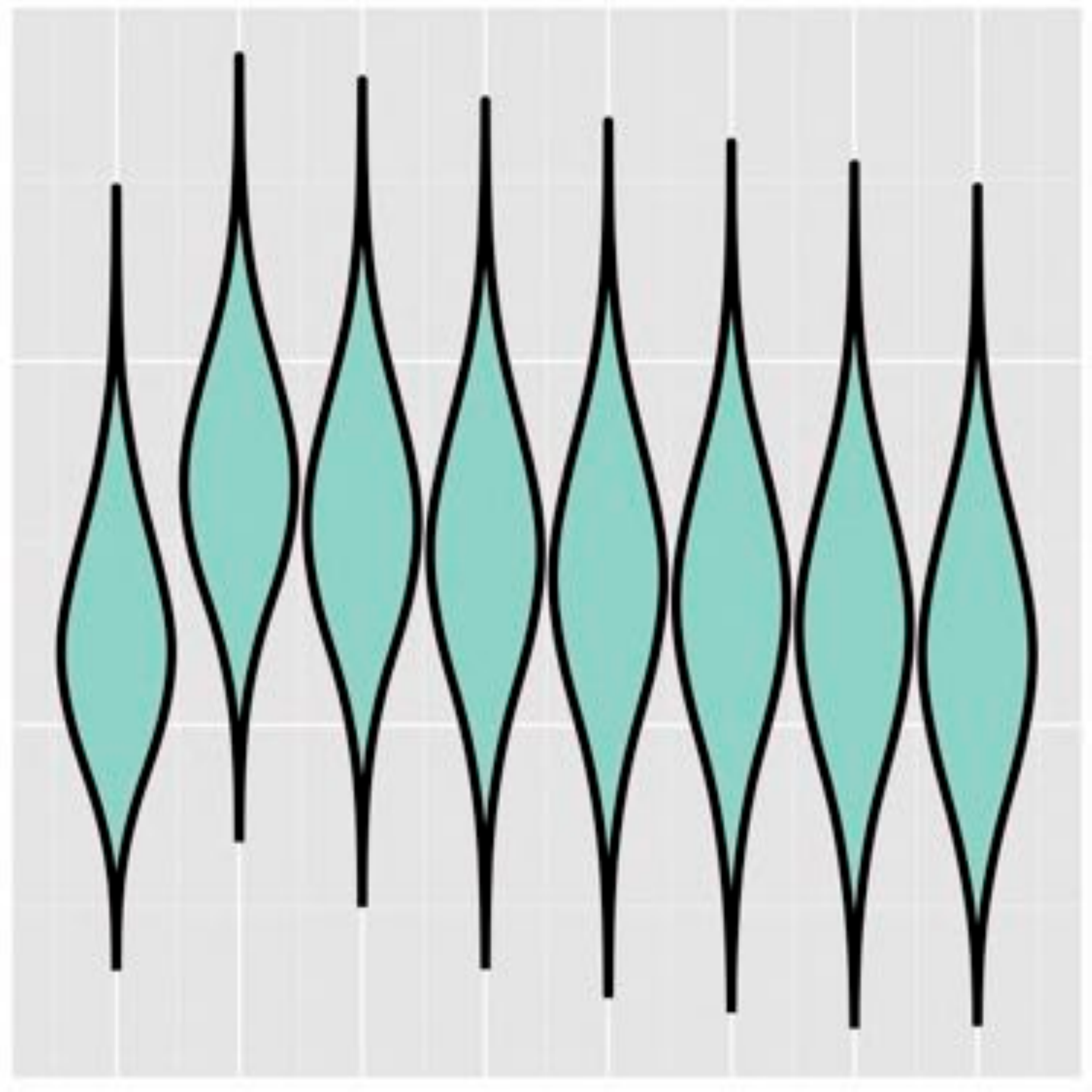}}\ \ %
  \subfloat[]{%
    \includegraphics[width=0.19\textwidth]
      {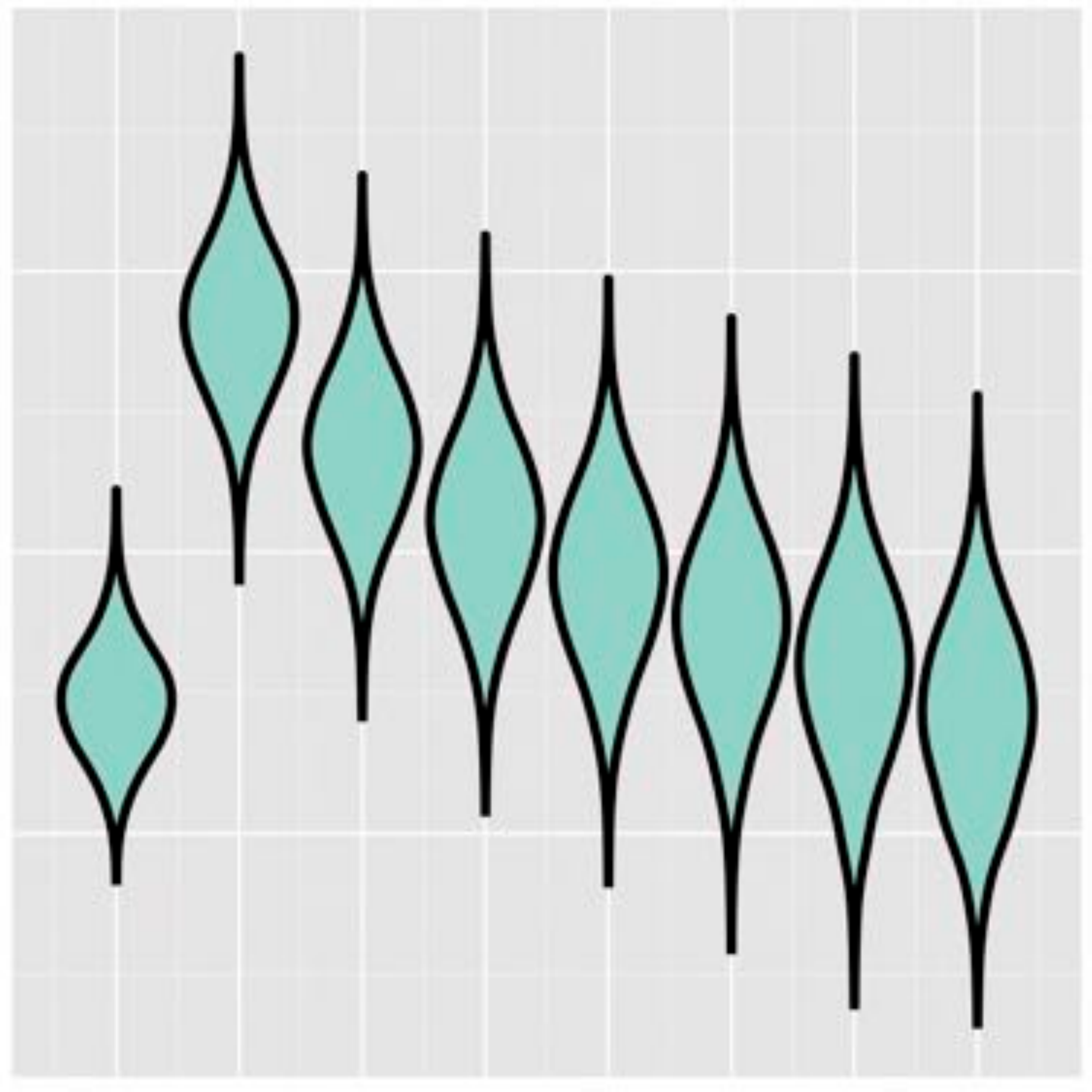}}%
  \caption{\label{fig:tweetevol2} \textbf{Time-localized networks from Fig.~\ref{fig:twitter} in the main text.} Constructed using Algorithm~\ref{alg:sizes} and the following 4-day windows: \protect\textbf{a}, April 11--15 (04:29:55am). \protect\textbf{b}, April 16--20 (am). \protect\textbf{c}, April 21--25 (am). \protect\textbf{d}, April 27--May 1 (am). \protect\textbf{e}, April 28--May 2 (am). \protect\textbf{f}, May 1--5 (am). \protect\textbf{g}, May 6--10 (am). \protect\textbf{h}, May 11--15 (am). Summaries \protect\textbf{a} and \protect\textbf{e} also appear as Figs.~\ref{fig:twitter}a and~\ref{fig:twitter}c, respectively, in the main text.}
\end{figure}

\paragraph{Time-localized networks} We also consider a collection of time-localized networks, with edges based on the similarity of tweets broadcasted within four days of one another.  This leads to denser networks relative to the above aggregate network over similar time windows, since users may broadcast similar information but point to different sources for it (see timeline below for examples). Here we consider two users to be connected if, within a given four-day window, they have sent at least one pair of tweets whose edit distance distance is no greater than 29 (just over 20\% of the 140-character limit of a tweet). Thus an edge is present between two users if, with at most 29 basic string modifications according to the restricted Damerau--Levenshtein distance~\cite{Damerau}, one user's tweet can be transformed into the other's tweet.

In this study we used the R package stringdist~\cite{stringdist} to evaluate the edit distance between all pairs of tweets in a given four-day window. We experimented with shorter and longer edit distances; in the former case, the networks obtained become very sparse, whereas in the latter case, with an edit distance of 30--40, we found many qualitative characteristics of the resulting networks to be similar. Once the allowable edit distance is increased significantly past 40, a large proportion of tweets can be transformed into one another.  We also experimented with shorter and longer time windows (2, 3, and 5 days), selecting four-day windows in half-day overlapping increments (66 in total between $t_1$ and $t_2$) as a compromise between network size, proportion of nodes in the largest connected component, and precision of time localization.

Figure~\ref{fig:tweetevol} shows scaled-based summaries of eight of these time-localized networks, at approximately equally spaced time points but with a greater density following the rapid growth in network activity in late April (see Fig.~\ref{fig:twitter}). We see that after starting from a relatively uniform set of contributions across scales (cf.~Fig.~\ref{fig:compnets}c), these networks undergo a transition towards more dominant tree-like structure (cf.~Fig.~\ref{fig:compnets}a) before relaxing towards a configuration reflecting small-world characteristics (cf.~Fig.~\ref{fig:compnets}b).

\paragraph{Timeline of events related to the social media discussion} Finally, we consider a timeline of iPhone-related news articles. As described in the main text, various rumors, leaks and news releases impacted the social media discussion under study. Google news single-day searches from March 30--May 30, 2014, first for ``Apple'' and then separately for ``iPhone,'' yield the following timeline of relevant events:
      
\begin{description}
\item[March 31, 2014] First iPhone photos start leaking to \emph{Sina Weibo}~\cite{weibo1,9to5mac1}.
\item[April 1, 2014] \emph{Reuters} reports iPhone suppliers to begin display production~\cite{reuters2}.
\item[April 15, 2014] Leaks of iPhone part images reported around the world~\cite{ibtimes3,nowhereelse4}.
\item[April 21, 2014] Samsung begins its arguments in Apple patent infringement suit (legal discussions ongoing for nearly two years)~\cite{cnet5}. Apple releases video with Tim Cook for Earth Day~\cite{appleinsider6}. Taiwanese site \emph{Commercial Times} discusses battery supplier issues with slimmer iPhone~\cite{ctee7,cnet6,businessinsider8}.
\item[April 22, 2014] Apple releases green logo and employee shirts for Earth Day~\cite{appleinsider9}.
\item[April 23, 2014] Apple reports fiscal second quarter results, announces stock split, live streams analyst call~\cite{apple10}.
\item[Late April--early May] Round-up of various rumors on iPhone design (e.g.,~\cite{gottabemobile11}), and on May 2, rumors said to ``[turn] from a trickle to a torrent''~\cite{ibtimes12}.  
\item[April 30, 2014] iPhone mockup photos leaked to Italian site \emph{Macitynet}~\cite{macitynet13}.
\item[May 2, 2014] Jury delivers verdict in Apple--Samsung patent case~\cite{wsj14}. \emph{Macitynet} posts photos of new iPhone~\cite{macitynet15}; additional photos posted on May 4~\cite{gottabemobile17}. 
\item[May 5, 2014] Leak on \emph{9to5Mac} of iPhone upgrade event to clear out stock, said to begin May 8, and subsequently announced~\cite{9to5mac18}.
\item[May 7, 2014] Taiwanese site \emph{Commercial Times} reports rumored information about iPhone supplier orders~\cite{ctee18},
then reported by \emph{CNET} on May 9~\cite{cnet19}.
\item[May 8, 2014] iPhone upgrade promotion email sent to customers~\cite{9to5mac20}.
\item[May 10, 2014] iPhone schematics leaked to \emph{Wei Feng}~\cite{feng21}.
\item[May 12, 2014] iPhone release date rumors reported by \emph{Tech Times}~\cite{techtimes22}.
\item[May 14, 2014] Leak on \emph{9to5Mac} of iPhone screen resolution~\cite{9to5mac23}.
\item[May 28, 2014] Apple confirms agreement to acquire the firm Beats by Dre~\cite{apple24}.
\end{description}

\section{Additional references}
\vspace{-2\baselineskip}%
\putbib[motif]
\end{bibunit}
\end{document}

%% file: Figures/Walks2.tex
\newcommand{\cadre}[1]{
	\draw {(-90: #1) node {}};%
	\draw {( 90: #1) node {}};%
	\draw {(   0: #1) node {}};%
	\draw {(   0:-#1) node {}};}%
\newcommand{\octagon}[4]{\!\raisebox{-#1 mm}{
	\begin{tikzpicture}[thick,scale=#2,color=#4]%
		\tikzstyle{pop}=[circle, draw=#4, fill=#4,
			inner sep=0pt, minimum width=#3 pt]%
		\foreach \j in {0,...,7}
			\draw {(0+\j*360/8:#1) node[pop] {}
				-- (0+\j*360/8+360/8:#1) node[pop] {}};%
		\cadre{#1}%
	\end{tikzpicture}}}%